\pdfoutput=1
\documentclass[a4paper,UKenglish, cleveref, autoref, thm-restate]{lipics-v2021}

\title{Proofs and Refutations for Intuitionistic and Second-Order Logic (Extended Version)}

\author{Pablo Barenbaum}{Universidad de Buenos Aires, Argentina \and Univeridad Nacional de Quilmes (CONICET), Argentina}{pbarenbaum@dc.uba.ar.org}{}{}

\author{Teodoro Freund}{Universidad de Buenos Aires, Argentina}{tfreund95@gmail.com}{}{}

\authorrunning{P. Barenbaum and T. Freund} 

\Copyright{Pablo Barenbaum and Teodoro Freund} 

\acknowledgements{This work was partially supported by project grant PUNQ 2247/22.}



\category{} 

\relatedversion{} 

\nolinenumbers 
\hideLIPIcs

\usepackage{xcolor}
\usepackage{xspace}
\usepackage{amsmath,amsthm,amssymb}
\usepackage{bm}
\usepackage{hyperref}
\usepackage{wrapfig}
\usepackage{mdframed}

\newdimen\proofrulebreadth \proofrulebreadth=.05em
\newdimen\proofdotseparation \proofdotseparation=1.25ex
\newdimen\proofrulebaseline \proofrulebaseline=2ex
\newcount\proofdotnumber \proofdotnumber=3
\let\then\relax
\def\hfi{\hskip0pt plus.0001fil}
\mathchardef\squigto="3A3B
\newif\ifinsideprooftree\insideprooftreefalse
\newif\ifonleftofproofrule\onleftofproofrulefalse
\newif\ifproofdots\proofdotsfalse
\newif\ifdoubleproof\doubleprooffalse
\let\wereinproofbit\relax
\newdimen\shortenproofleft
\newdimen\shortenproofright
\newdimen\proofbelowshift
\newbox\proofabove
\newbox\proofbelow
\newbox\proofrulename
\def\shiftproofbelow{\let\next\relax\afterassignment\setshiftproofbelow\dimen0 }
\def\shiftproofbelowneg{\def\next{\multiply\dimen0 by-1 }%
\afterassignment\setshiftproofbelow\dimen0 }
\def\setshiftproofbelow{\next\proofbelowshift=\dimen0 }
\def\setproofrulebreadth{\proofrulebreadth}

\def\prooftree{
\ifnum  \lastpenalty=1
\then   \unpenalty
\else   \onleftofproofrulefalse
\fi
\ifonleftofproofrule
\else   \ifinsideprooftree
        \then   \hskip.5em plus1fil
        \fi
\fi
\bgroup
\setbox\proofbelow=\hbox{}\setbox\proofrulename=\hbox{}%
\let\justifies\proofover\let\leadsto\proofoverdots\let\Justifies\proofoverdbl
\let\using\proofusing\let\[\prooftree
\ifinsideprooftree\let\]\endprooftree\fi
\proofdotsfalse\doubleprooffalse
\let\thickness\setproofrulebreadth
\let\shiftright\shiftproofbelow \let\shift\shiftproofbelow
\let\shiftleft\shiftproofbelowneg
\let\ifwasinsideprooftree\ifinsideprooftree
\insideprooftreetrue
\setbox\proofabove=\hbox\bgroup$\displaystyle 
\let\wereinproofbit\prooftree
\shortenproofleft=0pt \shortenproofright=0pt \proofbelowshift=0pt
\onleftofproofruletrue\penalty1
}

\def\eproofbit{
\ifx    \wereinproofbit\prooftree
\then   \ifcase \lastpenalty
        \then   \shortenproofright=0pt  
        \or     \unpenalty\hfil         
        \or     \unpenalty\unskip       
        \else   \shortenproofright=0pt  
        \fi
\fi
\global\dimen0=\shortenproofleft
\global\dimen1=\shortenproofright
\global\dimen2=\proofrulebreadth
\global\dimen3=\proofbelowshift
\global\dimen4=\proofdotseparation
\global\count255=\proofdotnumber
$\egroup  
\shortenproofleft=\dimen0
\shortenproofright=\dimen1
\proofrulebreadth=\dimen2
\proofbelowshift=\dimen3
\proofdotseparation=\dimen4
\proofdotnumber=\count255
}

\def\proofover{
\eproofbit 
\setbox\proofbelow=\hbox\bgroup 
\let\wereinproofbit\proofover
$\displaystyle
}%
\def\proofoverdbl{
\eproofbit 
\doubleprooftrue
\setbox\proofbelow=\hbox\bgroup 
\let\wereinproofbit\proofoverdbl
$\displaystyle
}%
\def\proofoverdots{
\eproofbit 
\proofdotstrue
\setbox\proofbelow=\hbox\bgroup 
\let\wereinproofbit\proofoverdots
$\displaystyle
}%
\def\proofusing{
\eproofbit 
\setbox\proofrulename=\hbox\bgroup 
\let\wereinproofbit\proofusing
\kern0.3em$
}

\def\endprooftree{
\eproofbit 
  \dimen5 =0pt
\dimen0=\wd\proofabove \advance\dimen0-\shortenproofleft
\advance\dimen0-\shortenproofright
\dimen1=.5\dimen0 \advance\dimen1-.5\wd\proofbelow
\dimen4=\dimen1
\advance\dimen1\proofbelowshift \advance\dimen4-\proofbelowshift
\ifdim  \dimen1<0pt
\then   \advance\shortenproofleft\dimen1
        \advance\dimen0-\dimen1
        \dimen1=0pt
        \ifdim  \shortenproofleft<0pt
        \then   \setbox\proofabove=\hbox{%
                        \kern-\shortenproofleft\unhbox\proofabove}%
                \shortenproofleft=0pt
        \fi
\fi
\ifdim  \dimen4<0pt
\then   \advance\shortenproofright\dimen4
        \advance\dimen0-\dimen4
        \dimen4=0pt
\fi
\ifdim  \shortenproofright<\wd\proofrulename
\then   \shortenproofright=\wd\proofrulename
\fi
\dimen2=\shortenproofleft \advance\dimen2 by\dimen1
\dimen3=\shortenproofright\advance\dimen3 by\dimen4
\ifproofdots
\then
        \dimen6=\shortenproofleft \advance\dimen6 .5\dimen0
        \setbox1=\vbox to\proofdotseparation{\vss\hbox{$\cdot$}\vss}%
        \setbox0=\hbox{%
                \advance\dimen6-.5\wd1
                \kern\dimen6
                $\vcenter to\proofdotnumber\proofdotseparation
                        {\leaders\box1\vfill}$%
                \unhbox\proofrulename}%
\else   \dimen6=\fontdimen22\the\textfont2 
        \dimen7=\dimen6
        \advance\dimen6by.5\proofrulebreadth
        \advance\dimen7by-.5\proofrulebreadth
        \setbox0=\hbox{%
                \kern\shortenproofleft
                \ifdoubleproof
                \then   \hbox to\dimen0{%
                        $\mathsurround0pt\mathord=\mkern-6mu%
                        \cleaders\hbox{$\mkern-2mu=\mkern-2mu$}\hfill
                        \mkern-6mu\mathord=$}%
                \else   \vrule height\dimen6 depth-\dimen7 width\dimen0
                \fi
                \unhbox\proofrulename}%
        \ht0=\dimen6 \dp0=-\dimen7
\fi
\let\doll\relax
\ifwasinsideprooftree
\then   \let\VBOX\vbox
\else   \ifmmode\else$\let\doll=$\fi
        \let\VBOX\vcenter
\fi
\VBOX   {\baselineskip\proofrulebaseline \lineskip.2ex
        \expandafter\lineskiplimit\ifproofdots0ex\else-0.6ex\fi
        \hbox   spread\dimen5   {\hfi\unhbox\proofabove\hfi}%
        \hbox{\box0}%
        \hbox   {\kern\dimen2 \box\proofbelow}}\doll%
\global\dimen2=\dimen2
\global\dimen3=\dimen3
\egroup 
\ifonleftofproofrule
\then   \shortenproofleft=\dimen2
\fi
\shortenproofright=\dimen3
\onleftofproofrulefalse
\ifinsideprooftree
\then   \hskip.5em plus 1fil \penalty2
\fi
}

\newcommand{\emptyPremise}{\vphantom{{}^@}}
\newcommand{\indrulename}[1]{\textsc{#1}}
\newcommand{\indrule}[3]{
\ensuremath{
\begin{array}{c}
  \prooftree #2
    \justifies #3
    \thickness=0.05em
    \using \indrulename{\footnotesize{#1}}
  \endprooftree
\end{array}}}
\newcommand{\indih}[2][\text{\ih}]{\indrule{}{#1}{#2}}

\newcommand{\indruleNPos}[4]{
\begin{array}[#1]{c@{}r}
\hspace{-.2cm}
 #3
\hspace{-.2cm}
\vspace{-.1cm}
\\
& \,#2\!\hspace{-.5cm}\vspace{-.2cm} \\
\cline{1-1}\vspace{-.3cm} \\
  #4 \hspace{.5cm}\,
\end{array}
}
\newcommand{\indruleN}[3]{{\small\indruleNPos{b}{#1}{#2}{#3}}}

\newcommand{\indNih}[2][\text{\ih}]{\indruleN{}{#1}{#2}}
\newcommand{\indNax}[2][\text{\Ax}]{\indruleN{#1}{}{#2}}

\newcommand{\deriv}{\pi}
\newcommand{\derivdots}[1]{
\begin{array}[b]{c@{}r}
\vdots 
\\
#1
\end{array}
}

\renewcommand{\theenumi}{\arabic{enumi}}
\renewcommand{\theenumii}{\arabic{enumii}}
\renewcommand{\theenumiii}{\arabic{enumiii}}
\renewcommand{\theenumiv}{\arabic{enumiv}}
\makeatletter
\renewcommand\p@enumii{\theenumi.}
\renewcommand\p@enumiii{\theenumi.\theenumii.}
\renewcommand\p@enumiv{\theenumi.\theenumii.\theenumiii.}
\makeatother

\newcommand{\llem}[1]{\label{lemma:#1}}
\newcommand{\rlem}[1]{Lem.~\ref{lemma:#1}}

\newcommand{\lprop}[1]{\label{prop:#1}}
\newcommand{\rprop}[1]{Prop.~\ref{prop:#1}}
\newcommand{\lthm}[1]{\label{thm:#1}}
\newcommand{\rthm}[1]{Thm.~\ref{thm:#1}}

\newcommand{\lsec}[1]{\label{section:#1}}
\newcommand{\rsec}[1]{Section~\ref{section:#1}}

\renewcommand{\emptyset}{\varnothing}

\newcommand{\Nat}{\mathbb{N}}

\newcommand{\NotNow}[1]{}

\newcommand{\Hs}{\hspace{.3cm}}
\newcommand{\HS}{\hspace{.5cm}}

\newcommand{\ST}{\ |\ }
\newcommand{\ie}{{\em i.e.}\xspace}
\newcommand{\eg}{{\em e.g.}\xspace}

\newcommand{\ih}{IH\xspace}
\newcommand{\set}[1]{\{#1\}}
\newcommand{\eqdef}{\,\mathrel{\overset{\mathrm{def}}{=}}\,}

\newcommand{\sub}[2]{\{#1:=#2\}}
\newcommand{\esub}[2]{[#1:=#2]}
\newcommand{\fv}[1]{\mathsf{fv}(#1)}
\newcommand{\ftv}[1]{\mathsf{ftv}(#1)}

\newcommand{\under}{\underline{\,\,\,}}

\newcommand{\imp}{\rightarrow}
\newcommand{\coimp}{\ltimes}

\newcommand{\var}{x}
\newcommand{\vartwo}{y}
\newcommand{\varthree}{z}

\newcommand{\varset}{X}

\newcommand{\tm}{t}
\newcommand{\tmtwo}{s}
\newcommand{\tmthree}{u}
\newcommand{\tmfour}{r}
\newcommand{\tmfive}{p}
\newcommand{\tmsix}{q}

\newcommand{\btyp}{\alpha}
\newcommand{\btyptwo}{\beta}
\newcommand{\btypthree}{\gamma}

\newcommand{\typ}{A}
\newcommand{\typtwo}{B}
\newcommand{\typthree}{C}

\newcommand{\ev}{P}
\newcommand{\evtwo}{Q}

\newcommand{\PP}{{}^\oplus}
\newcommand{\NN}{{}^\ominus}
\newcommand{\pp}{{}^+}
\newcommand{\nn}{{}^-}
\newcommand{\OP}{{}^{\sim}}

\newcommand{\dom}[1]{\mathsf{dom}(#1)}

\newcommand{\tctx}{\Gamma}
\newcommand{\tctxtwo}{\Delta}

\newcommand{\rulename}[1]{\textup{\textsc{#1}}}

\newcommand{\PRK}{\textup{PRK}}

\newcommand{\PRJ}{\textup{PRJ}}

\newcommand{\PRJVsym}{\ensuremath{\mathtt{J}}}
\newcommand{\PRJV}{\textup{PRJ$\star$}}
\newcommand{\NK}{\textup{NK}}
\newcommand{\NJ}{\textup{NJ}}
\newcommand{\lambdaPRK}{\lambda^{\PRK}}
\newcommand{\lambdaPRKU}{\lambda^{\PRK}_{\UTerms}}
\newcommand{\lambdaPRKFragment}[1]{\lambdaPRK(#1)}
\newcommand{\lambdaPRKAndImpNegAll}[1]{\lambdaPRKFragment{\fragAndImpNegAll}}

\newcommand{\lambdaPRJ}{\lambda^{\PRJ}}
\newcommand{\lambdaPRJV}{\lambda^{\PRJV}}

\newcommand{\classem}[1]{\iota(#1)}
\newcommand{\intem}[1]{\iota(#1)}
\newcommand{\judg}[3]{#1\vdash#2:#3}
\newcommand{\judgV}[4][\varset]{#2\vdash^{#1}#3:#4}
\newcommand{\judgPRK}[3]{#1\vdash_{\PRK}#2:#3}

\newcommand{\judgPRJ}[4][]{#2\vdash^{#1}_{\PRJ}#3:#4}

\newcommand{\judgPRJV}[4][\varset]{#2\vdash^{#1}_{\PRJV}#3:#4}
\newcommand{\logPRK}[2]{#1\vdash_{\PRK}#2}
\newcommand{\logPRJ}[2]{#1\vdash_{\PRJ}#2}

\renewcommand{\log}[2]{#1\vdash#2}
\newcommand{\logNK}[2]{#1\vdash_{\mathsf{NK}}#2}
\newcommand{\logNJ}[2]{#1\vdash_{\mathsf{NJ}}#2}

\newcommand{\strongabssym}[1]{{\RHD\!\!\!\LHD_{#1}}}
\newcommand{\strongabs}[3]{#2 \mathrel{\strongabssym{#1}} #3}

\newcommand{\abs}[3]{#2 \mathrel{\bowtie_{#1}} #3}

\newcommand{\pair}[2]{\langle#1,#2\rangle}
\newcommand{\pairp}[2]{\pair{#1}{#2}\pp}
\newcommand{\pairn}[2]{\pair{#1}{#2}\nn}
\newcommand{\pairpn}[2]{\pair{#1}{#2}^{\pm}}

\newcommand{\pairc}[2]{\pair{#1}{#2}^\clasym}
\newcommand{\pairj}[2]{\pair{#1}{#2}^\intsym}

\newcommand{\projisym}[1][i]{\pi_{#1}}
\newcommand{\proji}[2][i]{\projisym[#1](#2)}
\newcommand{\projip}[2][i]{\projisym[#1]^+(#2)}
\newcommand{\projin}[2][i]{\projisym[#1]^-(#2)}
\newcommand{\projipn}[2][i]{\projisym[#1]^{\pm}(#2)}
\newcommand{\projic}[2][i]{\projisym[#1]^{\clasym}(#2)}
\newcommand{\projij}[2][i]{\projisym[#1]^{\intsym}(#2)}

\newcommand{\inisym}[1][i]{\mathsf{in}_{#1}}
\newcommand{\ini}[2][i]{\inisym[#1](#2)}
\newcommand{\inip}[2][i]{\inisym[#1]\!\!\pp(#2)}
\newcommand{\inin}[2][i]{\inisym[#1]\!\!\nn(#2)}
\newcommand{\inipn}[2][i]{\inisym[#1]^{\pm}(#2)}

\newcommand{\inic}[2][i]{\inisym[#1]^{\clasym}(#2)}
\newcommand{\inij}[2][i]{\inisym[#1]^{\intsym}(#2)}

\newcommand{\casesym}{\delta}
\newcommand{\caseto}{.}
\newcommand{\case}[5]{\casesym#1\,[_{#2}\caseto#3][_{#4}\caseto#5]}
\newcommand{\casep}[5]{\casesym\pp#1\,[_{#2} \caseto #3][_{#4} \caseto #5]}
\newcommand{\caseptable}[5]{\begin{array}[t]{l}
  \casesym\pp #1 \\
    \HS[_{#2} \caseto #3] \\
    \HS[_{#4} \caseto #5] \\
  \end{array}
}

\newcommand{\casen}[5]{\casesym\nn#1\,[_{#2} \caseto #3][_{#4} \caseto #5]}
\newcommand{\casepn}[5]{\casesym^{\pm}#1\,[_{#2} \caseto #3][_{#4} \caseto #5]}
\newcommand{\casec}[5]{\casesym^{\clasym}#1\,[_{#2} \caseto #3][_{#4} \caseto #5]}
\newcommand{\casej}[5]{\casesym^{\intsym}#1\,[_{#2} \caseto #3][_{#4} \caseto #5]}

\newcommand{\lam}[2]{\lambda_{#1}.\,#2}
\newcommand{\lamp}[2]{\lambda\pp_{#1}.{#2}}
\newcommand{\lamn}[2]{\lambda\nn_{#1}.{#2}}
\newcommand{\lampn}[2]{\lambda^{\pm}_{#1}.{#2}}

\newcommand{\lamj}[2]{\lambda^{\intsym}_{#1}.{#2}}
\newcommand{\ap}[2]{#1@#2}
\newcommand{\app}[2]{#1@\pp#2}
\newcommand{\apn}[2]{#1@\nn#2}
\newcommand{\appn}[2]{#1@^{\pm}#2}
\newcommand{\appj}[2]{#1@^{\intsym}#2}

\newcommand{\copairwrap}[1]{(#1)}
\newcommand{\copair}[2]{\copairwrap{#1\,\bm{;}#2}}
\newcommand{\copairp}[2]{\copairwrap{#1\,\bm{;}\!\pp#2}}
\newcommand{\copairn}[2]{\copairwrap{#1\,\bm{;}\!\nn#2}}
\newcommand{\copairpn}[2]{\copairwrap{#1\,\bm{;}\!^{\pm}#2}}

\newcommand{\copairc}[2]{\copairwrap{#1\,\bm{;}^{\clasym}#2}}
\newcommand{\copairj}[2]{\copairwrap{#1\,\bm{;}^{\intsym}#2}}
\newcommand{\colamsym}{\varrho}
\newcommand{\colam}[4]{\colamsym{#1}[_{#2;#3}.{#4}]}
\newcommand{\colamp}[4]{\colamsym^+{#1}[_{#2;#3}.{#4}]}
\newcommand{\colamn}[4]{\colamsym^-{#1}[_{#2;#3}.{#4}]}
\newcommand{\colampn}[4]{\colamsym^{\pm}{#1}[_{#2;#3}.{#4}]}
\newcommand{\colamc}[4]{\colamsym^{\clasym}{#1}[_{#2;#3}.{#4}]}

\newcommand{\negisym}{\mathsf{N}}
\newcommand{\negi}[1]{\negisym#1}
\newcommand{\negip}[1]{\negisym\pp#1}
\newcommand{\negin}[1]{\negisym\nn#1}
\newcommand{\negipn}[1]{\negisym^{\pm}#1}

\newcommand{\negiP}[1]{\negisym\PP#1}
\newcommand{\negeP}[1]{\negesym\PP#1}
\newcommand{\negesym}{\mathsf{M}}
\newcommand{\nege}[1]{\negesym#1}
\newcommand{\negep}[1]{\negesym\pp#1}
\newcommand{\negen}[1]{\negesym\nn#1}
\newcommand{\negepn}[1]{\negesym^{\pm}#1}
\newcommand{\neglamsym}{\Lambda}
\newcommand{\negapsym}{\texttt{\textup{\#}}}
\newcommand{\neglamc}[2]{\neglamsym^{\clasym}_{#1}.\,#2}
\newcommand{\negapc}[2]{#1 \negapsym^{\clasym} #2}
\newcommand{\neglamj}[2]{\neglamsym^{\intsym}_{#1}.\,#2}
\newcommand{\negapj}[2]{#1 \negapsym^{\intsym} #2}

\newcommand{\clasym}{\mathcal{C}}
\newcommand{\clasapsym}{\bullet}
\newcommand{\claslamsym}{\textup{\Circle}}
\newcommand{\claslam}[2]{\mathsf{\claslamsym}_{#1}.\,#2}
\newcommand{\claslamp}[2]{\mathsf{\claslamsym}\pp_{#1}.\,#2}
\newcommand{\claslamn}[2]{\mathsf{\claslamsym}\nn_{#1}.\,#2}
\newcommand{\claslampn}[2]{\mathsf{\claslamsym}^{\pm}_{#1}.\,#2}
\newcommand{\claslamnp}[2]{\mathsf{\claslamsym}^{\mp}_{#1}.\,#2}

\newcommand{\clasap}[2]{#1 \clasapsym #2}
\newcommand{\clasapp}[2]{#1 \clasapsym\!\!\pp\, #2}
\newcommand{\clasapn}[2]{#1 \clasapsym\!\!\nn\, #2}
\newcommand{\clasappn}[2]{#1 \clasapsym\!\!^{\pm}\, #2}

\newcommand{\intsym}{\clasym}

\colorlet{darkorange}{orange!50!black}
\colorlet{darkgreen}{green!60!black}

\newcommand{\colorand}[1]{#1}
\newcommand{\colorimp}[1]{#1}
\newcommand{\colorneg}[1]{#1}
\newcommand{\colorall}[1]{#1}

\newcommand{\fragand}{\colorand{{}^\land_\lor}}
\newcommand{\fragimp}{\colorimp{{}^\imp_\coimp}}
\newcommand{\fragneg}{\colorneg{\neg}}
\newcommand{\fragall}{\colorall{{}^\forall_\exists}}

\newcommand{\fragAndImpNegAll}{\fragand\fragimp\fragneg\fragall}

\newcommand{\Ax}{\rulename{Ax}}
\newcommand{\Abs}{\rulename{Abs}}

\newcommand{\Iandp}{\rulename{I$_\land^+$}}
\newcommand{\Iorn}{\rulename{I$_\lor^-$}}

\newcommand{\Eandp}[1][i]{\rulename{E$^+_{\land{#1}}$}}
\newcommand{\Eorn}[1][i]{\rulename{E$^-_{\lor{#1}}$}}

\newcommand{\Iorp}[1][i]{\rulename{I$^+_{\lor{#1}}$}}
\newcommand{\Iandn}[1][i]{\rulename{I$^-_{\land{#1}}$}}

\newcommand{\Eorp}{\rulename{E$_\lor^+$}}
\newcommand{\Eandn}{\rulename{E$_\land^-$}}

\newcommand{\Inotp}{\rulename{I$_\lnot^+$}}
\newcommand{\Inotn}{\rulename{I$_\lnot^-$}}

\newcommand{\Enotp}{\rulename{E$_\lnot^+$}}
\newcommand{\Enotn}{\rulename{E$_\lnot^-$}}

\newcommand{\Icp}{\rulename{I$_{\circ}^+$}}
\newcommand{\Icn}{\rulename{I$_{\circ}^-$}}
\newcommand{\Icpn}{\rulename{I$_{\circ}^\pm$}}

\newcommand{\Ecp}{\rulename{E$_{\circ}^+$}}
\newcommand{\Ecn}{\rulename{E$_{\circ}^-$}}
\newcommand{\Ecpn}{\rulename{E$_{\circ}^\pm$}}

\newcommand{\Iallp}{\rulename{I$_\forall^+$}}
\newcommand{\Iexn}{\rulename{I$_\exists^-$}}

\newcommand{\Eallp}{\rulename{E$_\forall^+$}}
\newcommand{\Eexn}{\rulename{E$_\exists^-$}}

\newcommand{\Iexp}{\rulename{I$_\exists^+$}}
\newcommand{\Ialln}{\rulename{I$_\forall^-$}}

\newcommand{\Eexp}{\rulename{E$_\exists^+$}}
\newcommand{\Ealln}{\rulename{E$_\forall^-$}}

\newcommand{\Iimpp}{\rulename{I$_\imp^+$}}
\newcommand{\Icoimpn}{\rulename{I$_\coimp^-$}}

\newcommand{\Eimpp}{\rulename{E$_\imp^+$}}
\newcommand{\Ecoimpn}{\rulename{E$_\coimp^-$}}

\newcommand{\Icoimpp}{\rulename{I$_\coimp^+$}}
\newcommand{\Iimpn}{\rulename{I$_\imp^-$}}

\newcommand{\Ecoimpp}{\rulename{E$_\coimp^+$}}
\newcommand{\Eimpn}{\rulename{E$_\imp^-$}}

\newcommand{\Weakening}{\rulename{W}}
\newcommand{\ICut}{\rulename{ICut}}
\newcommand{\CCut}{\rulename{CCut}}
\newcommand{\IContrapose}{\rulename{IContra}}
\newcommand{\CContrapose}{\rulename{CContra}}
\newcommand{\ccontrapose}[3]{{\mathtt{cc}_{#1}^{#2}(#3)}}
\newcommand{\icontrapose}[3]{{\mathtt{ic}_{#1}^{#2}(#3)}}

\newcommand{\lemN}[1]{\pitchfork^-_{#1}}
\newcommand{\lemP}[1]{\pitchfork^+_{#1}}

\newcommand{\lemNinner}[2]{\Delta^-_{#1,#2}}
\newcommand{\lemPinner}[2]{\Delta^+_{#1,#2}}

\newcommand{\SubT}{\rulename{\textup{SubT}}}

\newcommand{\NDrulename}[1]{\rulename{{#1}}}

\newcommand{\NDAx}{\NDrulename{ax}}
\newcommand{\NDExpl}{\NDrulename{E$\bot$}}
\newcommand{\NDLem}{\NDrulename{lem}}

\newcommand{\NDIand}{\NDrulename{I$\land$}}
\newcommand{\NDEand}[1][i]{\NDrulename{E$\land_{#1}$}}

\newcommand{\NDIor}[1][i]{\NDrulename{I$\lor_{#1}$}}
\newcommand{\NDEor}{\NDrulename{E$\lor$}}

\newcommand{\NDInot}{\NDrulename{I$\lnot$}}
\newcommand{\NDEnot}{\NDrulename{E$\lnot$}}

\newcommand{\NDIall}{\NDrulename{I$\forall$}}
\newcommand{\NDEall}{\NDrulename{E$\forall$}}

\newcommand{\NDIex}{\NDrulename{I$\exists$}}
\newcommand{\NDEex}{\NDrulename{E$\exists$}}

\newcommand{\NDIimp}{\NDrulename{I$\imp$}}
\newcommand{\NDEimp}{\NDrulename{E$\imp$}}

\newcommand{\NDIcoimp}{\NDrulename{I$\coimp$}}
\newcommand{\NDEcoimp}{\NDrulename{E$\coimp$}}

\newcommand{\NDW}{\NDrulename{W}}
\newcommand{\NDCut}{\NDrulename{Cut}}

\newcommand{\rewritingRuleName}[1]{\mathsf{#1}}

\newcommand{\toa}[1]{\xrightarrow{#1}}

\newcommand{\rtoa}[1]{\mathrel{\xrightarrow{#1}\!\!{}^*\,\,}}

\newcommand{\rto}{\mathrel{\rightarrow^*}}

\newcommand{\ruleBAndP}{\rewritingRuleName{\beta\pp_\land}}
\newcommand{\ruleBOrN}{\rewritingRuleName{\beta\nn_\lor}}
\newcommand{\ruleBOrP}{\rewritingRuleName{\beta\pp_\lor}}
\newcommand{\ruleBAndN}{\rewritingRuleName{\beta\nn_\land}}
\newcommand{\ruleBImpP}{\rewritingRuleName{\beta\pp_\imp}}
\newcommand{\ruleBCoimpN}{\rewritingRuleName{\beta\nn_\coimp}}
\newcommand{\ruleBCoimpP}{\rewritingRuleName{\beta\pp_\coimp}}
\newcommand{\ruleBImpN}{\rewritingRuleName{\beta\nn_\imp}}
\newcommand{\ruleBNegP}{\rewritingRuleName{\beta\pp_\neg}}
\newcommand{\ruleBNegN}{\rewritingRuleName{\beta\nn_\neg}}
\newcommand{\ruleBAllP}{\rewritingRuleName{\beta\pp_\forall}}
\newcommand{\ruleBExN}{\rewritingRuleName{\beta\nn_\exists}}
\newcommand{\ruleBExP}{\rewritingRuleName{\beta\pp_\exists}}
\newcommand{\ruleBAllN}{\rewritingRuleName{\beta\nn_\forall}}
\newcommand{\ruleBClasP}{\rewritingRuleName{\beta\pp_{\circ}}}
\newcommand{\ruleBClasN}{\rewritingRuleName{\beta\nn_{\circ}}}
\newcommand{\ruleBClasPN}{\rewritingRuleName{\beta^{\pm}_{\circ}}}
\newcommand{\ruleAAnd}{\rewritingRuleName{\strongabssym{\land}}}
\newcommand{\ruleAOr}{\rewritingRuleName{\strongabssym{\lor}}}
\newcommand{\ruleAImp}{\rewritingRuleName{\strongabssym{\imp}}}
\newcommand{\ruleACoimp}{\rewritingRuleName{\strongabssym{\coimp}}}
\newcommand{\ruleANeg}{\rewritingRuleName{\strongabssym{\neg}}}
\newcommand{\ruleAAll}{\rewritingRuleName{\strongabssym{\forall}}}
\newcommand{\ruleAEx}{\rewritingRuleName{\strongabssym{\exists}}}

\newcommand{\SEP}{\,/\,}
\newcommand{\ruleProj}{\ruleBAndP\SEP\ruleBOrN}
\newcommand{\ruleCase}{\ruleBOrP\SEP\ruleBAndN}
\newcommand{\ruleBetaB}{\ruleBImpP\SEP\ruleBCoimpN}
\newcommand{\ruleCoproj}{\ruleBCoimpP\SEP\ruleBImpN}
\newcommand{\ruleNeg}{\ruleBNegP\SEP\ruleBNegN}
\newcommand{\ruleAppT}{\ruleBAllP\SEP\ruleBExN}
\newcommand{\ruleOpen}{\ruleBExP\SEP\ruleBAllN}
\newcommand{\ruleBetaC}{\ruleBClasP\SEP\ruleBClasN}
\newcommand{\ruleAbsPairInj}{\ruleAAnd}
\newcommand{\ruleAbsInjPair}{\ruleAOr}
\newcommand{\ruleAbsLamCopair}{\ruleAImp}
\newcommand{\ruleAbsCopairLam}{\ruleACoimp}
\newcommand{\ruleAbsNeg}{\ruleANeg}
\newcommand{\ruleAbsLamPairT}{\ruleAAll}
\newcommand{\ruleAbsPairLamT}{\ruleAEx}
\newcommand{\ruleEta}{\rewritingRuleName{\eta_{\circ}}}

\newcommand{\ctxhole}{\Box}
\newcommand{\ctxof}[1]{\langle#1\rangle}
\newcommand{\gctx}{\mathtt{C}}

\newcommand{\gctxof}[1]{\gctx\ctxof{#1}}

\newcommand{\crictx}{\mathtt{H}}
\newcommand{\crictxof}[1]{\crictx\ctxof{#1}}

\newcommand{\nf}{f}
\newcommand{\neu}{e}

\newcommand{\semF}[1]{\{\!\!\{#1\}\!\!\}}

\newcommand{\typingDerivation}{\pi}

\newcommand{\all}[2]{\forall #1.\,#2}
\newcommand{\ex}[2]{\exists #1.\,#2}

\newcommand{\dummy}{\Diamond}

\newcommand{\lamtu}[1]{\lambda_{\dummy}.\,#1}
\newcommand{\lamtp}[2]{\lambda\pp_{#1}.\,#2}
\newcommand{\lamtn}[2]{\lambda\nn_{#1}.\,#2}
\newcommand{\lamtpn}[2]{\lambda^{\pm}_{#1}.\,#2}

\newcommand{\lamc}[2]{\lambda^{\clasym}_{#1}.\,#2}
\newcommand{\lamtc}[2]{\lambda^{\clasym}_{#1}.\,#2}
\newcommand{\lamtj}[2]{\lambda^{\intsym}_{#1}.\,#2}

\newcommand{\apptu}[1]{#1 @ \dummy}
\newcommand{\apptp}[2]{#1 @\pp #2}
\newcommand{\apptn}[2]{#1 @\nn #2}
\newcommand{\apptpn}[2]{#1 @^{\pm} #2}
\newcommand{\appc}[2]{#1 \,\texttt{\textup{@}}^{\clasym}\, #2}
\newcommand{\apptc}[2]{#1 \,\texttt{\textup{@}}^{\clasym}\, #2}
\newcommand{\apptj}[2]{#1 \,\texttt{\textup{@}}^{\intsym}\, #2}

\newcommand{\patu}[1]{\langle\dummy, #1\rangle}
\newcommand{\patp}[2]{\langle#1, #2\rangle\pp}
\newcommand{\patn}[2]{\langle#1, #2\rangle\nn}
\newcommand{\patpn}[2]{\langle#1, #2\rangle^{\pm}}

\newcommand{\patc}[2]{\langle#1, #2\rangle^{\clasym}}

\newcommand{\patj}[2]{\langle#1, #2\rangle^{\intsym}}

\newcommand{\openp}[4]{\optp{#3}{#1}{#2}{#4}}

\newcommand{\optu}[3]{\nabla\,#1[_{(\dummy,#2)}.#3]}
\newcommand{\optp}[4]{\nabla\pp\,#1[_{(#2,#3)}.#4]}
\newcommand{\optn}[4]{\nabla\nn\,#1[_{(#2,#3)}.#4]}
\newcommand{\optpn}[4]{\nabla^{\pm}\,#1[_{(#2,#3)}.#4]}
\newcommand{\optc}[4]{\nabla^{\clasym}\,#1[_{(#2,#3)}.#4]}
\newcommand{\optj}[4]{\nabla^{\intsym}\,#1[_{(#2,#3)}.#4]}

\usepackage{wasysym}

\colorlet{colappendix}{blue}

\newcommand{\SeeAppendix}[1]{\textcolor{colappendix}{#1}}

\newcounter{alphasect}
\def\alphainsection{0}

\let\oldsection=\section
\def\section{%
  \ifnum\alphainsection=1%
    \addtocounter{alphasect}{1}%
  \fi%
\oldsection}%

\renewcommand\thesection{%
  \ifnum\alphainsection=1%
    \Alph{alphasect}%
  \else%
    \arabic{section}%
  \fi%
}%

\newcommand{\UTerms}{\mathbf{U}}
\newcommand{\erase}[1]{|#1|}
\newcommand{\atm}{a}
\newcommand{\atmtwo}{b}
\newcommand{\atmthree}{c}

\newcommand{\SNTerms}{\mathbf{SN}}
\newcommand{\CanTerms}{\mathbf{CAN}}
\newcommand{\tou}{\to_{\UTerms}}
\newcommand{\rtou}{\to^*_{\UTerms}}
\newcommand{\join}{\bigvee}
\newcommand{\meet}{\bigwedge}
\newcommand{\lfpF}[1]{\bm{\mu}(#1)}
\newcommand{\gfpF}[1]{\bm{\nu}(#1)}
\newcommand{\lfp}[2]{\lfpF{#1.#2}}
\newcommand{\gfp}[2]{\gfpF{#1.#2}}
\newcommand{\snsize}[1]{\#(#1)}

\newcommand{\red}[2]{[\![#1]\!]_{#2}}
\newcommand{\env}{\rho}
\newcommand{\subst}{\sigma}

\newcommand{\judgSubst}[3]{#1\vDash_{#2}#3}

\newcommand{\RCSet}{\mathbf{RC}}
\newcommand{\RCPerp}{\bm{\,\perp\!\!\!\perp}}
\newcommand{\RC}{r.c.\xspace}
\newcommand{\RCs}{r.c.'s\xspace}
\newcommand{\rclosure}[1]{\mathbb{C}{#1}}
\newcommand{\rc}{\xi}
\newcommand{\rcfamily}{\mathcal{R}}
\newcommand{\inI}{\in{I}}
\newcommand{\rctop}{\bm{\top}}
\newcommand{\rcbot}{\bm{\bot}}
\newcommand{\rctimes}{\bm{\times}}
\newcommand{\rcplus}{\bm{+}}
\newcommand{\rcimp}{\bm{\imp}}
\newcommand{\rccoimp}{\bm{\coimp}}
\newcommand{\rcneg}{\bm{\mathord{\sim}}}
\newcommand{\rcall}[2]{\bm{\Pi}_{#1}\,{#2}}
\newcommand{\rcex}[2]{\bm{\Sigma}_{#1}\,{#2}}
\newcommand{\rcperp}[2]{(#1,#2)\in\RCPerp}

\newcommand{\fragTitle}[1]{\textsf{\textup{#1}}}
\newcommand{\resultName}[1]{\textsf{\textup{#1}}}

\newenvironment{fragBox}[1]
  {\mdfsetup{
    frametitle={\colorbox{white}{\space\fragTitle{#1}\space}},
    frametitleaboveskip=-\ht\strutbox
    }
  \begin{mdframed}[linecolor=lipicsGray]\vspace{-.5cm}\small%
  }
  {\end{mdframed}}

\begin{document}

\maketitle

\begin{abstract}
The $\lambdaPRK$-calculus is a typed $\lambda$-calculus that exploits
the duality between the notions of proof and refutation to provide
a computational interpretation for classical propositional logic.
In this work, we extend $\lambdaPRK$ to encompass classical second-order logic,
by incorporating parametric polymorphism and existential types. The system is
shown to enjoy good computational properties, such as type preservation,
confluence, and strong normalization, which is established by means of a
reducibility argument. We identify a syntactic restriction on proofs that
characterizes exactly the intuitionistic fragment of second-order $\lambdaPRK$,
and we study canonicity results.
\end{abstract}

\section{Introduction}

Constructivism in logic is closely related with the notion of {\em algorithm}
in computer science.
The reason is that a constructive proof of existence of a mathematical object
fulfilling certain properties should provide an {\em effective construction}
of such an object.
For example, a constructive proof of
$\forall x \in \Nat.\,\exists y\in\Nat.\,P(x,y)$
may be understood as an algorithm that takes as input a natural number $x$
and produces as output a natural number $y$ that verifies $P(x,y)$.
The close relationship that exists
between {\em proofs} and {\em computer programs},
and between {\em logical propositions} and {\em program specifications}
(or {\em types}),
can be taken to its maximum consequences
in the form of the {\bf propositions-as-types correspondence}.

This correspondence has given rise to a broad and active area of research,
guided by the principle that
{\em each proof-theoretical notion has a
computational counterpart and vice-versa}.
The interest is that this correspondence
allows logic and computer science to feed back on each other.
Indeed, besides minimal propositional logic,
it has been extended to such settings as
{\em first-order logic}~\cite{debruijn1970mathematical,martinlof1971theory,CoquandH88},
{\em second-order logic}~\cite{thesisgirard,reynolds1974towards},
{\em linear logic}~\cite{girard1987linear},
{\em modal logic}~\cite{bierman2000intuitionistic,DBLP:journals/jacm/DaviesP01}
and {\em classical logic}~\cite{griffin1989formulae,Curien00theduality,symmetric-Barbanera-berardi,lambdamu-parigot}.

The question of what kind of computational system would
constitute a reasonable counterpart for {\bf classical logic},
from the point of view of the propositions-as-types correspondence,
is far from being definitely settled.
This work is part of the quest for a satisfactory answer to this
problem.

\subparagraph{The proofs and refutations calculus ($\lambdaPRK$)}
Until the late 1980s, it was widely thought that it was not possible
to extend the propositions-as-types correspondence to encompass classical
logic. This view changed when Griffin~\cite{griffin1989formulae}
remarked that the classical principle of double negation elimination
($\neg\neg\typ \to \typ$) can be understood as the typing rule for
a control operator $\mathcal{C}$, closely related to Felleisen's $\mathcal{C}$
operator~\cite{FelleisenFKD87} and to \textsc{Scheme}'s \texttt{call/cc}.
Since then, many other calculi for classical logic have been proposed.
Significant examples are Parigot's $\lambda\mu$~\cite{lambdamu-parigot},
Barbanera and Berardi's
symmetric $\lambda$-calculus~\cite{symmetric-Barbanera-berardi},
and
Curien and Herbelin's $\bar{\lambda}\mu\tilde{\mu}$ calculus~\cite{Curien00theduality}.

The starting point of this paper is the logical system $\PRK$,
introduced recently by the authors~\cite{BFlics21}
and extending Nelson's constructive negation~\cite{nelson1949constructible}.
In $\PRK$, propositions become classified along two dimensions:
their {\em sign}, which may be positive or negative,
and their {\em strength}, which may be strong or weak.
This results into four possible {\em modes} to state
\begin{wrapfigure}{r}{5cm}
$
  \begin{array}{l|ll}
                & \text{Positive} & \text{Negative} \\
  \hline
  \text{Strong} & \typ\pp         & \typ\nn \\
  \text{Weak}   & \typ\PP         & \typ\NN
  \vspace{-0.2cm}
  \end{array}
$
\end{wrapfigure} 
a proposition.
Positive ($\typ\pp/\typ\PP$) and negative ($\typ\nn/\typ\NN$)
propositions correspond to affirmations and
denials.
Strong ($\typ\pp/\typ\nn$)
and weak ($\typ\PP/\typ\NN$) propositions
impose restrictions on the shape of canonical proofs:
a canonical proof of a strong affirmation ($\typ\pp$)
must always be constructed with an introduction rule
for the corresponding logical connective,
whereas a canonical proof of a weak affirmation ($\typ\PP$)
must always proceed by {\em reductio ad absurdum},
by assuming the weak denial $\typ\NN$
and proving the strong affirmation $\typ\pp$.

We summarize some important characteristics of $\PRK$.
First,
  $\PRK$ is a {\bf refinement} of classical logic:
  $\typ_1,\hdots,\typ_n\vdash\typtwo$ holds in classical propositional logic
  if and only if
  $\typ\PP_1,\hdots,\typ\PP_n\vdash\typtwo\PP$ holds in $\PRK$.
  In fact $\PRK$ is ``finer'' than classical logic:
  for example,
  the law of excluded middle holds {\em weakly},
  \ie $(\typ\lor\neg\typ)\PP$ is valid in $\PRK$,
  whereas it does not hold {\em strongly},
  \ie $(\typ\lor\neg\typ)\pp$ is not valid (in general) in $\PRK$.
Second,
  the $\lambdaPRK$-calculus, which results
  from assigning proof terms to $\PRK$ proofs
  and endowing it with rewrite rules,
  turns out to be {\bf confluent} and {\bf strongly normalizing},
  besides enjoying {\bf subject reduction}.
Third,
  as a result, $\PRK$ enjoys {\bf canonicity}:
  a proof of a sequent $\vdash \ev$
  without assumptions can always be normalized to a
  {\em canonical} proof, headed by an introduction rule.

\subparagraph{Contributions and structure of this paper}
The $\PRK$ logical system of~\cite{BFlics21} only treats three propositional
connectives: conjunction, disjunction, and negation.
\begin{itemize}
\item
  In~\rsec{second_order_prk}, we {\bf extend the $\lambdaPRK$ calculus}
  to propositional second-order logic.
  We incorporate
  {\bf second-order universal and existential quantification},
  as well as two propositional connectives,
  implication and co-implication.
  The system is shown to {\bf refine classical second-order logic},
  and to enjoy good computational properties: {\bf subject reduction} and
  {\bf confluence}.
  This extension increases the expressivity of the system, allowing
  to encode inductive datatypes such as natural numbers, lists, and trees.
\item
  In~\rsec{bohm_berarducci_encodings},
  we study {\bf B\"ohm--Berarducci encodings}, that is,
  we study how the logical connectives of second-order $\lambdaPRK$
  may be encoded in terms of universal quantification and implication only
  ($\set{\forall,\imp}$).
  The encoding turns out to be only partially satisfactory:
  it simulates proof normalization for an introduction rule
  followed by an elimination rule in the {\em positive} case but,
  unfortunately, not in the negative case.
\item
  In \rsec{strong_normalization} we prove {\bf strong normalization}
  for the second-order $\lambdaPRK$-calculus.
  This is the most technically challenging part of the work.
  In~\cite{BFlics21}, normalization of the propositional
  fragment of $\lambdaPRK$ is attained by means of a translation
  to System~F with non-strictly positive recursion.
  This technique does not carry over
  to the second-order case.
  To prove strong normalization, we use a variant of Girard's technique
  of reducibility candidates and, in particular,
  we resort to a non-trivial adaptation
  of Mendler's proof of strong normalization
  for System~F with non-strictly positive recursion~\cite{mendler1991inductive}.
\item
  In~\rsec{prj},
  we define a subsystem of second-order $\lambdaPRK$,
  called $\lambdaPRJ$, by imposing a syntactic restriction on terms.
  We show that $\lambdaPRJ$
  {\bf refines second-order intuitionistic logic}, in the sense
  that $\lambdaPRJ$ is a conservative extension of second-order intuitionistic
  logic and, conversely,
  second-order intuitionistic logic can be embedded in $\lambdaPRJ$.
\item
  In~\rsec{canonicity} we formulate {\bf canonicity} results for $\lambdaPRK$.
  In particular, we strengthen the canonicity results of~\cite{BFlics21}
  to show that an explicit witness can be extracted from a proof of~$\ev$.
\item
  Finally, in~\rsec{conclusion} we conclude and we discuss
  some related and future work.
\end{itemize}

\section{Second-Order Proofs and Refutations}
\lsec{second_order_prk}

In this section we define a second-order extension of $\lambdaPRK$,
including its syntax, typing rules, and rewriting rules.
We show that
the system enjoys {\bf subject reduction}, it is {\bf confluent},
and it {\bf refines classical second-order logic}~(\rthm{lambdaPRK_refinement}).

\subparagraph{Syntax of types}
We assume given a denumerable set of
{\em type variables} $\btyp,\btyptwo,\btypthree,\hdots$.
The sets of {\em pure types} ($\typ,\typtwo,\hdots$)
and {\em types} ($\ev,\evtwo,\hdots$) are given by:
\[
  \typ ::=
       \btyp
  \mid \colorand{\typ \land \typ}
  \mid \colorand{\typ \lor \typ}
  \mid \colorimp{\typ \imp \typ}
  \mid \colorimp{\typ \coimp \typ}
  \mid \colorneg{\neg\typ}
  \mid \colorall{\all{\btyp}{\typ}}
  \mid \colorall{\ex{\btyp}{\typ}}
  \HS\HS
  \ev ::=
       \typ\pp
  \mid \typ\nn
  \mid \typ\PP
  \mid \typ\NN
\]
where $\typ\coimp\typtwo$ represents {\em co-implication},
the dual connective to implication,
to be understood (roughly) as $\neg\typ\land\typtwo$.
The four modes represent strong affirmation ($\typ\pp$),
strong denial ($\typ\nn$), weak affirmation ($\typ\PP$),
and weak denial ($\typ\NN$).
Note that {\em modes} ($\pp,\nn,\PP,\NN$) can only decorate the
root of a type, \ie they cannot be nested.

Sometimes one may be interested in {\em fragments} of the system.
For instance, the $\lambdaPRK$-calculus of~\cite{BFlics21}
corresponds to the $\set{\land,\lor,\neg}$ fragment.
In this paper we are usually interested in the full
$\set{\land,\lor,\imp,\coimp,\neg,\forall,\exists}$ fragment.
As long as there is little danger of confusion
we still speak of $\lambdaPRK$ without further qualifications.

\subparagraph{Syntax of terms}
Terms of $\lambdaPRK$ are given by the following grammar.
The letter $i$ ranges over $\set{1,2}$.
Some terms are decorated with either ``$\pp$'' or ``$\nn$''.
In the grammar we write ``${}^{\pm}$''
to stand for either ``$\pp$'' or ``$\nn$''.
\[
\begin{array}{rllllll}
  \tm,\tmtwo,\hdots ::= \hspace{-1cm}
    & \var^{\ev}
    & \text{variable}
    & \mid & \strongabs{\ev}{\tm}{\tmtwo}
    & \text{absurdity}
  \\
      \mid & \claslampn{(\var:\ev)}{\tm}
    & \text{$\PP/\NN$ introduction}
    & \mid & \clasappn{\tm}{\tmtwo}
    & \text{$\PP/\NN$ elimination}
  \\
      \mid & \pairpn{\tm}{\tmtwo}
    & \text{\colorand{$\land\pp$/$\lor\nn$ introduction}}
    & \mid & \projipn{\tm}
    & \text{\colorand{$\land\pp$/$\lor\nn$ elimination}}
  \\
      \mid & \inipn{\tm}
    & \text{\colorand{$\lor\pp$/$\land\nn$ introduction}}
    & \mid & \casepn{\tm}{\var:\ev}{\tmtwo}{\vartwo:\evtwo}{\tmthree}
    & \text{\colorand{$\lor\pp$/$\land\nn$ elimination}}
  \\
      \mid & \lampn{(\var:\ev)}{\tm}
    & \text{\colorimp{$\imp\pp$/$\coimp\nn$ introduction}}
    & \mid & \appn{\tm}{\tmtwo}
    & \text{\colorimp{$\imp\pp$/$\coimp\nn$ elimination}}
  \\
      \mid & \copairpn{\tm}{\tmtwo}
    & \text{\colorimp{$\coimp\pp$/$\imp\nn$ introduction}}
    & \mid & \colampn{\tm}{\var:\ev}{\vartwo:\evtwo}{\tmtwo}
    & \text{\colorimp{$\coimp\pp$/$\imp\nn$ elimination}}
  \\
      \mid & \negipn{\tm}
    & \text{\colorneg{$\neg\pp$/$\neg\nn$ introduction}}
    & \mid & \negepn{\tm}
    & \text{\colorneg{$\neg\pp$/$\neg\nn$ elimination}}
  \\
      \mid & \lamtpn{\btyp}{\tm}
    & \text{\colorall{$\forall\pp$/$\exists\nn$ introduction}}
    & \mid & \apptpn{\tm}{\typ}
    & \text{\colorall{$\forall\pp$/$\exists\nn$ elimination}}
  \\
      \mid & \patpn{\typ}{\tm}
    & \text{\colorall{$\exists\pp$/$\forall\nn$ introduction}}
    & \mid & \optpn{\tm}{\btyp}{\var:\ev}{\tmtwo}
    & \text{\colorall{$\exists\pp$/$\forall\nn$ elimination}}
\end{array}
\]
The notions of free and bound occurrences of variables are defined as
expected,
with the typographical convention that subscripted variable occurrences
are binding.
Terms are considered up to $\alpha$-renaming of bound variables.
We write $\fv{\tm}$ for the set of free variables of $\tm$
and $\ftv{\tm}$ for the set of type variables occurring free in $\tm$.
By $\tm\sub{\var}{\tmtwo}$ we mean the capture-avoiding substitution
of the free occurrences of $\var$ in $\tm$ by $\tmtwo$.

Variables are formally annotated with their type, which we usually omit.
Sometimes we also omit the types of bound variables if they are clear
from the context, as well as the name of unused bound variables,
writing ``$\under$'' instead. For example, if $\var \not\in \fv{\tm}$
we may write $\claslamp{\under}{\tm}$ rather than
$\claslamp{(\var:\typ\nn)}{\tm}$.
Application-like operators are assumed to be left-associative;
for example,
$\apptp{\clasapp{\app{\tm}{\tmtwo}}{\tmthree}}{\typ}$
stands for
$\apptp{(\clasapp{(\app{\tm}{\tmtwo})}{\tmthree})}{\typ}$.
In a term of the form $\claslampn{(\var:\ev)}{\tm}$,
the variable $\var$ is called the {\em counterfactual},
and more specifically a {\em negative counterfactual}
in a term of the form $\claslamp{(\var:\typ\NN)}{\tm}$.
In a term of the form $\clasappn{\tm}{\tmtwo}$,
we call $\tm$ the {\em subject} and $\tmtwo$ the {\em argument}.
We write $\gctx$ for arbitrary term {\em contexts},
\ie terms with a single free occurrence
of a distinguished variable $\ctxhole$ called a {\em hole}.
We write $\gctxof{\tm}$ for the variable-capturing substitution of the
hole of $\gctx$ by $\tm$.

\subparagraph{The $\lambdaPRK$ type system}
A {\em typing context}, ranged over by $\tctx,\tctxtwo,\hdots$,
is a finite assignment of variables to types,
written as $\var_1:\ev_1,\hdots,\var_n:\ev_n$.
We write $\dom{\tctx}$ for the {\em domain} of $\tctx$,
\ie the finite set $\set{\var_1,\hdots,\var_n}$.
Typing judgments in $\lambdaPRK$ are of the form $\judg{\tctx}{\tm}{\ev}$,
meaning that $\tm$ has type $\ev$ under the context $\tctx$.
Derivable judgments are given inductively by the typing rules below.

We write $\judgPRK{\tctx}{\tm}{\ev}$
if the typing judgment $\judg{\tctx}{\tm}{\ev}$
is derivable in $\lambdaPRK$.
When we wish to emphasize the logical point of view,
we may write sequents as $\log{\ev_1,\hdots,\ev_n}{\evtwo}$,
and we may write $\logPRK{\ev_1,\hdots,\ev_n}{\evtwo}$ to mean
that there exists a term $\tm$ such that
$\judgPRK{\var_1:\ev_1,\hdots,\var_n:\ev_n}{\tm}{\evtwo}$.

\begin{fragBox}{Basic rules}
\[
\indrule{\Ax}{
  \emptyPremise
}{
  \judg{\tctx,\var:\ev}{\var}{\ev}
}
\indrule{\Abs}{
  \judg{\tctx}{\tm}{\typ\pp}
  \HS
  \judg{\tctx}{\tmtwo}{\typ\nn}
}{
  \judg{\tctx}{\strongabs{\ev}{\tm}{\tmtwo}}{\ev}
}
\indrule{\Icp}{
  \judg{\tctx, \var : \typ\NN}{\tm}{\typ\pp}
}{
  \judg{\tctx}{\claslamp{\var:\typ\NN}{\tm}}{\typ\PP}
}
\]
\[
\indrule{\Icn}{
  \judg{\tctx, \var : \typ\PP}{\tm}{\typ\nn}
}{
  \judg{\tctx}{\claslamp{\var:\typ\PP}{\tm}}{\typ\NN}
}
\indrule{\Ecp}{
  \judg{\tctx}{\tm}{\typ\PP}
  \Hs
  \judg{\tctx}{\tmtwo}{\typ\NN}
}{
  \judg{\tctx}{\clasapp{\tm}{\tmtwo}}{\typ\pp}
}
\indrule{\Ecn}{
  \judg{\tctx}{\tm}{\typ\NN}
  \Hs
  \judg{\tctx}{\tmtwo}{\typ\PP}
}{
  \judg{\tctx}{\clasapn{\tm}{\tmtwo}}{\typ\nn}
}
\]
\end{fragBox}

\begin{fragBox}{Conjunction and disjunction}
\[
\indrule{\colorand{\Iandp}}{
  \judg{\tctx}{\tm}{\typ\PP}
  \HS
  \judg{\tctx}{\tmtwo}{\typtwo\PP}
}{
  \judg{\tctx}{\pairp{\tm}{\tmtwo}}{(\typ \land \typtwo)\pp}
}
\indrule{\colorand{\Iorn}}{
  \judg{\tctx}{\tm}{\typ\NN}
  \HS
  \judg{\tctx}{\tmtwo}{\typtwo\NN}
}{
  \judg{\tctx}{\pairn{\tm}{\tmtwo}}{(\typ \lor \typtwo)\nn}
}
\indrule{\colorand{\Eandp}}{
  \judg{\tctx}{\tm}{(\typ_1 \land \typ_2)\pp}
}{
  \judg{\tctx}{\projip{\tm}}{\typ_i\PP}
}
\]
\[
\indrule{\colorand{\Eorn}}{
  \judg{\tctx}{\tm}{(\typ_1 \lor \typ_2)\nn}
}{
  \judg{\tctx}{\projin{\tm}}{\typ_i\NN}
}
\indrule{\colorand{\Iorp}}{
  \judg{\tctx}{\tm}{\typ_i\PP}
  \HS i \in \set{1, 2}
}{
  \judg{\tctx}{\inip{\tm}}{(\typ_1 \lor \typ_2)\pp}
}
\indrule{\colorand{\Iandn}}{
  \judg{\tctx}{\tm}{\typ_i\NN}
  \HS i \in \set{1, 2}
}{
  \judg{\tctx}{\inin{\tm}}{(\typ_1 \land \typ_2)\nn}
}
\]
\[
\indrule{\colorand{\Eorp}}{
  \judg{\tctx}{\tm}{(\typ \lor \typtwo)\pp}
  \HS
  \judg{\tctx, \var:\typ\PP}{\tmtwo}{\ev}
  \HS
  \judg{\tctx, \vartwo:\typtwo\PP}{\tmthree}{\ev}
}{
  \judg{\tctx}{\casep{\tm}{\var:\typ\PP}{\tmtwo}{\vartwo:\typtwo\PP}{\tmthree}}{\ev}
}
  \] %
  \[ %
\indrule{\colorand{\Eandn}}{
  \judg{\tctx}{\tm}{(\typ \land \typtwo)\nn}
  \HS
  \judg{\tctx, \var:\typ\NN}{\tmtwo}{\ev}
  \HS
  \judg{\tctx, \vartwo:\typtwo\NN}{\tmthree}{\ev}
}{
  \judg{\tctx}{\casen{\tm}{\var:\typ\NN}{\tmtwo}{\vartwo:\typtwo\NN}{\tmthree}}{\ev}
}
\]
\end{fragBox}

\begin{fragBox}{Implication and co-implication}
\[
\indrule{\colorimp{\Iimpp}}{
  \judg{\tctx,\var:\typ\PP}{\tm}{\typtwo\PP}
}{
  \judg{\tctx}{\lamp{\var:\typ\PP}{\tm}}{(\typ\imp\typtwo)\pp}
}
\!\!
\indrule{\colorimp{\Icoimpn}}{
  \judg{\tctx,\var:\typ\NN}{\tm}{\typtwo\NN}
}{
  \judg{\tctx}{\lamn{\var:\typ\NN}{\tm}}{(\typ\coimp\typtwo)\nn}
}
\!\!
\indrule{\colorimp{\Eimpp}}{
  \judg{\tctx}{\tm}{(\typ\imp\typtwo)\pp}
  \Hs\!
  \judg{\tctx}{\tmtwo}{\typ\PP}
}{
  \judg{\tctx}{\app{\tm}{\tmtwo}}{\typtwo\PP}
}
\]
\[
\indrule{\colorimp{\Ecoimpn}}{
  \judg{\tctx}{\tm}{(\typ\coimp\typtwo)\nn}
  \Hs\!
  \judg{\tctx}{\tmtwo}{\typ\NN}
}{
  \judg{\tctx}{\apn{\tm}{\tmtwo}}{\typtwo\NN}
}
\!\!
\indrule{\colorimp{\Icoimpp}}{
  \judg{\tctx}{\tm}{\typ\NN}
  \HS\!
  \judg{\tctx}{\tmtwo}{\typtwo\PP}
}{
  \judg{\tctx}{\copairp{\tm}{\tmtwo}}{(\typ\coimp\typtwo)\pp}
}
\!\!
\indrule{\colorimp{\Iimpn}}{
  \judg{\tctx}{\tm}{\typ\PP}
  \Hs\!
  \judg{\tctx}{\tmtwo}{\typtwo\NN}
}{
  \judg{\tctx}{\copairn{\tm}{\tmtwo}}{(\typ\imp\typtwo)\nn}
}
\]
\[
\indrule{\colorimp{\Ecoimpp}}{
  \judg{\tctx}{\tm}{(\typ\coimp\typtwo)\pp}
  \Hs\!\!
  \judg{\tctx,\var:\typ\NN,\vartwo:\typtwo\PP}{\tmtwo}{\ev}
}{
  \judg{\tctx}{\colamp{\tm}{\var:\typ\NN}{\vartwo:\typtwo\PP}{\tmtwo}}{\ev}
}
\!\!\!
\indrule{\colorimp{\Eimpn}}{
  \judg{\tctx}{\tm}{(\typ\imp\typtwo)\nn}
  \Hs\!\!
  \judg{\tctx,\var:\typ\PP,\vartwo:\typtwo\NN}{\tmtwo}{\ev}
}{
  \judg{\tctx}{\colamn{\tm}{\var:\typ\PP}{\vartwo:\typtwo\NN}{\tmtwo}}{\ev}
}
\]
\end{fragBox}

\begin{fragBox}{Negation}
\[
\indrule{\colorneg{\Inotp}}{
  \judg{\tctx}{\tm}{\typ\NN}
}{
  \judg{\tctx}{\negip{\tm}}{(\neg\typ)\pp}
}
\indrule{\colorneg{\Inotn}}{
  \judg{\tctx}{\tm}{\typ\PP}
}{
  \judg{\tctx}{\negin{\tm}}{(\neg\typ)\nn}
}
\indrule{\colorneg{\Enotp}}{
  \judg{\tctx}{\tm}{(\neg\typ)\pp}
}{
  \judg{\tctx}{\negep{\tm}}{\typ\NN}
}
\indrule{\colorneg{\Enotn}}{
  \judg{\tctx}{\tm}{(\neg\typ)\nn}
}{
  \judg{\tctx}{\negen{\tm}}{\typ\PP}
}
\]
\end{fragBox}

\begin{fragBox}{Second-order quantification}
\[
  \indrule{\colorall{\Iallp}}{
    \tctx \vdash \tm : \typ\PP
    \HS
    \btyp \notin \ftv{\tctx}
  }{
    \tctx \vdash \lamtp{\btyp}{\tm} : (\all{\btyp}{\typ})\pp
  }
  \indrule{\colorall{\Iexn}}{
    \tctx \vdash \tm : \typ\NN
    \HS
    \btyp \notin \ftv{\tctx}
  }{
    \tctx \vdash \lamtn{\btyp}{\tm} : (\ex{\btyp}{\typ})\nn
  }
  \indrule{\colorall{\Eallp}}{
    \tctx \vdash \tm : (\all{\btyp}{\typtwo})\pp
  }{
    \tctx \vdash \apptp{\tm}{\typ} : \typtwo\PP\sub{\btyp}{\typ}
  }
\]
\[
  \indrule{\colorall{\Eexn}}{
    \tctx \vdash \tm : (\ex{\btyp}{\typtwo})\nn
  }{
    \tctx \vdash \apptn{\tm}{\typ} : \typtwo\NN\sub{\btyp}{\typ}
  }
  \indrule{\colorall{\Iexp}}{
    \tctx \vdash \tm : \typtwo\PP\sub{\btyp}{\typ}
  }{
    \tctx \vdash \patp{\typ}{\tm} : (\ex{\btyp}{\typtwo})\pp
  }
  \indrule{\colorall{\Ialln}}{
    \tctx \vdash \tm : \typtwo\NN\sub{\btyp}{\typ}
  }{
    \tctx \vdash \patn{\typ}{\tm} : (\all{\btyp}{\typtwo})\nn
  }
\]
\[
  \indrule{\colorall{\Eexp}}{
    \tctx \vdash \tm : (\ex{\btyp}{\typ})\pp
    \HS
    \tctx,\var:\typ\PP \vdash \tmtwo : \ev
    \HS
    \btyp \notin \ftv{\tctx,\ev}
  }{
    \tctx \vdash \optp{\tm}{\btyp}{\var:\typ\PP}{\tmtwo} : \ev
  }
  \] %
  \[ %
  \indrule{\colorall{\Ealln}}{
    \tctx \vdash \tm : (\all{\btyp}{\typ})\nn
    \HS
    \tctx,\var:\typ\NN \vdash \tmtwo : \ev
    \HS
    \btyp \notin \ftv{\tctx,\ev}
  }{
    \tctx \vdash \optn{\tm}{\btyp}{\var:\typ\NN}{\tmtwo} : \ev
  }
\]
\end{fragBox}

The typing rules may be informally explained as follows.
$\Ax$ is the standard axiom.
The {\em absurdity rule} ($\Abs$) allows to derive any conclusion
from a strong proof and a strong refutation of $\typ$.
Introduction and elimination rules for
weak affirmation and denial ($\Icpn$, $\Ecpn$)
follow the principle that a weak affirmation $\typ\PP$
behaves like an implication ``$\typ\NN \to \typ\pp$''.
Indeed, $\Icp$ and $\Ecp$ have the same structure as the introduction
and the elimination rule for an implication ``$\typ\NN \to \typ\pp$'',
where $\claslamp{\var}{\tm}$ and $\clasapp{\tm}{\tmtwo}$
are akin to $\lambda$-abstraction and application.
The intuition behind this is that $\typ\PP$ is the type of
weak proofs of a proposition $\typ$,
where a weak proof proceeds by {\em reductio ad absurdum},
assuming a weak refutation ($\typ\NN$),
and providing a {\em strong} proof~($\typ\pp$).
Dually, a weak denial $\typ\NN$ behaves like ``$\typ\PP \to \typ\nn$''.

The remaining rules are introduction and elimination rules
for positive and negative strong connectives.
These rules come in dual pairs: {\em for each rule for a connective with
positive sign there is a symmetric rule for the dual connective
with negative sign}.
For example, the introduction rule for positive conjunction ($\Iandp$)
states that to strongly prove $\typ\land\typtwo$
it suffices to weakly prove $\typ$ and weakly prove $\typtwo$.
Dually, the introduction rule for negative disjunction ($\Iorn$)
states that to strongly refute $\typ\lor\typtwo$
it suffices to weakly refute $\typ$ and weakly refute $\typtwo$.

The introduction and elimination rules for most logical connectives
($\land$, $\lor$, $\imp$, $\coimp$, $\forall$, $\exists$)
are mechanically derived from the standard natural deduction rules
following this methodology, taking in account that the dual pairs
of connectives are
$(\land,\lor)$, $(\imp,\coimp)$, and $(\forall,\exists)$.
In general,
{\em introduction rules have weak premises and strong conclusions},
whereas {\em elimination rules have strong premises and weak conclusions}.

The typing rules for conjunction ($\Iandp,\Eandp,\Iandn,\Eandn$),
disjunction ($\Iorp,\Eorp,\Iorn,\Eorn$),
positive implication ($\Iimpp,\Eimpp$),
negative co-implication ($\Icoimpn,\Ecoimpn$),
universal ($\Iallp,\Eallp,\Ialln,\Ealln$)
and existential quantification ($\Iexp,\Eexp,\Iexn,\Eexn$)
are typical, so for instance
$\pairpn{\tm}{\tmtwo}$ forms a pair,
$\projipn{\tm}$ is the $i$-th projection,
$\inipn{\tm}$ is the $i$-th injection into a disjoint union type,
$\casepn{\tm}{\var}{\tmtwo}{\vartwo}{\tmthree}$
is a pattern matching construct, and so on.
The rules differ from usual typed $\lambda$-calculi
only in the signs and strengths that decorate premises and conclusions.

The typing rules for positive co-implication ($\Icoimpp,\Ecoimpp$),
sometimes called {\em subtraction}~\cite{Crolard01},
and negative implication ($\Iimpn,\Eimpn$)
follow the rough interpretation of $\typ\coimp\typtwo$
as $\neg\typ\land\typtwo$,
so a strong proof of a co-implication $\typ\coimp\typtwo$
is given by a pair $\copairp{\tm}{\tmtwo}$
comprising a weak refutation of $\typ$ and a weak proof of $\typtwo$.
Dually, a strong refutation of $\typ\imp\typtwo$
is given by a pair $\copairn{\tm}{\tmtwo}$
comprising a weak proof of $\typ$ and a weak refutation of $\typtwo$.
The eliminators $\colampn{\tm}{\var}{\vartwo}{\tmtwo}$ are presented
as generalized elimination rules, in multiplicative style.

The typing rules for negation ($\Inotp,\Enotp,\Inotn,\Enotn$) express
that to strongly prove $\neg\typ$ is the same as to weakly refute $\typ$,
and dually for strong refutations of $\neg\typ$.

\begin{example}
Let $\top \eqdef \all{\btyp}{(\btyp\imp\btyp)}$
and $\bot \eqdef \all{\btyp}{\btyp}$.
Recall from~\cite{BFlics21} that
the weak non-contradiction principle $\log{\tctx}{(\typ\land\neg\typ)\NN}$
holds in $\lambdaPRK$.
Then $\log{}{\top\PP}$ and $\log{\bot\PP}{\typ\PP}$
hold, where $\typ$ stands for any pure type:
\[
  \begin{array}{c@{\hspace{2cm}}c}
    \indruleN{\Icp}{
      \indruleN{\Iallp}{
        \indruleN{\Icp}{
          \indruleN{\Iimpp}{
            \indruleN{\Ax}{
            }{
              \log{\top\NN,(\btyp\imp\btyp)\NN,\btyp\PP}{\btyp\PP}
            }
          }{
            \log{\top\NN,(\btyp\imp\btyp)\NN}{(\btyp\imp\btyp)\pp}
          }
        }{
          \log{\top\NN}{(\btyp\imp\btyp)\PP}
        }
      }{
        \log{\top\NN}{\top\pp}
      }
    }{
      \log{}{\top\PP}
    }
  &
    \indruleN{\Abs'}{
      \indruleN{\Ax}{
      }{\log{\bot\PP}{\bot\PP}}
      \HS
      \indruleN{\Icn}{
        \indruleN{\Ialln}{
          \indruleN{}{
            \text{(By weak non-contradiction.)}
          }{
            \log{\bot\PP,\bot\PP}{(\typtwo\land\neg\typtwo)\NN}
          }
        }{
          \log{\bot\PP,\bot\PP}{\bot\nn}
        }
      }{
        \log{\bot\PP}{\bot\NN}
      }
    }{
      \log{\bot\PP}{\typ\PP}
    }
  \end{array}
\]
\end{example}

\subparagraph{The $\lambdaPRK$-calculus}
The opposite type $\ev\OP$ of a given type $\ev$
is defined by flipping the sign,
\ie
$(\typ\PP)\OP \eqdef \typ\NN$;
$(\typ\NN)\OP \eqdef \typ\PP$;
$(\typ\pp)\OP \eqdef \typ\nn$;
and $(\typ\nn)\OP \eqdef \typ\pp$.
If $\judgPRK{\tctx}{\tm}{\ev}$ and $\judgPRK{\tctx}{\tmtwo}{\ev\OP}$
then a term $\abs{\evtwo}{\tm}{\tmtwo}$ may be constructed
such that $\judgPRK{\tctx}{\abs{\evtwo}{\tm}{\tmtwo}}{\evtwo}$,
as follows:
\[
  \begin{array}{rl@{\hspace{.5cm}}rl}
  \abs{\evtwo}{\tm}{\tmtwo} \eqdef \strongabs{\evtwo}{\tm}{\tmtwo}
    & \text{if $\ev = \typ\pp$}
  &
  \abs{\evtwo}{\tm}{\tmtwo} \eqdef \strongabs{\evtwo}{\tmtwo}{\tm}
    & \text{if $\ev = \typ\nn$}
  \\
  \abs{\evtwo}{\tm}{\tmtwo} \eqdef
    \strongabs{\evtwo}{(\clasapp{\tm}{\tmtwo})}{(\clasapn{\tmtwo}{\tm})}
    & \text{if $\ev = \typ\PP$}
  &
  \abs{\evtwo}{\tm}{\tmtwo} \eqdef
    \strongabs{\evtwo}{(\clasapp{\tmtwo}{\tm})}{(\clasapn{\tm}{\tmtwo})}
    & \text{if $\ev = \typ\NN$}
  \\
  \end{array}
\]

We endow typable $\PRK$ terms with a notion of reduction,
defining the {\bf $\lambdaPRK$-calculus}
by a binary rewriting relation $\toa{}$
on typable $\PRK$ terms,
given by the rewriting rules below,
and closed by compatibility under arbitrary contexts.
Rules are presented following the convention that,
if many occurrences of ``$\pm$'' appear in the same expression,
they are all supposed to stand for the same sign:
\[
{\small
  \begin{array}{r@{\,}l@{\,}l@{\,\,\,}@{\,\,\,}r@{\,}l@{\,}l}
    \clasappn{(\claslampn{\var}{\tm})}{\tmtwo}
    & \toa{\ruleBetaC} &
    \tm\sub{\var}{\tmtwo}
  \\
    \projipn{\pairpn{\tm_1}{\tm_2}}
    & \toa{\colorand{\ruleProj}} &
    \tm_i
  &
    \casepn{(\inipn{\tm})}{\var}{\tmtwo_1}{\var}{\tmtwo_2}
    & \toa{\colorand{\ruleCase}} &
    \tmtwo_i\sub{\var}{\tm}
  \\
    \appn{(\lampn{\var}{\tm})}{\tmtwo}
    & \toa{\colorimp{\ruleBetaB}} &
    \tm\sub{\var}{\tmtwo}
  &
    \colampn{\copairpn{\tm}{\tmtwo}}{\var}{\vartwo}{\tmthree}
    & \toa{\colorimp{\ruleCoproj}} &
    \tmthree\sub{\var}{\tm}\sub{\vartwo}{\tmtwo}
  \\
    \negepn{(\negipn{\tm})}
    & \toa{\colorneg{\ruleNeg}} &
    \tm
  \\
    \apptpn{(\lamtpn{\btyp}{\tm})}{\typ}
    & \toa{\colorall{\ruleAppT}} &
    \tm\sub{\btyp}{\typ}
  &
    \optpn{\patpn{\typ}{\tm}}{\btyp}{\var}{\tmtwo}
    & \toa{\colorall{\ruleOpen}} &
    \tmtwo\sub{\btyp}{\typ}\sub{\var}{\tm}
  \\
    \strongabs{}{\pairp{\tm_1}{\tm_2}}{\inin{\tmtwo}}
    & \toa{\colorand{\ruleAbsPairInj}} &
    \abs{}{\tm_i}{\tmtwo}
  &
    \strongabs{}{\inip{\tm}}{\pairn{\tmtwo_1}{\tmtwo_2}}
    & \toa{\colorand{\ruleAbsInjPair}} &
    \abs{}{\tm}{\tmtwo_i}
  \\
    \strongabs{}{\lamp{\var}{\tm}}{\copairn{\tmtwo}{\tmthree}}
    & \toa{\colorimp{\ruleAbsLamCopair}} &
    \abs{}{\tm\sub{\var}{\tmtwo}}{\tmthree}
  &
    \strongabs{}{\copairp{\tm}{\tmtwo}}{\lamn{\var}{\tmthree}}
    & \toa{\colorimp{\ruleAbsCopairLam}} &
    \abs{}{\tmtwo}{\tmthree\sub{\var}{\tm}}
  \\
    \strongabs{}{(\negip{\tm})}{(\negin{\tmtwo})}
    & \toa{\colorneg{\ruleAbsNeg}} &
    \abs{}{\tm}{\tmtwo}
  \\
    \strongabs{}{(\lamtp{\btyp}{\tm})}{\patn{\typ}{\tmtwo}}
    & \toa{\colorall{\ruleAbsLamPairT}} &
    \abs{}{\tm\sub{\btyp}{\typ}}{\tmtwo}
  &
    \strongabs{}{\patp{\typ}{\tm}}{(\lamtn{\btyp}{\tmtwo})}
    & \toa{\colorall{\ruleAbsPairLamT}} &
    \abs{}{\tm}{\tmtwo\sub{\btyp}{\typ}}
  \end{array}
}
\]

The $\lambdaPRK$-calculus has two kinds of rules:
``$\beta$'' rules, akin to proof normalization rules in natural deduction,
and ``$\strongabs{}{}{}$'' rules, akin to cut elimination rules in
sequent calculus.
The $\ruleBetaC$ rules are exactly like the standard $\beta$-rule
of the $\lambda$-calculus, with the difference that the
abstraction $\claslamp{\var:\typ\NN}{\tm}$
is not an introduction of an implication ``$\typ\NN \to \typ\pp$''
but rather the introduction of a weak affirmation $\typ\PP$.
The $\ruleBetaB$ rules also describe a similar behavior,
where $\lamp{\var:\typ\PP}{\tm}$ is of type $(\typ\to\typtwo)\PP$.
The remaining $\beta$ rules are straightforward,
encoding projection ($\ruleProj$), pattern matching ($\ruleCase$),
etcetera.

The $\strongabs{}{}$ rules simplify
an absurdity $(\strongabs{}{\tm}{\tmtwo})$ as much as possible,
but they are never able to get rid of the absurdity.
Indeed, note that the right-hand side of all the $\strongabs{}{}$ rules
include the generalized absurdity operator ($\abs{}{}{}$),
which is in turn defined in terms of the absurdity operator
($\strongabs{}{}$).
Recall, for example, that if $\tm:\typ\PP$ and $\tmtwo:\typ\NN$
then
$\abs{}{\tm}{\tmtwo} =
 \strongabs{}{(\clasapp{\tm}{\tmtwo})}{(\clasapn{\tmtwo}{\tm})}$.
This means that an instance of the absurdity
is a proof which provides {\em no relevant information}.
For example, if $\ev = (\typtwo\land\typthree)\pp$
then $\projip[1]{\strongabs{\ev}{\tm}{\tmtwo}}$ is a well-formed term
of type $\typtwo\PP$;
but the argument of the projection is a term $\strongabs{\ev}{\tm}{\tmtwo}$
which may never become a pair $\pairp{p}{q}$.
This means that the normal form will be a ``stuck''
term of the form $\projip[1]{\strongabs{\ev}{\tm'}{\tmtwo'}}$.

\begin{example}[Reduction in $\lambdaPRK$]
Let $\judgPRK{\tctx}{\tm}{\typ\PP}$ and $\judgPRK{\tctx}{\tmtwo}{\typ\NN}$.
Then:
\[
\begin{array}{ll}
&
  \strongabs{}{
    (\lamtp{\btyp}{
      \claslamp{\under:(\btyp\imp\btyp)\NN}{
        \lamp{\var:\btyp\PP}{
          \var
        }
      }
    })
  }{
    \patn{\typ}{
      \claslamn{\under:(\btyp\imp\btyp)\PP}{
        \copairn{
          \tm
        }{
          \tmtwo
        }
      }
    }
  }
\\
\toa{\ruleAbsLamPairT} &
  \abs{}{
    (\claslamp{\under:(\typ\imp\typ)\NN}{
        \lamp{\var:\typ\PP}{
          \var
        }
      })
  }{
    (\claslamn{\under:(\typ\imp\typ)\PP}{
      \copairn{
        \tm
      }{
        \tmtwo
      }
    })
  }
\\
= &
    (\clasapp{
      (\claslamp{\under:(\typ\imp\typ)\NN}{
          \lamp{\var:\typ\PP}{
            \var
          }
        })
    }{
      (\claslamn{\under:(\typ\imp\typ)\PP}{
        \copairn{
          \tm
        }{
          \tmtwo
        }
      })
    })
    \,\strongabssym{}\,
\\
&
    (\clasapn{
      (\claslamn{\under:(\typ\imp\typ)\PP}{
        \copairn{
          \tm
        }{
          \tmtwo
        }
      })
    }{
      (\claslamp{\under:(\typ\imp\typ)\NN}{
          \lamp{\var:\typ\PP}{
            \var
          }
        })
    })
\\
\rtoa{{\small\ruleBClasPN}} &
  \strongabs{}{
    (\lamp{\var:\typ\PP}{ \var })
  }{
    \copairn{
      \tm
    }{
      \tmtwo
    }
  }
\,\,\toa{\ruleAImp}\,\, 
  \abs{}{\tm}{\tmtwo}
  =
  \strongabs{}{
    (\clasapp{
      \tm
    }{
      \tmtwo
    })
  }{
    (\clasapn{
      \tmtwo
    }{
      \tm
    })
  }
\end{array}
\]
\end{example}

An $\eta$-like rewriting rule can also be incorporated to $\lambdaPRK$
as done in~\cite[Thm.~37]{BFlics21},
declaring that
$\claslampn{\var}{(\clasappn{\tm}{\var})} \toa{\ruleEta} \tm$
if $\var\notin\fv{\tm}$.
The $\lambdaPRK$-calculus (with or without the $\ruleEta$ rule)
enjoys the following properties:
\begin{theorem}
\lthm{subject_reduction}
\lthm{lambdaPRK_confluence}
\lthm{lambdaPRK_refinement}
\quad
\begin{enumerate}
\item \resultName{Subject reduction.}
  If $\judgPRK{\tctx}{\tm}{\ev}$ and $\tm \toa{} \tmtwo$,
  then $\judgPRK{\tctx}{\tmtwo}{\ev}$.
\item \resultName{Confluence.}
  The $\lambdaPRK$-calculus has the Church--Rosser property.
\item \resultName{Classical refinement.}
  $\log{\typ\PP_1,\hdots,\typ\PP_n}{\typtwo\PP}$
  holds in $\lambdaPRK$
  if and only if
  $\log{\typ_1,\hdots,\typ_n}{\typtwo}$
  holds in the classical second-order natural deduction system $\NK$.
\end{enumerate}
\end{theorem}
\begin{proof}
\resultName{Subject reduction}
is a straightforward extension of~\cite[Prop.~24]{BFlics21},
with minor adaptations to account for
\colorimp{implication, co-implication},
and \colorall{second-order quantification}.
\resultName{Confluence} follows from the fact that $\lambdaPRK$
can be modeled as an orthogonal higher-order rewriting system
in the sense of Nipkow~\cite{nipkow1991higher}.
\resultName{Classical refinement}
is an extension of Prop.~38 and Thm.~39 from~\cite{BFlics21};
this theorem has two parts:
\begin{itemize}
\item
  The ``only if'' direction
  ($\logPRK{\typ\PP_1,\hdots,\typ\PP_n}{\typtwo\PP}$ implies
   $\logNK{\typ_1,\hdots,\typ_n}{\typtwo}$)
  means that $\PRK$ is a {\em conservative extension}
  of classical second-order logic.
  To prove this statement, we generalize the statement as follows:
  if $\logPRK{\ev_1,\hdots,\ev_n}{\evtwo}$
  then $\logNK{\classem{\ev_1},\hdots,\classem{\ev_n}}{\classem{\evtwo}}$,
  where $\classem{\typ\PP} = \classem{\typ\pp} = \typ$
  and $\classem{\typ\NN} = \classem{\typ\nn} = \neg\typ$.
  This can be shown by a straightforward induction on the derivation
  of the first judgment.
\item
  The ``if'' direction
  ($\logNK{\typ_1,\hdots,\typ_n}{\typtwo}$ implies
   $\logPRK{\typ\PP_1,\hdots,\typ\PP_n}{\typtwo\PP}$)
  means that classical logic can be {\em embedded} into $\PRK$.
  The essence of the proof is showing that all the inference rules of
  classical second-order natural deduction are admissible
  in $\lambdaPRK$, taking the weak affirmation of all propositions
  (\ie decorating all formulas with ``$\PP$'').
  Some cases are subtle, especially elimination rules.
  Here we show the introduction and elimination rules for quantifiers
  (\SeeAppendix{see Sections~\ref{section:appendix:subject_reduction},
  \ref{section:appendix:prk_refinement} in the appendix for complete
  proofs}):
  \begin{enumerate}
  \item {\em Universal introduction.}
    Let $\judgPRK{\tctx}{\tm}{(\all{\btyp}{\typtwo})\PP}$,
    and define $\apptc{\tm}{\typ}$ as the following term:
      $
        \claslamp{(\var:(\typtwo\sub{\btyp}{\typ})\NN)}{
          (\clasapp{
            \apptp{
              (\clasapp{
                \tm
              }{
                \claslamp{(\under:(\all{\btyp}{\typtwo})\PP)}{
                  \patn{
                    \typ
                  }{
                    \var
                  }
                }
              })
            }{
              \typ
            }
          }{
            \var
          })
        }
      $.
    Then we have that
      $\judgPRK{\tctx}{\apptc{\tm}{\typ}}{\typtwo\sub{\btyp}{\typ}\PP}$.
  \item {\em Universal elimination.}
    Let $\judgPRK{\tctx}{\tm}{(\all{\btyp}{\typtwo})\PP}$.
    Define $\apptc{\tm}{\typ}$ as the following term:
      $
        \claslamp{(\var:(\typtwo\sub{\btyp}{\typ})\NN)}{
          (\clasapp{
            \apptp{
              (\clasapp{
                \tm
              }{
                \claslamp{(\under:(\all{\btyp}{\typtwo})\PP)}{
                  \patn{
                    \typ
                  }{
                    \var
                  }
                }
              })
            }{
              \typ
            }
          }{
            \var
          })
        }
      $.
      Then we have that
      $\judgPRK{\tctx}{\apptc{\tm}{\typ}}{\typtwo\sub{\btyp}{\typ}\PP}$.
  \item {\em Existential introduction.}
    Let $\judgPRK{\tctx}{\tm}{(\typtwo\sub{\btyp}{\typ})\PP}$.
    Define
      $\patc{\typ}{\tm}$
    as the following term:
      $
        \claslamp{(\under:(\ex{\btyp}{\typtwo})\NN)}{
          \patp{\typ}{
            \tm
          }
        }
      $.
    Then we have that
    $\judgPRK{\tctx}{\patc{\typ}{\tm}}{(\ex{\btyp}{\typtwo})\PP}$.
  \item {\em Existential elimination.}
    Let
    $\judgPRK{\tctx}{\tm}{(\ex{\btyp}{\typ})\PP}$
    and
    $\judgPRK{\tctx,\var:\typ\PP}{\tmtwo}{\typtwo\PP}$
    with $\btyp \not\in \ftv{\tctx,\ev}$.
    Define
      $\optj{\tm}{\btyp}{\var}{\tmtwo}$
    as the following term:
      $
        \claslamp{(\vartwo:\typtwo\NN)}{
          (\clasapp{
            \optp{
              \tm'
            }{\btyp}{\var}{
              \tmtwo
            }
          }{
            \vartwo
          })
        }
      $
    where $\tm' \eqdef 
            \clasapp{
              \tm
            }{
              \claslamn{(\under:(\ex{\btyp}{\typ})\PP)}{
                \lamtn{\btyp}{
                  \claslamn{(\var:\typ\PP)}{
                    (\abs{\typ\nn}{
                      \tmtwo
                    }{ 
                      \vartwo
                    })
                  }
                }
              }
            }$.
    Then
    $\judgPRK{\tctx}{\optc{\tm}{\btyp}{\var}{\tmtwo}}{\typtwo\PP}$.
    \qedhere
  \end{enumerate}
\end{itemize}
\end{proof}


\section{B\"ohm--Berarducci Encodings }
\lsec{bohm_berarducci_encodings}

It is well-known that, in System~F, logical connectives
such as $\top$, $\bot$, $\land$, $\lor$, $\exists$,
as well as inductive data types,
can be represented using only $\forall$ and $\imp$
by means of their {\em B\"ohm--Berarducci encodings}~\cite{BohmB85},
which can be understood as
{\em universal properties} or {\em structural induction principles}.
B\"ohm--Berarducci encodings can be reproduced in~$\lambdaPRK$.
In the following subsections we study the encoding of connectives
in terms of universal quantification and implication.

The encoding of conjunction, for instance, can be taken to be
$\typ\land\typtwo \eqdef \all{\btyp}{((\typ\to\typtwo\to\btyp)\to\btyp)}$.
Then {\bf positive typing rules for conjunction},
analogous to $\Iandp$ and $\Eandp$
are derivable, and their constructions simulate the $\ruleBAndP$ rule.
Indeed, let $X := (\typ_1\imp\typ_2\imp\typ_i)\imp\btyp$
and $Y := \typ_1\imp\typ_2\imp\btyp$.
Moreover,
let $X_i := (\typ_1\imp\typ_2\imp\typ_i)\imp\typ_i$
and $Y_i := \typ_1\imp\typ_2\imp\typ_i$.
Given $\judg{\tctx}{\tm_1}{\typ_1\PP}$ and $\judg{\tctx}{\tm_2}{\typ_2\PP}$
and $\judg{\tctx}{\tmtwo}{(\typ_1\land\typ_2)\pp}$, define:
\begin{itemize}
\item
  $
    \pairp{\tm_1}{\tm_2} \eqdef
      \lamtp{\btyp}{
        \claslamp{(\under:X\NN)}{
          \lamp{(\var:Y\PP)}{
            \claslamp{(\vartwo:\btyp\NN)}{
              \clasapp{
                \app{
                  \clasapp{
                    \app{
                      \clasapp{
                        \var
                      }{
                        (\claslamn{(\under:Y\PP)}{
                          \copairn{
                            \tm_1
                          }{
                            \tmthree
                          }
                        })
                      }
                    }{
                      \tm_1
                    }
                  }{
                    \tmthree
                  }
                }{
                  \tm_2
                }
              }{
                \vartwo
              }
            }
          }
        }
      }
  $, \\
  where
  $
    \tmthree \eqdef
       \claslamn{(\under:(\typtwo\to\btyp)\PP)}{
         \copairn{
           \tm_2
         }{
           \vartwo
         }
       }
  $.
\item
  $
    \projip{\tmtwo} \eqdef
      \claslamp{(\var:\typ_i\NN)}{
        (\clasap{
          \app{
            \clasapp{
              \apptp{
                \tm_1
              }{
                \typ_i
              }
            }{
              (\claslamp{(\under:X_i\PP)}{
                \copairn{
                  \tmfour
                }{
                  \var
                }
              })
            }
          }{
            \tmfour
          }
        }{
          \var
        })
      }
  $, \\
  where
  $
    \tmfour \eqdef
      \claslamp{(\under:Y_i\NN)}{
        \lamp{(\vartwo_1:\typ_1\PP)}{
          \claslamp{(\under:(\typ_2\imp\typ_i)\NN)}{
            \lamp{(\vartwo_2:\typ_2\PP)}{
              \vartwo_i
            }
          }
        }
      }
  $.
\end{itemize}
  Then $\judg{\tctx}{\pairp{\tm_1}{\tm_2}}{(\typ_1\land\typ_2)\pp}$
  and $\judg{\tctx}{\projip{\tmtwo}}{\typ_i\PP}$
  and it can be easily checked that
  $\projip{\pairp{\tm_1}{\tm_2}} \rto \tm_i$
  (using $\ruleEta$).

On the other hand, {\bf negative typing rules for conjunction},
analogous to $\Iandn$ and a weak variant of $\Eandn$ can also be derived.
First note that, given terms $\judg{\tctx}{\tmfive}{(\typ\to\typtwo)\PP}$
and $\judg{\tctx}{\tmsix}{\typ\PP}$,
a term $\appc{\tmfive}{\tmsix}$ may be defined
in such a way that $\judgPRK{\tctx}{\appc{\tmfive}{\tmsix}}{\typtwo\PP}$.
An explicit construction is
  $
   \appc{\tmfive}{\tmsix} \eqdef
    \claslamp{(\var:\typtwo\NN)}{
      (\clasapp{
        \app{
          \clasapp{
            \tmfive
          }{
            (\claslamn{(\under:(\typ\to\typtwo)\PP)}{
              \copairn{
                \tmsix
              }{
                \var
              }
            })
          }
        }{
          \tmsix
        }
      }{
        \var
      })
    }
  $.
  Then we may encode $\Iandn$ and a weak variant of $\Eandn$
  as follows:
  \begin{itemize}
  \item
    $
      \inin{\tm} \eqdef
        \patn{\typ_i}{
          \claslamn{(\under:X_i\PP)}{
            \copairn{
              \tmfour
            }{
              \tm
            }
          }
        }
    $,\\
    where
    $
      \tmfour \eqdef
        \claslamp{(\under:Y_i\NN)}{
          \lamp{(\vartwo_1:\typ_1\PP)}{
            \claslamp{(\under:(\typ_2\imp\typ_i)\NN)}{
              \lamp{(\vartwo_2:\typ_2\PP)}{
                \vartwo_i
              }
            }
          }
        }
    $.
  \item
    $
      \casen{\tm}{a_1}{\tmtwo_1}{a_2}{\tmtwo_2} \eqdef
          \optn{
            \tm
          }{
            \btyp
          }{
            \var:X\NN
          }{
            \claslamn{(c:\typthree\PP)}{
              \clasapn{
                \tmtwo'_1
              }{
                c
              }
            }
          }
    $,\\
    where
    $
      \tmtwo'_1 \eqdef
        \tmtwo_1\sub{a_1}{
          \claslamn{(\varthree_1:\typ_1\PP)}{
            (\abs{\typ\nn}{
              \tmtwo'_2
            }{ 
              c
            })
          }
        }
    $,
    and
    $
      \tmtwo'_2 \eqdef
        \tmtwo_2\sub{a_2}{
          \claslamn{(\varthree_2:\typ_2\PP)}{
            (\abs{\typ\nn}{
              \tmthree
            }{ 
              \var
            })
          }
        }
    $,
    and
    $
      \tmthree \eqdef
        \claslamp{\under:X\NN}{
          \lamp{\vartwo:Y\PP}{
            (\appc{
              \appc{
                \vartwo
              }{
                \varthree_1
              }
            }{
              \varthree_2
            })
          }
        }
    $.
  \end{itemize}
Note that $\judg{\tctx}{\inin{\tm}}{(\typ_1\land\typ_2)\nn}$
and $\judg{\tctx}{\casen{\tm}{a_1}{\tmtwo_1}{a_2}{\tmtwo_2}}{\typthree\NN}$.
However, unfortunately, it is
{\bf not} the case that 
$\casen{\inin{\tm}}{a_1}{\tmtwo_1}{a_2}{\tmtwo_2} \rto \tmtwo_i\sub{a_i}{\tm}$;
in fact the computation becomes stuck.

In general, these kinds of encodings are able to simulate reduction
for the positive half of the system but not for the negative half\footnote{Naturally,
one may consider dual encodings in terms of $\exists$ and $\coimp$,
for example
$\typ \lor \typtwo = \ex{\btyp}{((\typ\coimp\typtwo\coimp\btyp)\coimp\btyp)}$,
which behave well only for the negative half of the system.}.
This seems to suggest that $\lambdaPRK$ cannot be fully
simulated by the $\set{\forall,\imp}$ fragment, although we do not
know of a proof of this fact and there might exist
other encodings which allow simulating the full $\lambdaPRK$ calculus.

\SeeAppendix{In~\rsec{appendix:bohm_encoding} of the appendix,
encodings for (positive) disjunction and existential quantification
in terms of $\set{\forall,\imp}$ are also studied.}

\section{Normalization of Second-Order $\lambdaPRK$}
\lsec{strong_normalization}

In this section we construct a {\bf reducibility model}
for $\lambdaPRK$ and we prove {\bf adequacy}~(\rthm{adequacy}) of the model,
from which {\bf strong normalization}
of second-order $\lambdaPRK$ follows.
We only discuss the proof of strong normalization for the calculus without
the $\ruleEta$ rule\footnote{Strong
normalization for the full $\lambdaPRK$-calculus with the $\ruleEta$ rule comes
out as a relatively easy corollary by postponing $\ruleEta$ steps
(see \eg~\cite[Theorem~37]{BFlics21} for a similar result).}.

In~\cite{BFlics21}, strong normalization for the propositional fragment
of $\lambdaPRK$ is shown via a translation to System~F extended
with recursive type constraints
enjoying a (non-strict) positivity condition.
This technique does not seem to extend to the second-order case.
The problem is that the translation given in~\cite{BFlics21}
is {\em not closed under type substitution}.
More precisely, if we denote the translation by $\semF{-}$,
an equality such that
$\semF{\tm\sub{\btyp}{\typ}} = \semF{\tm}\sub{\btyp}{\semF{\typ}}$
does not hold in general, making the proof fail.

Our proof of strong normalization is based on an adaptation of
Girard's technique of reducibility candidates.
Specifically, we adapt Mendler's proof of strong normalization
for the extended System~F given in~\cite{mendler1991inductive}.
We begin by defining an {\em untyped} version of $\lambdaPRK$:

\subparagraph{The untyped $\lambdaPRK$-calculus ($\lambdaPRKU$)}
By $\UTerms$ we denote the set of {\em untyped terms},
given by the following grammar:
\[
  \begin{array}{rrl}
  \atm,\atmtwo,\atmthree,\hdots & ::= &
       \var
  \mid \strongabs{}{\atm}{\atmtwo}
  \mid \pair{\atm}{\atmtwo}
  \mid \proji{\atm}
  \mid \ini{\atm}
  \mid \case{\atm}{\var}{\atmtwo}{\vartwo}{\atmthree}
  \mid \lam{\var}{\atm}
  \mid \ap{\atm}{\atmtwo}
  \mid \copair{\atm}{\atmtwo}
  \mid \colam{\atm}{\var}{\vartwo}{\atmtwo}
  \\
  & \mid &
       \negi{\atm}
  \mid \nege{\atm}
  \mid \lamtu{\atm}
  \mid \apptu{\atm}
  \mid \patu{\atm}
  \mid \optu{\atm}{\var}{\atmtwo}
  \end{array}
\]
The reduction relation $\tou \subseteq \UTerms \times \UTerms$
of the $\lambdaPRKU$-calculus
is defined by the following reduction rules,
closed by compatibility under arbitrary reduction contexts:
\[
  \begin{array}{rcl@{\HS}rcl}
    \proji{\pair{\atm_1}{\atm_2}}
    & \tou &
    \atm_i
  &
    \case{(\ini{\atm})}{\var}{\atmtwo_1}{\var}{\atmtwo_2}
    & \tou &
    \atmtwo_i\sub{\var}{\atm}
  \\
    \ap{(\lam{\var}{\atm})}{\atmtwo}
    & \tou &
    \atm\sub{\var}{\atmtwo}
  &
    \colam{\copair{\atm_1}{\atm_2}}{\var}{\vartwo}{\atmtwo}
    & \tou &
    \atmtwo\sub{\var}{\atm_1}\sub{\vartwo}{\atm_2}
  \\
    \nege{(\negi{\atm})}
    & \tou &
    \atm
  \\
    \apptu{(\lamtu{\atm})}
    & \tou &
    \atm
  &
    \optu{\patu{\atm}}{\var}{\atmtwo}
    & \tou &
    \atmtwo\sub{\var}{\atm}
  \\
    \strongabs{}{\pair{\atm_1}{\atm_2}}{\ini{\atmtwo}}
    & \tou &
    \abs{}{\atm_i}{\atmtwo}
  &
    \strongabs{}{\ini{\atm}}{\pair{\atmtwo_1}{\atmtwo_2}}
    & \tou &
    \abs{}{\atm}{\atmtwo_i}
  \\
    \strongabs{}{\lam{\var}{\atm}}{\copair{\atmtwo}{\atmthree}}
    & \tou &
    \abs{}{\atm\sub{\var}{\atmtwo}}{\atmthree}
  &
    \strongabs{}{\copair{\atm}{\atmtwo}}{\lam{\var}{\atmthree}}
    & \tou &
    \abs{}{\atmtwo}{\atmthree\sub{\var}{\atm}}
  \\
    \strongabs{}{(\negi{\atm})}{(\negi{\atmtwo})}
    & \tou &
    \abs{}{\atmtwo}{\atm}
  \\
    \strongabs{}{\lamtu{\atm}}{\patu{\atmtwo}}
    & \tou &
    \abs{}{\atm}{\atmtwo}
  &
    \strongabs{}{\patu{\atm}}{\lamtu{\atmtwo}}
    & \tou &
    \abs{}{\atm}{\atmtwo}
  \end{array}
\]
where
$\abs{}{\atm}{\atmtwo} \eqdef
  \strongabs{}{(\ap{\atm}{\atmtwo})}{(\ap{\atmtwo}{\atm})}$.
The set $\CanTerms \subseteq \UTerms$ of
{\em canonical terms} is the set of terms
built with a constructor, \ie of any of the forms:
  $\pair{\atm}{\atmtwo}$,
  $\ini{\atm}$,
  $\lam{\var}{\atm}$,
  $\copair{\atm}{\atmtwo}$,
  $\negi{\atm}$,
  $\lamtu{\atm}$,
  $\patu{\atm}$.

Note that the untyped calculus $\lambdaPRKU$ is obtained from $\lambdaPRK$ by
erasing all signs and type annotations from terms,
replacing types and type variables by a placeholder ``$\dummy$''
in introductors and eliminators for quantifiers,
and identifying\footnote{This identification is not essential,
but just a matter of syntactic economy.}
``weak'' abstraction and application
($\claslam{\var}{\tm}$ and $\clasap{\tm}{\tmtwo}$)
with regular abstraction and application
($\lam{\var}{\tm}$ and $\ap{\tm}{\tmtwo}$).
It is easy to note that the $\tou$ reduction is confluent,
observing that it can be modeled as an
orthogonal higher-order rewriting system~\cite{nipkow1991higher}.

One difficult aspect of the strong normalization proof is that
terms of type $\typ\PP$ behave as functions ``$\typ\NN \to \typ\pp$''
and, dually,
terms of $\typ\NN$ behave as functions ``$\typ\PP \to \typ\nn$''.
Consequently,
sets of {\em reducible terms} cannot be defined by straightforward recursion,
as this would lead to a non-well-founded mutual dependency
between reducible terms of types $\typ\PP$ and $\typ\NN$.
To address this difficulty, we follow Mendler's approach
of taking fixed points in the
{\em complete lattice of reducibility candidates}.

\subsection{A Reducibility Model for $\lambdaPRK$}

We begin by recalling a few standard notions from order theory.
A {\em complete lattice} is a partially ordered set $(A,\leq)$
such that every subset $B \subseteq A$
has a least upper bound and a greatest lower bound,
denoted respectively by $\join{B}$ and $\meet{B}$.
Then (see~\cite[Thm.~2.35]{davey1990}):
\begin{theorem}[Knaster--Tarski fixed point theorem]
\lthm{knaster_tarski}
If $(A,\leq)$ is a complete lattice and
$f : A \to A$ is an order-preserving map,
\ie $a \leq a' \implies f(a) \leq f(a')$,
then $f$ has a {\em least fixed point}
and a {\em greatest fixed point}, given respectively by:
$
  \lfpF{f} = \meet\set{a \in A \ST f(a) \leq a}
$
and
$
  \gfpF{f} = \join\set{a \in A \ST a \leq f(a)}
$.
\end{theorem}
\noindent We write
$\lfp{\rc}{f(\rc)}$ for $\lfpF{f}$
and
$\gfp{\rc}{f(\rc)}$ for $\gfpF{f}$.

\subparagraph{Reducibility candidates}
Let $\SNTerms \subseteq \UTerms$ denote the set of
{\em strongly normalizing} terms, with respect to $\tou$.
A set $\rc \subseteq \SNTerms$ is {\em closed by reduction}
if for every $\atm,\atmtwo\in\UTerms$ such that $\atm \in \rc$
and $\atm \tou \atmtwo$,
one has that $\atmtwo \in \rc$.
A set $\rc \subseteq \SNTerms$ is {\em complete}
if for every $\atm \in \SNTerms$ the following holds:
\[
  \text{
    ($\forall \atmtwo \in \CanTerms.\ %
         ((\atm \rtou \atmtwo) \implies \atmtwo \in \rc$))
    \HS implies \HS
    $\atm \in \rc$
  }
\]
A set $\rc \subseteq \SNTerms$ is a {\em reducibility candidate}
(or a \RC for short)
if it is closed by reduction and complete.
We write $\RCSet$ for the set of all \RCs,
that is,
$\RCSet \eqdef
 \set{\rc \subseteq \SNTerms \ST \rc \text{ is a \RC}}$.

It is easy to see that reducibility candidates are non-empty.
In particular, for every $\rc \in \RCSet$
we have that any variable $\var \in \rc$
is strongly normalizing and it vacuously verifies the property
$\forall \atmthree \in \CanTerms.\ %
 ((\var \rtou \atmthree) \implies \atmthree \in \rc)$
so, since $\rc$ is complete, we have that $\var \in \rc$.
Moreover, the set $\RCSet$ forms a complete lattice
ordered by inclusion $\subseteq$.
Following Mendler~\cite[Prop.~2]{mendler1991inductive},
the greatest lower bound of $\set{\rc_i}_{i\inI}$
is given by the intersection $\cap_{i\inI}{\rc_i}$,
and the bottom element is the set
$\rcbot = \set{\atm \in \SNTerms \ST \nexists\atmtwo\in\CanTerms.\,\atm\rtou\atmtwo}$
of terminating terms that have no canonical reduct.
\SeeAppendix{%
See~\rsec{appendix:strong_normalization} in the appendix for details.}

\subparagraph{Operations on reducibility candidates}
For each set of canonical terms $X \subseteq \CanTerms$, we define its
{\em closure} $\rclosure{X}$
as the set of all strongly normalizing terms
whose canonical reducts are in $X$. More precisely,
$
  \rclosure{X} \eqdef
    \set{\atm \in \SNTerms \ST
      \forall\atmtwo\in\CanTerms.
      ((\atm \rtou \atmtwo) \implies \atmtwo \in X)
    }
$.
If $\rc_1,\rc_2$ are \RCs and if $\set{\rc_i}_{i\inI}$ is a set of \RCs, we define
the following operations:
\[
{\small
  \begin{array}{rcl@{\HS}rcl}
    (\rc_1 \rctimes \rc_2)
    & \eqdef &
    \rclosure{\set{\pair{\atm_1}{\atm_2} \ST \atm_1 \in \rc_1, \atm_2 \in \rc_2}}
  &
    (\rc_1 \rcplus \rc_2)
    & \eqdef &
    \rclosure{\set{\ini{\atm} \ST i \in \set{1,2}, \atm \in \rc_i}}
  \\
    (\rc_1 \rcimp \rc_2)
    & \eqdef & 
     \set{\atm \in \SNTerms \ST
       \forall\atmtwo\in\rc_1.\ \ap{\atm}{\atmtwo} \in \rc_2
     }
  &
    (\rc_1 \rccoimp \rc_2)
    & \eqdef &
    \rclosure{\set{\copair{\atm_1}{\atm_2} \ST \atm_1 \in \rc_1, \atm_2 \in \rc_2}}
  \\
    \rcneg\rc
    & \eqdef &
    \rclosure{\set{\negi{\atm} \ST \atm \in \rc}}
  \\
    \rcall{i\inI}{\rc_i} & \eqdef &
      \set{\atm \in \SNTerms \ST \forall i \inI.\ \apptu{\atm} \in \rc_i}
  &
    \rcex{i\inI}{\rc_i} & \eqdef &
      \rclosure{\set{\patu{\atm} \ST \exists i \in I.\ \atm \in \rc_i}}
  \end{array}
}
\]
It can be checked that all these operations
map \RCs to \RCs.
\SeeAppendix{See~\rsec{appendix:strong_normalization} for details.}

\noindent A straightforward observation is that the arrow operator
is order-reversing on the left,
\ie that
if $\rc_1 \subseteq \rc'_1$
then $(\rc'_1\rcimp\rc_2) \subseteq (\rc_1\rcimp\rc_2)$.

\subparagraph{Orthogonality}
The idea of the normalization proof is, as usual, to associate, to each type $\ev$,
a set of {\em reducible terms} $\red{\ev}{} \in \RCSet$.
The interpretation of a type variable,
such as $\red{\btyp\pp}{}$ or $\red{\btyp\nn}{}$
shall be given by an {\em environment} $\env$,
mapping type variables to \RCs.
However, the sets $\red{\btyp\pp}{}$ and $\red{\btyp\nn}{}$
{\em should not be chosen independently of each other}:
we require them to be orthogonal in the following sense.

Two reducibility candidates
$\rc_1,\rc_2\in\RCSet$ are {\em orthogonal},
if for all $\atm_1\in\rc_1$ and $\atm_2\in\rc_2$
we have that $(\strongabs{}{\atm_1}{\atm_2}) \in \SNTerms$.
We write $\RCPerp$ for the set of all pairs $(\rc_1,\rc_2) \in \RCSet^2$
such that $\rc_1$ and $\rc_2$ are orthogonal.

\subparagraph{Reducible terms}
The set of reducible terms is defined by induction on the following
notion of {\em measure} $\#(-)$ of a type $\ev$,
given by $\#(\typ\pp) = \#(\typ\nn) \eqdef 2 |\typ|$
and $\#(\typ\PP) = \#(\typ\NN) \eqdef 2 |\typ| + 1$,
where $|\typ|$ denotes the {\em size}, \ie the number of symbols,
of the pure type $\typ$.
Note for example that
$\#((\typ\land\typtwo)\PP) > \#((\typ\land\typtwo)\pp) > \#(\typ\PP)$.

An {\em environment}
is a function $\env$
mapping each type variable $\btyp^{\pm}$,
with either positive or negative sign,
to a reducibility candidate
$\env(\btyp^{\pm}) \in \RCSet$,
in such a way that $\rcperp{\env(\btyp\pp)}{\env(\btyp\nn)}$.
If $(\rc\pp,\rc\nn) \in \RCPerp$,
we write $\env\esub{\btyp}{\rc\pp,\rc\nn}$
for the environment $\env'$ that extends $\env$
in such a way that $\env'(\btyp\pp) = \rc\pp$
and $\env'(\btyp\nn) = \rc\nn$
and $\env'(\btyptwo^{\pm}) = \env(\btyptwo^{\pm})$
for any other type variable $\btyptwo \neq \btyp$.

Given an environment $\env$,
we define the set of {\em reducible terms} of type $\ev$
under the environment $\env$,
written $\red{\ev}{\env}$,
by induction on the measure $\#(\ev)$
as follows:
\[
{\small
\begin{array}{r@{\,\,}c@{\,\,}l@{\,\,}r@{\,\,}c@{\,\,}l}
  \red{\btyp\pp}{\env} & \eqdef & \env(\btyp\pp)
&
  \red{\btyp\nn}{\env} & \eqdef & \env(\btyp\nn)
\\
  \red{(\typ\land\typtwo)\pp}{\env}
  & \eqdef &
  \red{\typ\PP}{\env}
  \rctimes
  \red{\typtwo\PP}{\env}
&
  \red{(\typ\land\typtwo)\nn}{\env}
  & \eqdef &
  \red{\typ\NN}{\env}
  \rcplus
  \red{\typtwo\NN}{\env}
\\
  \red{(\typ\lor\typtwo)\pp}{\env}
  & \eqdef &
  \red{\typ\PP}{\env}
  \rcplus
  \red{\typtwo\PP}{\env}
&
  \red{(\typ\lor\typtwo)\nn}{\env}
  & \eqdef &
  \red{\typ\NN}{\env}
  \rctimes
  \red{\typtwo\NN}{\env}
\\
  \red{(\typ\imp\typtwo)\pp}{\env}
  & \eqdef &
  \red{\typ\PP}{\env} \rcimp \red{\typtwo\PP}{\env}
&
  \red{(\typ\imp\typtwo)\nn}{\env}
  & \eqdef &
  \red{\typ\PP}{\env} \rccoimp \red{\typtwo\NN}{\env}
\\
  \red{(\typ\coimp\typtwo)\pp}{\env}
  & \eqdef &
  \red{\typ\NN}{\env} \rccoimp \red{\typtwo\PP}{\env}
&
  \red{(\typ\coimp\typtwo)\nn}{\env}
  & \eqdef &
  \red{\typ\NN}{\env} \rcimp \red{\typtwo\NN}{\env}
\\
  \red{(\neg\typ)\pp}{\env}
  & \eqdef &
  \rcneg\red{\typ\NN}{\env}
&
  \red{(\neg\typ)\nn}{\env}
  & \eqdef &
  \rcneg\red{\typ\PP}{\env}
\\
  \red{(\all{\btyp}{\typ})\pp}{\env}
  & \eqdef &
  \rcall{(\rc\pp,\rc\nn)\in\RCPerp}{
    \red{\typ\PP}{\env\esub{\btyp}{\rc\pp,\rc\nn}}
  }
&
  \red{(\all{\btyp}{\typ})\nn}{\env}
  & \eqdef &
  \rcex{(\rc\pp,\rc\nn)\in\RCPerp}{
    \red{\typ\NN}{\env\esub{\btyp}{\rc\pp,\rc\nn}}
  }
\\
  \red{(\ex{\btyp}{\typ})\pp}{\env}
  & \eqdef &
  \rcex{(\rc\pp,\rc\nn)\in\RCPerp}{
    \red{\typ\PP}{\env\esub{\btyp}{\rc\pp,\rc\nn}}
  }
&
  \red{(\ex{\btyp}{\typ})\nn}{\env}
  & \eqdef &
  \rcall{(\rc\pp,\rc\nn)\in\RCPerp}{
    \red{\typ\NN}{\env\esub{\btyp}{\rc\pp,\rc\nn}}
  }
\\
  \red{\typ\PP}{\env}
  & \eqdef &
  \lfp{\rc}{((\rc\rcimp\red{\typ\nn}{\env})\rcimp\red{\typ\pp}{\env})}
&
  \red{\typ\NN}{\env}
  & \eqdef &
  \gfp{\rc}{((\rc\rcimp\red{\typ\pp}{\env})\rcimp\red{\typ\nn}{\env})}
\end{array}
}
\]

It is straightforward to check
for each type $\ev$ and each environment $\env$
that $\red{\ev}{\env}$ is a reducibility candidate,
by induction on the measure $\#(\ev)$.
In the case of $\red{\typ\PP}{\env}$,
by the Knaster--Tarski theorem~(\rthm{knaster_tarski}),
to see that the least fixed point exists,
it suffices to observe that the mapping
$f(\rc) = ((\rc\rcimp\red{\typ\nn}{\env})\rcimp\red{\typ\pp}{\env})$
is order-preserving.
The case of $\red{\typ\NN}{\env}$ is similar.

\subparagraph{Adequacy of the reducibility model}
For each term $\tm$ of $\lambdaPRK$,
we define an untyped term $\erase{\tm}$ of $\lambdaPRKU$
via the obvious forgetful map.
For instance $\erase{\claslamp{(\var:(\ex{\btyp}{\btyp})\NN)}{\patp{\typtwo}{\varthree^{\typtwo\PP}}}} = \lam{\var}{\patu{\varthree}}$.
Note that each reduction step $\tm \toa{} \tmtwo$ in $\lambdaPRK$
is mapped to a reduction step $\erase{\tm} \tou \erase{\tmtwo}$
in $\lambdaPRKU$.
Hence, if $\erase{\tm}$ is strongly normalizing with respect to $\tou$,
then $\tm$ is strongly normalizing with respect to $\toa{}$.

A {\em substitution} is a function $\subst$ mapping
each variable to a term in $\UTerms$.
We write $\atm^\subst$ for the term that results from
the capture-avoiding substitution
of each free occurrence of each variable $\var$ in $\atm$
by $\subst(\var)$.
We say that the substitution $\subst$ is
{\em adequate} for the typing context $\tctx$ under the environment $\env$,
and we write $\judgSubst{\subst}{\env}{\tctx}$,
if for each type assignment $(\var:\ev)\in\tctx$
we have that $\subst(\var) \in \red{\ev}{\env}$.
We are finally able to state the key result:
\begin{theorem}[Adequacy]
\lthm{adequacy}
If $\judg{\tctx}{\tm}{\ev}$
and $\judgSubst{\subst}{\env}{\tctx}$
then $\erase{\tm}^{\subst} \in \red{\ev}{\env}$.
\end{theorem}

The proof of the adequacy theorem relies on a number of auxiliary lemmas
stating properties such as
  $\red{\typ\PP}{\env} = \red{\typ\NN}{\env} \rcimp \red{\typ\pp}{\env}$
and
  $\rcperp{\red{\typ\pp}{\env}}{\red{\typ\nn}{\env}}$.
\SeeAppendix{%
See~\rsec{appendix:strong_normalization} in the appendix for the detailed proof.}
From this we obtain as an easy corollary
that the $\lambdaPRK$-calculus is strongly normalizing,
taking $\env$ as the environment that maps all type variables
to the bottom reducibility candidate,
and $\subst$ as the identity substitution.

\section{Intuitionistic Proofs and Refutations}
\lsec{prj}

In natural deduction, it is well-known that classical logic can be
obtained from the intuitionistic system by adding a single classical
axiom, such as excluded middle or double negation elimination.
In sequent calculus, it is well-known that
intuitionistic logic can be obtained by restricting sequents
$\typ_1,\hdots,\typ_n\vdash\typtwo_1,\hdots,\typtwo_m$
to have at most one formula on the right.
As we have seen, $\lambdaPRK$ refines classical logic.
It is a natural question to ask what subsystem of $\lambdaPRK$
corresponds to intuitionistic logic.

In this section we characterize a restricted subsystem
of $\lambdaPRK$ that corresponds to intuitionistic logic,
called $\lambdaPRJ$,
by imposing a syntactic restriction on the shape of $\lambdaPRK$
proofs, that forbids certain specific patterns of reasoning.
In particular, a variable $\var$ introduced by a
positive weak introduction $\claslamp{\var}{\tm}$
can only occur free in $\tm$ inside the arguments of weak eliminations.
The main result in this section is that
{\bf $\lambdaPRJ$ refines intuitionistic second-order
logic}~(\rthm{lambdaPRJ_refinement}).
\medskip

As mentioned before, proofs of {\em strong} propositions in $\PRK$ must be
constructive.
However, this is only true for the toplevel logical connective in the
formula.
In general, a proof of $\typ\pp$ in $\PRK$ does not necessarily correspond
to an intuitionistic proof of $\typ$.
For example, a canonical proof of $(\typ\land\typtwo)\pp$
is given by a proof of $\typ\PP$ and a proof of $\typtwo\PP$,
but these subproofs may resort to classical reasoning principles.

The key to identify an intuitionistic subset of
$\lambdaPRK$ is to disallow inference rules which embody
classical principles.
One example is the $\Enotn$ rule, which derives $\typ\PP$ from $(\neg\typ)\nn$.
This rule embodies the classical principle
of double negation elimination ($\neg\neg\typ \to \typ$).
Another important example is the $\Icp$ rule,
which derives $\typ\PP$ from $\typ\NN\vdash\typ\pp$.
This rule embodies the classical principle of
{\em consequentia mirabilis} ($(\neg\typ \to \typ) \to \typ$).

The analysis of the $\Icp$ rule suggests that,
in the intuitionistic fragment, $\typ\PP$
should not be identified with ``$\typ\NN \to \typ\pp$'',
but directly with $\typ\pp$.
One natural idea would be to impose an invariant over terms
of the form $\claslamp{(\var:\typ\NN)}{\tm}$,
in such a way that the body~$\tm$ may have no free occurrences 
of the negative counterfactual $\var$.
With this invariant,
\begin{wrapfigure}{r}{5.5cm}
$
  \indrule{$\Icp$ \textup{(variant)}}{
    \judg{\tctx}{\tm}{\typ\pp}
    \HS
    \var\notin\fv{\tm}
  }{
    \judg{\tctx}{\claslamp{(\var:\typ\NN)}{\tm}}{\typ\PP}
  }
$
\end{wrapfigure} 
all instances of the $\Icp$ rule are actually instances of the
variant of $\Icp$ shown on the right.
This in turn means that, in an application $\clasapp{\tm}{\tmtwo}$,
the argument $\tmtwo$ is {\em useless}.
Indeed, if $\tm$ becomes $\claslamp{(\var:\typ\NN)}{\tm'}$,
the invariant ensures that $\var\notin\fv{\tm'}$, so
$\clasapp{(\claslamp{(\var:\typ\NN)}{\tm'})}{\tmtwo} \to \tm'$,
which does not depend on the specific choice of $\tmtwo$.

Rather than completely forbidding classical reasoning principles,
we relax this condition so that classical principles
are allowed as long as they are useless, \ie inside the argument
of an application $\clasapp{\tm}{\tmtwo}$.
Furthermore, the invariant over terms of the form
$\claslamp{(\var:\typ\NN)}{\tm}$, requesting that $\var\notin\fv{\tm}$,
can also be relaxed, in such a way that $\var$ is allowed to occur in $\tm$
as long as all of its free occurrences are useless.
Formally:

\begin{definition}[Intuitionistic terms]
A subterm of a term $\tm$ is said to be {\bf useless}
if it lies inside the argument of a positive weak elimination.
More precisely, given a term $\tm = \gctxof{\tmtwo}$,
we say that the subterm $\tmtwo$ under the context $\gctx$
is useless if and only if there exist contexts $\gctx_1,\gctx_2$
and a term $\tmthree$ such that $\gctx$ is of the form
$\gctx_1\ctxof{\clasapp{\tmthree}{\gctx_2\ctxof{\ctxhole}}}$.
A subterm of $\tm$ is {\bf useful} if it is not useless.
A term is said to be {\bf intuitionistic} if and only if the
two following conditions hold:
\begin{enumerate}
\item
  {\bf Useless negative eliminations ($\Eandn$, $\Eimpn$, $\Enotn$, $\Ealln$).}
  There are no useful subterms of any of the following forms:
  $\casen{\tm}{(\var:\typ\NN)}{\tmtwo}{(\vartwo:\typtwo\NN)}{\tmthree}$,
  $\colamn{\tm}{(\var:\typ\PP)}{(\vartwo:\typtwo\NN)}{\tmtwo}$,
  $\negen{\tm}$,
  $\optn{\tm}{\btyp}{\var}{\tmtwo}$.
\item
  {\bf Useless negative counterfactuals.}
  In every useful subterm of the form $\claslamp{(\var:\typ\NN)}{\tm}$,
  there are no useful occurrences of $\var$ in $\tm$.
\end{enumerate}
\end{definition}

\subparagraph{The $\lambdaPRJ$ type system}
The type system $\lambdaPRJ$
is defined by imposing the restriction
on $\lambdaPRK$ that terms be intuitionistic.
More precisely,
we say that a judgment $\judg{\tctx}{\tm}{\ev}$
holds in $\PRJ$,
and in this case we write $\judgPRJ{\tctx}{\tm}{\ev}$,
if the judgment holds in $\PRK$ and furthermore $\tm$
is an intuitionistic term.
We also write $\logPRJ{\ev_1,\hdots,\ev_n}{\evtwo}$
if there exists a term $\tm$ such that
$\judgPRJ{\var_1:\ev_1,\hdots,\var_n:\ev_n}{\tm}{\evtwo}$.

\begin{example}
The weak variant of the law of excluded middle, $(\typ\lor\neg\typ)\PP$,
can be proven in $\PRK$. For example, if we define $\lemP{\typ}$:
\[
  \!\!
  {\small
    \begin{array}{r@{\,}c@{\,}l}
    \lemP{\typ} & \eqdef &
      \claslamp{(\var:(\typ\lor\neg\typ)\NN)}{
        \inip[2]{
          \claslamp{(\vartwo:\neg\typ\NN)}{
            \negip{
              \projin[1]{
                \clasapn{
                  \var
                }{
                  \lemPinner{\vartwo}{\typ}
                }
              }
            }
          }
        }
      }
    \\
    &&\text{where }
    \lemPinner{\vartwo}{\typ} \eqdef
      \claslamp{(\under:(\typ\lor\neg\typ)\NN)}{
        \inip[1]{
          \claslamp{(\varthree:\typ\NN)}{
            (\abs{
              \typ\pp
            }{
              \vartwo
            }{
              \claslamp{(\under:\neg\typ\NN)}{
                \negip{
                  \varthree
                }
              }
            })
          }
        }
      }
    \end{array}
  }
\]
we can note that $\judgPRK{}{\lemP{\typ}}{(\typ\lor\neg\typ)\PP}$ holds.
However, $\lemP{\typ}$ is not intuitionistic,
due to the fact that
{\em there is a useful occurrence of the negative counterfactual $\var$}.
\end{example}

The intuitionistic fragment is stable by reduction:
\begin{proposition}[Subject reduction for $\PRJ$]
Let $\judgPRJ{\tctx}{\tm}{\ev}$ and $\tm \to \tmtwo$.
Then $\judgPRJ{\tctx}{\tmtwo}{\ev}$.
\end{proposition}
\begin{proof}
If $\varset$ is a set of variables,
we say that a term $\tm$ is {\em $\varset$-intuitionistic}
if it is intuitionistic and, furthermore,
it has no {\em useful} free occurrences of variables in $\varset$.
We write $\judgPRJ[\varset]{\tctx}{\tm}{\ev}$,
if the judgment is derivable in $\PRK$ and $\tm$ is $\varset$-intuitionistic.
The statement of subject reduction is generalized as follows:
$\judgPRJ[\varset]{\tctx}{\tm}{\ev}$ and $\tm \to \tmtwo$,
then $\judgPRJ[\varset]{\tctx}{\tmtwo}{\ev}$.

The interesting case is the $\beta$ rule for positive weak proofs,
$\clasapp{(\claslamp{\var:\typ\NN}{\tm})}{\tmtwo} \toa{\ruleBClasP} \tm\sub{\var}{\tmtwo}$.
By hypothesis, $\judgPRJ[\varset]{\tctx}{\clasapp{(\claslamp{\var:\typ\NN}{\tm})}{\tmtwo}}{\typ\pp}$.
This judgment can only be derived from the $\Icp$ rule,
so $\judgPRJ[\varset]{\tctx,\var:\typ\NN}{\tm}{\typ\pp}$.
Moreover, $\var$ is a negative counterfactual,
so there cannot be useful free occurrences of $\var$ in $\tm$,
which means that $\judgPRJ[\varset\cup\set{\var}]{\tctx,\var:\typ\NN}{\tm}{\typ\pp}$.
On the other hand, $\tmtwo$ lies inside a positive application,
so it is not necessarily $X$-intuitionistic,
\ie we only know $\judgPRK{\tctx}{\tmtwo}{\typ\NN}$.
The key observation is that all the copies of $\tmtwo$ on the
right-hand side $\tm\sub{\var}{\tmtwo}$ must be useless,
because all the occurrences of $\var$ in $\tm$ are useless.
More precisely,
from $\judgPRJ[\varset\cup\set{\var}]{\tctx,\var:\typ\NN}{\tm}{\typ\pp}$
and $\judgPRK{\tctx}{\tmtwo}{\typ\NN}$
one concludes $\judgPRJ[\varset]{\tctx}{\tm\sub{\var}{\tmtwo}}{\ev}$
by induction on $\tm$.
\end{proof}

Reasoning principles in $\PRJ$ differ from those of $\PRK$.
For example, if $\ev$ is a weak formula,
a sequent $\judgPRK{\tctx,\var:\ev}{\tm}{\evtwo}$ valid in $\PRK$
can always be {\em contraposed} to a sequent of the form
$\judgPRK{\tctx,\vartwo:\evtwo\OP}{\tm'}{\ev\OP}$.
The analogous of this contraposition principle in $\PRJ$
depends on the sign of $\ev$.
If $\ev$ is positive, \ie $\ev = \typ\PP$
the sequent $\judgPRJ{\tctx,\var:\typ\PP}{\tm}{\evtwo}$
can always be contraposed to $\judgPRJ{\tctx,\vartwo:\evtwo\OP}{\tm'}{\typ\NN}$.
But if $\ev$ is negative, \ie $\ev = \typ\NN$
the sequent $\judgPRJ{\tctx,\var:\typ\NN}{\tm}{\evtwo}$
can only be contraposed to $\judgPRJ{\tctx,\vartwo:\evtwo\OP}{\tm'}{\typ\PP}$
if there are no useful occurrences of $\var$ in $\tm$.

The following theorem is an analog of~\rthm{lambdaPRK_refinement}
for $\lambdaPRJ$. \SeeAppendix{We omit the proof for lack of space;
see Sections~\ref{section:appendix:prj_principles},
\ref{section:appendix:prjv_definition}, \ref{section:appendix:prj_refinement}
in the appendix.}

\begin{theorem}[Intuitionistic refinement]
\lthm{lambdaPRJ_refinement}
$\log{\typ\PP_1,\hdots,\typ\PP_n}{\typtwo\PP}$ holds in $\lambdaPRJ$
if and only if
$\log{\typ_1,\hdots,\typ_n}{\typtwo}$ holds in intuitionistic
second-order logic.
\end{theorem}

\section{Canonicity}
\lsec{canonicity}

In sequent calculus and natural deduction,
an {\em indirect proof} (\eg with cuts),
can always be mechanically converted into
a {\em canonical proof} (\eg cut-free), in which the justification
for the conclusion is immediately available,
as is known from the
works of Gentzen~\cite{gentzen1935untersuchungen}
and Prawitz~\cite{prawitz1965natural}.
Its philosophical importance is that the validity of an indirect proof
can thus be justified by understanding it as a notation for describing
a canonical proof. A practical consequence is that
an explicit witness may be extracted from a proof of existence.

In this section, we formulate a {\bf canonicity} result
strengthening those of~\cite{BFlics21}.
We start by introducing some nomenclature.
{\em Neutral terms} ($\neu,\hdots$) and
{\em normal terms} ($\nf,\hdots$) are
given by
$\neu ::= 
       \var
  \mid \strongabs{\ev}{\neu}{\nf}
  \mid \strongabs{\ev}{\nf}{\neu}
  \mid \colorand{\projipn{\neu}}
  \mid \colorand{\casepn{\neu}{\var}{\nf}{\var}{\nf}}
  \mid
       \colorimp{\appn{\neu}{\nf}}
  \mid \colorimp{\colampn{\neu}{\var}{\vartwo}{\nf}}
  \mid \colorneg{\negepn{\neu}}
  \mid \colorall{\apptpn{\neu}{\typ}}
  \mid \colorall{\optpn{\neu}{\var}{\btyp}{\nf}}
  \mid \clasappn{\neu}{\nf}
$
and
$
  \nf ::=
       \neu
  \mid \colorand{\pairpn{\nf}{\nf}}
  \mid \colorand{\inipn{\nf}}
  \mid \colorimp{\lampn{\var}{\nf}}
  \mid \colorimp{\copairpn{\nf}{\nf}}
  \mid \colorneg{\negipn{\nf}}
  \mid \colorall{\lamtpn{\btyp}{\nf}}
  \mid \colorall{\patpn{\typ}{\nf}}
  \mid \claslampn{\var}{\nf}
$.
Terms built with an introduction rule:
$\colorand{\pairpn{\tm}{\tmtwo}}$,
$\colorand{\inipn{\tm}}$,
$\colorimp{\lampn{\var}{\tm}}$,
$\colorimp{\copairpn{\tm}{\tmtwo}}$,
$\colorneg{\negipn{\tm}}$,
$\colorall{\lamtpn{\btyp}{\tm}}$,
$\colorall{\patpn{\typ}{\tm}}$,
$\claslampn{\var}{\tm}$
are called {\em canonical}. Then
(\SeeAppendix{see~\rthm{appendix:canonicity} in the appendix for details}):

\begin{theorem}[Canonicity]
\lthm{canonicity}
\quad
\begin{enumerate}
\item
  If $\judgPRK{}{\tm}{\ev}$,
  then $\tm$ reduces to a canonical normal form $\nf$
  such that $\judgPRK{}{\nf}{\ev}$.
\item
  If $\judgPRK{}{\tm}{\ev}$, where $\ev$ is weak,
  then a canonical normal form $\nf$ can be effectively found
  such that $\judgPRK{}{\claslampn{(\var:\ev\OP)}{\nf}}{\ev}$
\end{enumerate}
\end{theorem}

Note that this canonicity theorem applies to closed terms only,
so there is no need to include commutative conversions,
such as
$\clasapp{\casep{\var}{\vartwo}{\tm}{\varthree}{\tmtwo}}{\tmthree}
\to
\casep{\var}{\vartwo}{\clasapp{\tm}{\tmthree}}{\varthree}{\clasapp{\tmtwo}{\tmthree}}$,
to unblock redexes.
The preceding canonicity result extends and strengthens~Thm.~35 of~\cite{BFlics21}.
The first part of the theorem confirms the intuition,
mentioned in the introduction,
that canonical proofs of strong propositions
are always constructed with an introduction rule
for the corresponding logical connective,
while canonical proofs of weak propositions
proceed by {\em reductio ad absurdum}.
For example, if $\judg{}{\tm}{(\typ_1\land\typ_2)\pp}$
is the {\em strong} proof of a conjunction,
its normal form must be a pair $\pairp{\tm_1}{\tm_2}$
and we know that $\judg{}{\tm_i}{\typ_i\PP}$ must hold for $i \in \set{1,2}$.
On the other hand, if $\judg{}{\tmtwo}{(\typ_1\land\typ_2)\pp}$
is the {\em weak} proof of a conjunction,
we can only assure that
its normal form is of the form $\claslamp{(\var:(\typ_1\land\typ_2)\NN)}{\tmtwo'}$,
where $\judg{\var:(\typ_1\land\typ_2)\NN}{\tmtwo'}{(\typ_1\land\typ_2)\pp}$.
The second part of the theorem provides the stronger guarantee
that, in such case, one can compute a {\em canonical} $\tmtwo''$
such that
$\judg{\var:(\typ_1\land\typ_2)\NN}{\tmtwo''}{(\typ_1\land\typ_2)\pp}$,
so in particular one has that $\tmtwo'' = \pairp{\tmtwo_1}{\tmtwo_2}$
and that
$\judg{\var:(\typ_1\land\typ_2)\NN}{\tmtwo_i}{\typ_i\PP}$
for all $i \in \set{1,2}$.

Canonicity can also
be used to obtain a (weak) form of {\em disjunctive property}.
In particular, from $\judg{}{\tm}{(\typ_1\lor\typ_2)\PP}$
one can always find an $i \in \set{1,2}$ and a term
$\judg{}{\claslamp{(\var:(\typ_1\lor\typ_2)\NN)}{\inip{\tm'}}}{(\typ_1\lor\typ_2)\PP}$
such that $\judg{\var:(\typ_1\lor\typ_2)\NN}{\tm'}{\typ_i\PP}$.
Similarly, a (weak) form of {\em witness extraction} can
be obtained:
given $\judg{}{\tm}{(\ex{\btyp}{\typ})\PP}$
one can always find a term
$\judg{}{\claslamp{(\var:(\ex{\btyp}{\typ})\NN)}{\patp{\typtwo}{\tm'}}}{(\ex{\btyp}{\typ})\PP}$
where $\judg{\var:(\ex{\btyp}{\typ})\NN}{\tm'}{\typ\sub{\btyp}{\typtwo}\PP}$.

Furthermore,
canonicity provides a purely syntactic proof of the consistency of $\PRK$\footnote{Another
way to prove consistency is using \rthm{lambdaPRK_refinement},
noting that $\logPRK{}{\btyp\pp}$ implies $\logNK{}{\btyp}$.}.
Note, for example, that if $\btyp$ is a base type,
there is no canonical term $\tm$ such that $\judgPRK{}{\tm}{\btyp\pp}$.

\section{Conclusion}
\lsec{conclusion}

In this paper we have extended the $\lambdaPRK$-calculus of~\cite{BFlics21}
to incorporate implication and co-implication,
as well as second-order quantifiers.
From the logical point of view, this extension of $\lambdaPRK$
{\bf refines classical second order logic} (\rthm{lambdaPRK_refinement}).
From the computational point of view, it is
{\bf confluent}~(\rthm{lambdaPRK_confluence})
and {\bf strongly normalizing}~(\rsec{strong_normalization}).
These ingredients constitute a computational interpretation for
second-order classical logic.
We have identified a well-behaved subset of the system,
called $\lambdaPRJ$, that 
{\bf refines intuitionistic second-order logic} (\rthm{lambdaPRJ_refinement}).
We have also formulated a {\bf canonicity}~(\rthm{canonicity})
result that strengthens results of previous works.
One noteworthy property of $\lambdaPRK$
is that both typing and reduction rules are fully symmetric
with respect to the operation
that flips signs and exchanges the roles of dual connectives,
while still being confluent.

\subparagraph{Related work}
The ``$\Icp$''/``$\Icn$'' rules in $\lambdaPRK$
encode a primitive variant of {\em consequentia mirabilis} ($(\neg\typ\to\typ)\to\typ$),
while the ``$\lambda$'' rule in
Barbanera and Berardi's calculus~\cite{symmetric-Barbanera-berardi},
as well as the ``$\mu$'' rule in
Parigot's $\lambda\mu$-calculus~\cite{lambdamu-parigot},
encode a primitive variant of {\em double negation elimination} ($\neg\neg\typ\to\typ$).
To prove $\typ$ using double negation elimination
one may assume $\neg\typ$ and then provide a proof of $\bot$.
This proof {\em cannot be canonical},
as there are no introduction rules for the empty type.
This motivates that we instead rely on {\em consequentia mirabilis}.

Strong normalization proofs are often based on reducibility candidates.
Yamagata proves strong normalization for second-order
formulations of classical calculi~\cite{Yamagata02,Yamagata04},
via reducibility candidates.
Our proof is inspired by ideas known from the literature
of logical relations and biorthogonality:
for instance,
the notions of {\em orthogonal} \RCs and {\em closure} of a \RC
can be traced back to Krivine's work on classical realizability~\cite{krivine2009},
Pitts' $\top\top$-closed logical relations~\cite{Pitts00},
and related notions (see \eg~\cite{DownenJA20}).
The challenging aspect of $\lambdaPRK$ is the mutually recursive dependency
between $\typ\PP$ and $\typ\NN$,
for which
our key reference is Mendler's work~\cite{mendler1991inductive}.

The problem of finding a good calculus for
classical logic has not been unquestionably settled.
Current proof assistants based on type theory, such as \textsc{Coq},
allow classical reasoning by postulating axioms with no
computational content, which breaks canonicity.
An established classical calculus
is Parigot's $\lambda\mu$,
whose metatheory has been thoroughly developed;
see for instance~\cite{DBLP:journals/mscs/Groote98,david2001lambdamu,DBLP:journals/tcs/Laurent03,saurin2005separation,DBLP:conf/csl/Saurin08,pedrot2016classical,DBLP:journals/lmcs/KesnerV19,DBLP:conf/csl/KesnerBV20}.
One difference between $\lambda\mu$ and $\lambdaPRK$
is that $\lambdaPRK$ computational rules are based on the standard
operation of substitution, while $\lambda\mu$ is based on an {\em ad hoc}
substitution operator.
Another difference is that the embedding of the
$\lambda$-calculus into $\lambda\mu$ is an inclusion,
whereas the embedding into $\lambdaPRK$ is much more convoluted.

Another established classical calculus
is Curien and Herbelin's~\cite{Curien00theduality}
$\bar{\lambda}\mu\tilde{\mu}$, whose study is also quite
mature; see for instance~\cite{Polonovski04,Herbelin08,DoughertyGL08,Bakel10,CurienM10,ariola2011classical,AriolaDHNS12,Miquey19}.
One difference between $\bar{\lambda}\mu\tilde{\mu}$ and $\lambdaPRK$
is that $\bar{\lambda}\mu\tilde{\mu}$ is not confluent unless
a reduction strategy is fixed in the presence of a specific critical
pair, whereas $\lambdaPRK$ is orthogonal.
Another difference is that $\bar{\lambda}\mu\tilde{\mu}$ is derived
from a proof term assignment for classical sequent calculus,
while $\lambdaPRK$ is defined in natural deduction style with
four forms of judgment (given by the modes $\typ\pp$, $\typ\nn$,
$\typ\PP$, $\typ\NN$).
Munch-Maccagnoni~\cite{Munch-Maccagnoni09} proposes a classical calculus
by {\em polarizing} Curien and Herbelin's calculus,
in such a way that the reduction strategy becomes determined by the polarities.

As mentioned in the introduction, $\lambdaPRK$ is related to Nelson's
constructible falsity~\cite{nelson1949constructible}.
Parigot~\cite{Parigot91}
studies {\em free deduction}, a system for classical logic
in which natural deduction and sequent calculus can both be embedded.
Rumfitt~\cite{10.2307/2660026} proposes {\em bilateral} logical systems,
in which assertion and denial judgments, with dual rules, are
formulated.
Zeilberger~\cite{Zeilberger08} studies a polarized logical system with
proofs and refutations distinguishing between verificationist and pragmatist
connectives.
These systems, however, do not distinguish between weak and strong
propositions, and they do not have rules analogous to $\Icpn/\Ecpn$.

\subparagraph{Future work}
The merely {\em logical} correspondence between $\lambdaPRK$ and
other classical calculi is immediate, given the fact that $\lambdaPRK$ refines
second-order classical logic~(\rthm{lambdaPRK_refinement}).
However, it is not obvious what their relation is from
the {\em computational} point of view.

In order to be able to build programming languages and proof assistants
based on the principles of $\lambdaPRK$, it would be convenient to
study dependently typed extensions of the system. It is not {\em a priori}
clear what such an extension would look like.

It is known that, in classical logic, {\em reductio ad absurdum} can be
postponed, in such a way that it is used at most once,
as the last rule in the derivation~\cite{DBLP:journals/jsyml/Seldin86,DBLP:journals/sLogica/GuerrieriN19}.
It would be interesting to see if this result can be reproduced in $\PRK$.




\bibliography{biblio}

\newpage
\appendix
\section{Second-Order Natural Deduction}
\lsec{appendix:natural_deduction}

\begin{definition}[Second-Order Natural Deduction]
Formulas are given by:
\[
  \begin{array}{rrl}
  \typ & ::=  & \btyp
           \mid \bot
           \mid \colorand{\typ \land \typ}
           \mid \colorand{\typ \lor \typ}
           \mid \colorimp{\typ \imp \typ}
           \mid \colorimp{\typ \coimp \typ}
           \mid \colorneg{\neg\typ}
           \mid \colorall{\all{\btyp}{\typ}}
           \mid \colorall{\ex{\btyp}{\typ}}
  \end{array}
\]
The intuitionistic second-order natural deduction system
$\NJ$ is given by the following inference rules.

\begin{fragBox}{Basic rules}
\[
\indrule{\NDAx}{
  \emptyPremise
}{
  \tctx, \typ \vdash \typ
}
\indrule{\NDExpl}{
  \tctx \vdash \bot
}{
  \tctx \vdash \typ
}
\]
\end{fragBox}

\begin{fragBox}{Conjunction and disjunction}
\[
\indrule{\colorand{\NDIand}}{
  \tctx \vdash \typ
  \HS
  \tctx \vdash \typtwo
}{
  \tctx \vdash \typ \land \typtwo
}
\indrule{\colorand{\NDEand}}{
  \tctx \vdash \typ_1 \land \typ_2
}{
  \tctx \vdash \typ_i
}
\]
\[
\indrule{\colorand{\NDIor}}{
  \tctx \vdash \typ_i
}{
  \tctx \vdash \typ_1 \lor \typ_2
}
\indrule{\colorand{\NDEor}}{
  \tctx \vdash \typ \lor \typtwo
  \HS
  \tctx, \typ \vdash \typthree
  \HS
  \tctx, \typtwo \vdash \typthree
}{
  \tctx \vdash \typthree
}
\]
\end{fragBox}

\begin{fragBox}{Implication and co-implication}
\[
\indrule{\colorimp{\NDIimp}}{
  \tctx, \typ \vdash \typtwo
}{
  \tctx \vdash \typ \imp \typtwo
}
\indrule{\colorimp{\NDEimp}}{
  \tctx \vdash \typ \imp \typtwo
  \HS
  \tctx \vdash \typ
}{
  \tctx \vdash \typtwo
}
\]
\[
\indrule{\colorimp{\NDIcoimp}}{
  \tctx \vdash \neg\typ
  \HS
  \tctx \vdash \typtwo
}{
  \tctx \vdash \typ \coimp \typtwo
}
\indrule{\colorimp{\NDEcoimp}}{
  \tctx \vdash \typ \coimp \typtwo
  \HS
  \tctx,\neg\typ,\typtwo \vdash \typthree
}{
  \tctx \vdash \typthree
}
\]
\end{fragBox}

\begin{fragBox}{Negation}
\[
\indrule{\colorneg{\NDInot}}{
  \tctx, \typ \vdash \bot
}{
  \tctx \vdash \lnot \typ
}
\indrule{\colorneg{\NDEnot}}{
  \tctx \vdash \neg\typ
  \HS
  \tctx \vdash \typ
}{
  \tctx \vdash \bot
}
\]
\end{fragBox}

\begin{fragBox}{Second-order quantification}
\[
  \indrule{\colorall{\NDIall}}{
    \tctx \vdash \typ
    \HS
    \btyp \not\in \fv{\tctx}
  }{
    \tctx \vdash \all{\btyp}{\typ}
  }
  \indrule{\colorall{\NDEall}}{
    \tctx \vdash \all{\btyp}{\typtwo}
  }{
    \tctx \vdash \typtwo\sub{\btyp}{\typ}
  }
\]
\[
  \indrule{\colorall{\NDIex}}{
    \tctx \vdash \typtwo\sub{\btyp}{\typ}
  }{
    \tctx \vdash \ex{\btyp}{\typtwo}
  }
  \indrule{\colorall{\NDEex}}{
    \tctx \vdash \ex{\btyp}{\typ}
    \HS
    \tctx,\typ \vdash \typtwo
    \HS
    \btyp \not\in \fv{\tctx,\typtwo}
  }{
    \tctx \vdash \typtwo
  }
\]
\end{fragBox}

The classical second-order natural deduction system $\NK$
is obtained by extending $\NJ$ with the law of excluded middle:
\[
\indrule{\NDLem}{
}{
  \tctx \vdash \typ \lor \neg\typ
}
\]
Furthermore, the following rules are admissible in $\NJ$ and $\NK$.
\[
  \indrule{\NDW}{
    \log{\tctx}{\typ}
  }{
    \log{\tctxtwo}{\typ}
  }
  \indrule{\NDCut}{
    \log{\tctx,\typtwo}{\typ}
    \HS
    \log{\tctx}{\typtwo}
  }{
    \log{\tctx}{\typ}
  }
\]
\bigskip

\noindent
We write $\logNJ{\tctx}{\typ}$ (respectively, $\logNK{\tctx}{\typ}$)
if the sequent $\log{\tctx}{\typ}$ holds in $\NJ$ (respectively, $\NK$).
\end{definition}

\newpage
\section{Subject Reduction for $\lambdaPRK$}
\lsec{appendix:subject_reduction}

\begin{lemma}
\llem{lambdaTwoC:subT_admissible}
The substitution rule is admissible in $\lambdaPRK$:

\[
  \indrule{\SubT}{
    \tctx \vdash \tm : \ev
  }{
    \tctx\sub{\btyp}{\typ} \vdash \tm\sub{\btyp}{\typ} : \ev\sub{\btyp}{\typ}
  }
\]
\end{lemma}

\begin{proposition}[Subject Reduction]
  If $\judgPRK{\tctx}{\tm}{\ev}$
  and $\tm \toa{} \tmtwo$, then $\judgPRK{\tctx}{\tmtwo}{\ev}$.
\end{proposition}
\begin{proof}
This extends the proof of Prop.~24 from~\cite{BFlics21}.
The proof follows the usual methodology, by case analysis on
the derivation of the reduction step.
We focus on the more interesting cases, namely the
second-order quantifiers. We only study the positive cases,
the negative cases being symmetric.

\medskip \noindent{\bf $\ruleAppT$ rule.}
Let:
\[
  \prooftree 
  \[
    \[
      \typingDerivation
      \justifies 
      \tctx \vdash \tm : \typtwo\PP 
    \]
    \HS 
    \btyp \not\in \fv{\tctx}
    \justifies \tctx \vdash \lamtp{\btyp}{\tm} : \all{\btyp}{\typtwo}\pp
    \using{\indrulename{\Iallp}}
    \thickness=0.05em
  \]
  \justifies \tctx \vdash \apptp{(\lamtp{\btyp}{\tm})}{\typ} : \typtwo\PP\sub{\btyp}{\typ}
  \thickness=0.05em
  \using{\indrulename{\Eallp}}
  \endprooftree
\]
Then:
\[
   \prooftree 
   \[
     \typingDerivation
     \justifies 
     \tctx \vdash \tm : \typtwo\PP 
   \]
   \justifies 
   \tctx \vdash \tm\sub{\btyp}{\typ} : \typtwo\PP\sub{\btyp}{\typ}
   \thickness=0.05em
   \using{\indrulename{\SubT}}
   \endprooftree
\]
Note that $\tctx = \tctx\sub{\btyp}{\typ}$ holds
since $\btyp \not\in \fv{\tctx}$.

\medskip \noindent{\bf $\ruleOpen$ rule.}
Let:
\[
  \prooftree 
  \[
    \[
      \typingDerivation
      \justifies
      \tctx \vdash \tm : \typtwo\PP\sub{\btyp}{\typ}
    \]
    \justifies
    \tctx \vdash \patpn{\typ}{\tm} : (\ex{\btyp}{\typtwo})\pp
    \using{\indrulename{\Iexp}}
  \]
  \[
    \typingDerivation'
    \justifies
    \tctx,\var:\typtwo\PP \vdash \tmtwo : \ev
  \]
  \justifies \tctx \vdash \openp{\btyp}{\var}{\patp{\typ}{\tm}}{\tmtwo} : \ev
  \thickness=0.05em
  \using{\indrulename{\Eexp}}
  \endprooftree
\]
where $\btyp \not\in \fv{\tctx,\ev}$. Then:
\[
  \prooftree 
  \[
    \typingDerivation
    \justifies
    \tctx \vdash \tm : \typtwo\PP\sub{\btyp}{\typ}
  \]
  \[
    \[
      \typingDerivation'
      \justifies
      \tctx, \var : \typtwo\PP \vdash \tmtwo : \ev             
    \]
    \justifies
    \tctx, \var:\typtwo\PP\sub{\btyp}{\typ} \vdash \tmtwo\sub{\btyp}{\typ} : \ev
    \using{\indrulename{\SubT}}
  \]
  \justifies 
  \tctx \vdash \tmtwo\sub{\btyp}{\typ}\sub{\var}{\tm} : \ev
  \thickness=0.05em
  \using{\indrulename{\CCut}}
  \endprooftree
\]

\medskip \noindent{\bf $\ruleAbsLamPairT$ rule.}
Let:
\[
  \prooftree
  \[
    \[
      \typingDerivation
      \justifies
      \tctx \vdash \tm : \typtwo\PP
    \]
    \justifies
    \tctx \vdash \lamtp{\btyp}{\tm} : \all{\btyp}{\typtwo}\pp
    \using{\indrulename{\Iallp}}
  \]
  \[
    \[
      \typingDerivation'
      \justifies
      \tctx \vdash \tmtwo : \typtwo\NN\sub{\btyp}{\typ}
    \]
    \justifies
    \tctx \vdash \patn{\typ}{\tmtwo} : \all{\btyp}{\typtwo}\nn
    \using{\indrulename{\Ialln}}
  \]
  \justifies \tctx \vdash \strongabs{\ev}{(\lamtp{\btyp}{\tm})}{\patn{\typ}{\tmtwo}} : \ev
      \thickness=0.05em
      \using{\indrulename{\Abs}}
      \endprooftree
\]
where $\btyp \not\in \ftv{\tctx}$.
Then:
\[
  \prooftree
  \[
    \[
      \typingDerivation
      \justifies
      \tctx \vdash \tm : \typtwo\PP
    \]
    \justifies
    \tctx \vdash \tm\sub{\btyp}{\typ} : \typtwo\PP\sub{\btyp}{\typ}
    \using{\indrulename{\SubT}}
  \]
  \[
    \typingDerivation'
    \justifies
    \tctx \vdash \tmtwo : \typtwo\NN\sub{\btyp}{\typ}
  \]
  \justifies \tctx \vdash  \abs{\ev}{\tm\sub{\btyp}{\typ}}{\tmtwo} : \ev
  \thickness=0.05em
  \using{\indrulename{\Abs}}
  \endprooftree
\]

\medskip \noindent{\bf $\toa{\ruleAbsPairLamT}$ rule.}
  Symmetric to $\toa{\ruleAbsLamPairT}$.
\end{proof}

\newpage
\section{$\lambdaPRK$ Refines Classical Second-Order Logic}
\lsec{appendix:prk_refinement}

The proof that $\lambdaPRK$ refines classical second-order logic
is an extension of analogous results for
classical propositional logic in~\cite{BFlics21}.
The proof of~\rthm{lambdaPRK_refinement}
is split into two lemmas:
{\bf Classical~Conservativity}~(\rlem{lambdaPRK_conservativity})
proves the implication $1 \implies 2$, and
{\bf Classical~Embedding}~(\rlem{lambdaPRK_embedding})
proves the implication $2 \implies 1$.

\begin{lemma}[Classical Conservativity]
\llem{lambdaPRK_conservativity}
If $\log{\ev_1,\hdots,\ev_n}{\evtwo}$ holds in $\lambdaPRK$,
then $\log{\classem{\ev_1},\hdots,\classem{\ev_n}}{\classem{\evtwo}}$
holds in $\lambdaPRK$.
\end{lemma}
\begin{proof}
By induction on the derivation of $\log{\ev_1,\hdots,\ev_n}{\evtwo}$,
it suffices to show that all rules of $\lambdaPRK$
are, from the logical point of view,
admissible in classical second-order logic.
The formal target system is
the classical second-order natural deduction system $\NK$
described in \rsec{appendix:natural_deduction}
of the appendix.

Recall that double negation $\typ \leftrightarrow \neg\neg\typ$,
as well as all De~Morgan's laws hold in classical second-order logic;
for example
$\logNK{}{\neg(\typ\land\typtwo) \leftrightarrow \neg\typ\lor\neg\typtwo}$
and
$\logNK{}{\neg\all{\btyp}{\typ} \leftrightarrow \ex{\btyp}{\neg\typ}}$.
Hence the proof can be reduced to studying the rules for positive
connectives, noting that the proof for the
negative rule for the dual connective must be symmetric:
\begin{enumerate}
\item $\Ax$:
  Let $\judgPRK{\tctx}{\var}{\ev}$
  where $(\var:\ev) \in \tctx$.
  Then $\logNK{\classem{\tctx}}{\classem{\ev}}$ by \NDAx.
\item $\Abs$:
  Note that $\logNK{}{\typ \imp \neg\typ \imp \classem{\ev}}$.
\item $\Icp$:
  Let $\judgPRK{\tctx}{\claslamp{(\var:\typ\NN)}{\tm}}{\typ\PP}$
  be derived from $\judgPRK{\tctx,\var:\typ\NN}{\tm}{\typ\pp}$.
  By \ih we have that $\logNK{\classem{\tctx},\neg\typ}{\typ}$
  which, {\bf classically}, implies
  $\logNK{\classem{\tctx}}{\typ}$.
\item $\Ecp$:
  Let $\judgPRK{\tctx}{\clasapp{\tm}{\tmtwo}}{\typ\pp}$
  be derived from $\judgPRK{\tctx}{\tm}{\typ\PP}$
  and $\judgPRK{\tctx}{\tmtwo}{\typ\NN}$.
  By \ih on the first premise, we have that $\logNK{\classem{\tctx}}{\typ}$,
  as required.
\item $\colorand{\Iandp}$:
  Note that $\logNK{}{\typ \imp \typtwo \imp (\typ\land\typtwo)}$.
\item $\colorand{\Eandp}$:
  Note that $\logNK{}{(\typ_1\land\typ_2) \imp \typ_i}$.
\item $\colorand{\Iorp}$:
  Note that $\logNK{}{\typ_i \imp (\typ_1\lor\typ_2)}$.
\item $\colorand{\Eorp}$:
  Let
  $\judgPRK{\tctx}{\casep{\tm}{\var:\typ\PP}{\tmtwo}{\vartwo:\typtwo\PP}{\tmthree}}{\ev}$
  be derived from $\judgPRK{\tctx}{\tm}{(\typ \lor \typtwo)\pp}$
  and $\judgPRK{\tctx,\var:\typ\PP}{\tmtwo}{\ev}$
  and $\judgPRK{\tctx,\vartwo:\typtwo\PP}{\tmthree}{\ev}$.
  By \ih, $\logNK{\classem{\tctx}}{\typ \lor \typtwo}$
  and $\logNK{\classem{\tctx},\typ}{\classem{\ev}}$
  and $\logNK{\classem{\tctx},\typtwo}{\classem{\ev}}$,
  which imply $\logNK{\classem{\tctx}}{\classem{\ev}}$
  by \NDEor.
\item $\colorimp{\Iimpp}$:
  Let $\judgPRK{\tctx}{\lamp{(\var:\typ\PP)}{\tm}}{(\typ\imp\typtwo)\pp}$
  be derived from $\judgPRK{\tctx,\var:\typ\PP}{\tm}{\typtwo\PP}$.
  By \ih, $\logNK{\classem{\tctx},\typ}{\typtwo}$,
  which implies $\logNK{\classem{\tctx}}{\typ\imp\typtwo}$
  by \NDIimp.
\item $\colorimp{\Eimpp}$:
  Note that $\logNK{}{(\typ\imp\typtwo) \imp \typ \imp \typtwo}$.
\item $\colorimp{\Icoimpp}$:
  Note that $\logNK{}{\neg\typ \imp \typtwo \imp (\typ\coimp\typtwo)}$.
\item $\colorimp{\Ecoimpp}$:
  Note that $\logNK{}{(\typ\coimp\typtwo) \imp (\neg\typ \imp \typtwo \imp \typthree) \imp \typthree}$.
\item $\colorneg{\Inotp}$:
  Note that $\logNK{}{\neg\typ \imp \neg\typ}$.
  For the dual rule $\Inotn$,
  note that $\logNK{}{\typ \imp \neg\neg\typ}$.
\item $\colorneg{\Enotp}$:
  Note that $\logNK{}{\neg\typ \imp \neg\typ}$.
  For the dual rule $\Enotn$,
  note that $\logNK{}{\neg\neg\typ \imp \typ}$,
  which holds {\bf classically}.
\item $\colorall{\Iallp}$:
  Let
    $\judgPRK{\tctx}{\lamtp{\btyp}{\tm}}{(\all{\btyp}{\typ})\pp}$
  be derived from
    $\judgPRK{\tctx}{\tm}{\typ\PP}$,
  where $\btyp \notin \ftv{\tctx}$.
  By \ih, $\logNK{\classem{\tctx}}{\typ}$.
  Moreover, note that $\btyp \notin \ftv{\classem{\tctx}}$
  since $\btyp \notin \ftv{\tctx}$.
  Hence by $\NDIall$ we have that
  $\logNK{\classem{\tctx}}{\all{\btyp}{\typ}}$.
\item $\colorall{\Eallp}$:
  Let
    $\judgPRK{\tctx}{\apptp{\tm}{\typ}}{\typtwo\PP\sub{\btyp}{\typ}}$
  be derived from
    $\judgPRK{\tctx}{\tm}{(\all{\btyp}{\typtwo})\pp}$.
  By \ih, $\logNK{\classem{\tctx}}{\all{\btyp}{\typtwo}}$,
  which implies
    $\logNK{\classem{\tctx}}{\typtwo\sub{\btyp}{\typ}}$
  by $\NDEall$.
\item $\colorall{\Iexp}$:
  Let
    $\judgPRK{\tctx}{\patp{\typ}{\tm}}{(\ex{\btyp}{\typtwo})\pp}$
  be derived from
    $\judgPRK{\tctx}{\tm}{\typtwo\PP\sub{\btyp}{\typ}}$.
  By \ih,
    $\logNK{\classem{\tctx}}{\typtwo\sub{\btyp}{\typ}}$,
  which implies
    $\logNK{\classem{\tctx}}{\ex{\btyp}{\typtwo}}$
  by \NDIex.
\item $\colorall{\Eexp}$:
  Let
    $\judgPRK{\tctx}{\optp{\tm}{\btyp}{\var}{\tmtwo}}{\ev}$
  be derived from
    $\judgPRK{\tctx}{\tm}{(\ex{\btyp}{\typ})\pp}$
  and
    $\judgPRK{\tctx,\var:\typ\PP}{\tmtwo}{\ev}$,
  where
    $\btyp \not\in \ftv{\tctx,\ev}$.
  By \ih we have that
    $\logNK{\classem{\tctx}}{\ex{\btyp}{\typ}}$
  and
    $\logNK{\classem{\tctx},\typ}{\classem{\ev}}$.
  Moreover, note that $\btyp\notin\ftv{\classem{\tctx},\classem{\ev}}$.
  Hence by $\NDEex$ we have $\logNK{\classem{\tctx}}{\classem{\ev}}$,
  as required.
\end{enumerate}
\end{proof}

Before proving embedding, we recall some auxiliary lemmas
from~\cite{BFlics21}:

\begin{lemma}[Excluded middle and non-contradiction]
\llem{prk_lemP_and_lemN}
For every pure type $\typ$,
there exist terms $\lemP{\typ}$ and $\lemN{\typ}$
such that:
\begin{enumerate}
\item {\bf Excluded middle.}
  $\judgPRK{\tctx}{\lemP{\typ}}{(\typ \lor \neg\typ)\PP}$
\item {\bf Non-contradiction.}
  $\judgPRK{\tctx}{\lemN{\typ}}{(\typ \land \neg\typ)\NN}$
\end{enumerate}
\end{lemma}
\begin{proof}
For {\bf excluded middle}, take:
  \[
  \!\!
  {\small
    \begin{array}{r@{\,}c@{\,}l}
    \lemP{\typ} & \eqdef &
      \claslamp{(\var:(\typ\lor\neg\typ)\NN)}{
        \inip[2]{
          \claslamp{(\vartwo:\neg\typ\NN)}{
            \negip{
              \projin[1]{
                \clasapn{
                  \var
                }{
                  \lemPinner{\vartwo}{\typ}
                }
              }
            }
          }
        }
      }
    \\
    \lemPinner{\vartwo}{\typ} & \eqdef &
      \claslamp{(\under:(\typ\lor\neg\typ)\NN)}{
        \inip[1]{
          \claslamp{(\varthree:\typ\NN)}{
            (\abs{
              \typ\pp
            }{
              \vartwo
            }{
              \claslamp{(\under:\neg\typ\NN)}{
                \negip{
                  \varthree
                }
              }
            })
          }
        }
      }
    \end{array}
  }
  \]
For {\bf non-contradiction}, take:
  \[
  \!\!
  {\small
    \begin{array}{r@{\,}c@{\,}l}
    \lemN{\typ} & \eqdef &
      \claslamn{(\var:(\typ \land \neg\typ)\PP)}{
        \inin[2]{
          \claslamn{(\vartwo:\neg\typ\PP)}{
            \negin{
              \projip[1]{
                \clasapp{
                  \var
                }{
                  \lemNinner{\vartwo}{\typ}
                }
              }
            }
          }
        }
      }
    \\
    \lemNinner{\vartwo}{\typ} & \eqdef &
      \claslamn{(\under:(\typ \land \neg\typ)\PP)}{
        \inin[1]{
          \claslamn{(\varthree:\typ\PP)}{
            (\abs{
              \typ\nn
            }{
              \vartwo
            }{
              \claslamn{(\under:\neg\typ\PP)}{
                \negin{
                  \varthree
                }
              }
            })
          }
        }
      }
    \end{array}
  }
  \]
\end{proof}

\begin{lemma}[Classical contraposition]
If $\ev$ is weak and
$\judgPRK{\tctx, \var : \ev}{\tm}{\evtwo}$,
there is a term $\ccontrapose{\var}{\vartwo}{\tm}$
such that
$\judgPRK{\tctx, \vartwo : \evtwo\OP}{\ccontrapose{\var}{\vartwo}{\tm}}{\ev\OP}$.
\end{lemma}
\begin{proof}
As in~\cite{BFlics21}, it suffices to take:
\[
  \ccontrapose{\var}{\vartwo}{\tm} \eqdef
    \begin{cases}
      \claslamn{(\var:\typ\PP)}{
        (\abs{\typ\nn}{
          \tm
        }{ 
          \vartwo
        })
      }
      & \text{if $\ev = \typ\PP$}
    \\
      \claslamn{(\var:\typ\NN)}{
        (\abs{\typ\pp}{
          \tm
        }{ 
          \vartwo
        })
      }
      & \text{if $\ev = \typ\NN$}
    \end{cases}
\]
\end{proof}

\begin{lemma}[Weak negation]
\quad
\begin{enumerate}
\item {\bf Weak negation introduction}:
  If $\judgPRK{\tctx}{\tm}{\typ\NN}$,
  there is a term $\negiP{\tm}$
  such that $\judgPRK{\tctx}{\negiP{\tm}}{(\neg\typ)\PP}$.
\item {\bf Weak negation elimination}:
  If $\judgPRK{\tctx}{\tm}{(\neg\typ)\PP}$,
  there is a term $\negeP{\tm}$
  such that $\judgPRK{\tctx}{\negeP{\tm}}{\typ\NN}$.
\end{enumerate}
\end{lemma}
\begin{proof}
For {\bf weak negation introduction},
let $\judgPRK{\tctx}{\tm}{\typ\NN}$ and take:
\[
  \negiP{\tm} \eqdef
  \claslamp{(\under:(\neg\typ)\NN)}{
    \negip{
      \tm
    }
  }
\]
For {\bf weak negation elimination},
let $\judgPRK{\tctx}{\tm}{(\neg\typ)\PP}$ and take:
\[
  \negeP{\tm} \eqdef
  \claslamn{(\var:\typ\PP)}{
    \clasappn{
      \negep{
        (\clasapp{
          \tm
        }{
          \claslamn{\under:(\neg\typ)\PP}{
            \negin{
              \var
            }
          }
        })
      }
    }{
      \var
    }
  }
\]

\end{proof}

\begin{lemma}[Classical Embedding]
\llem{lambdaPRK_embedding}
If $\log{\ev_1,\hdots,\ev_n}{\evtwo}$ holds in
classical second-order logic,
then $\log{\ev\PP_1,\hdots,\ev\PP_n}{\evtwo}$ holds in $\lambdaPRK$.
\end{lemma}
\begin{proof}
We proceed by induction of the derivation
of $\logNK{\typ_1,\hdots,\typ_n}{\typtwo}$
in the classical second-order natural deduction system $\NK$.
\begin{enumerate}
\item \NDAx:
  Let $\log{\typ_1,\hdots,\typ_n}{\typ_i}$ be derived from the $\NDAx$ rule.
  Then $\judgPRK{\var_1:\typ_1\PP,\hdots,\var_n:\typ_n\PP}{\var_i}{\typ_i\PP}$
  by the $\Ax$ rule.
\item \colorand{\NDIand}:
  Suppose by \ih that
  $\judgPRK{\tctx}{\tm}{\typ\PP}$ and $\judgPRK{\tctx}{\tmtwo}{\typtwo\PP}$.
  Take:
  \[
    \pairc{\tm}{\tmtwo}
    \eqdef
    \claslamp{\under:(\typ\land\typtwo)\NN}{
      \pairp{\tm}{\tmtwo}
    }
  \]
  Then $\judgPRK{\tctx}{\pairc{\tm}{\tmtwo}}{(\typ\land\typtwo)\PP}$.
\item \colorand{\NDEand}:
  Suppose by \ih that $\judgPRK{\tctx}{\tm}{(\typ_1\land\typ_2)\PP}$
  and let $i \in \set{1,2}$.
  Let $\projic{\tm}$ be the term:
  \[
    \claslamp{(\var:\typ_i\NN)}{
      \clasapp{
        \projip{
          \clasapp{
            \tm
          }{
            \claslamn{(\under:(\typ_1\land\typ_2)\PP)}{
              \inin{\var}
            }
          }
        }
      }{
        \var
      }
    }
  \]
  Then $\judgPRK{\tctx}{\projic{\tm}}{\typ_i\PP}$.
\item \colorand{\NDIor}:
  Suppose by \ih that $\judgPRK{\tctx}{\tm}{\typ_i\PP}$
  for some $i \in \set{1,2}$. Take:
  \[
    \inic{\tm} \eqdef
    \claslamp{(\under:(\typ_1\lor\typ_2)\NN)}{
      \inip{
        \tm
      }
    }
  \]
  Then $\judgPRK{\tctx}{\inic{\tm}}{(\typ_1\lor\typ_2)\PP}$.
\item \colorand{\NDEor}:
  Suppose by \ih that
  $\judgPRK{\tctx}{\tm}{(\typ \lor \typtwo)\PP}$
  and
  $\judgPRK{\tctx,\var:\typ\PP}{\tmtwo}{\typthree\PP}$
  and
  $\judgPRK{\tctx,\var:\typtwo\PP}{\tmthree}{\typthree\PP}$.
  Let $\casec{\tm}{(\var:\typ\PP)}{\tmtwo}{(\var:\typtwo\PP)}{\tmthree}$
  be the term:
  \[
    \claslamp{(\vartwo:\typthree\NN)}{
        (\clasapp{
          \casep{
            (\clasapp{
              \tm
            }{
              \claslamn{(\under:(\typ\lor\typtwo)\PP)}{
                \pairn{
                  \ccontrapose{\var}{\vartwo}{\tm}
                }{
                  \ccontrapose{\var}{\vartwo}{\tmtwo}
                }
              }
            })
          }{(\var:\typ\PP)}{
            \tmtwo
          }{(\var:\typtwo\PP)}{
            \tmthree
          }
      }{
        \vartwo
      })
    }
  \]
  Then
  $\judgPRK{\tctx}{\casec{\tm}{(\var:\typ\PP)}{\tmtwo}{(\var:\typtwo\PP)}{\tmthree}}{\typthree\PP}$.
\item
  \colorimp{\NDIimp}:
  Suppose by \ih that $\judgPRK{\tctx,\var:\typ\PP}{\tm}{\typtwo\PP}$.
  Take:
  \[
    \lamc{(\var:\typ\PP)}{\tm} \eqdef
    \claslamp{(\under:(\typ\to\typtwo)\NN)}{
      \lamp{(\var:\typ\PP)}{
        \tm
      }
    }
  \]
  Note that 
  $\judgPRK{\tctx}{\lamc{(\var:\typ\PP)}{\tm}}{(\typ\imp\typtwo)\PP}$.
\item \colorimp{\NDEimp}:
  Suppose by \ih that $\judgPRK{\tctx}{\tm}{(\typ \imp \typtwo)\PP}$
  and $\judgPRJ{\tctx}{\tmtwo}{\typ\PP}$.
  Let $\appc{\tm}{\tmtwo}$ be the term:
  \[
    \claslamp{(\var:\typtwo\NN)}{
      (\clasapp{
        \app{
          \clasapp{
            \tm
          }{
            (\claslamn{(\under:(\typ\to\typtwo)\PP)}{
              \copairn{
                \tmtwo
              }{
                \var
              }
            })
          }
        }{
          \tmtwo
        }
      }{
        \var
      })
    }
  \]
  Note that 
  $\judgPRK{\tctx}{\appc{\tm}{\tmtwo}}{\typtwo\PP}$.
\item \colorimp{\NDIcoimp}:
  Suppose by \ih that $\judgPRK{\tctx}{\tm}{(\neg\typ)\PP}$
  and $\judgPRK{\tctx}{\tmtwo}{\typtwo\PP}$.
  Take:
  \[
    \copairc{\tm}{\tmtwo}
    \eqdef
    \claslamp{(\under:(\typ\coimp\typtwo)\NN)}{
      \copairp{
        \negeP{\tm}
      }{
        \tmtwo
      }
    }
  \]
  Note that $\judgPRK{\tctx}{\copairc{\tm}{\tmtwo}}{(\typ\coimp\typtwo)\PP}$.
\item \colorimp{\NDEcoimp}:
  Suppose by \ih that 
  $\judgPRK{\tctx}{\tm}{(\typ \coimp \typtwo)\PP}$
  and
  $\judgPRK{\tctx,\var:(\neg\typ)\PP,\vartwo:\typtwo\PP}{\tmtwo}{\typthree\PP}$.
  Let:
  \[
    \colamc{\tm}{\var}{\vartwo}{\tmtwo}
    \eqdef
    \claslamp{\varthree:\typthree\NN}{
      (\clasapp{
        \colamp{
          (\clasapp{
            \tm
          }{
            (\claslamn{\under}{
              \lamn{\var_0:\typ\NN}{
                \claslamn{\vartwo:\typtwo\PP}{
                  (\abs{\typtwo\nn}{
                    \tmtwo'
                  }{
                    \varthree
                  })
                }
              }
            })
          })
        }{\var_0:\typ\NN}{\vartwo:\typtwo\PP}{
          \tmtwo'
        }
      }{
        \varthree
      })
    }
  \]
  where $\tmtwo' \eqdef \tmtwo\sub{\var}{\negiP{\var_0}}$.
  Note that $\judgPRK{\tctx}{\colamc{\tm}{\var}{\vartwo}{\tmtwo}}{\typthree\PP}$.
\item \colorneg{\NDInot}:
  Suppose by \ih that
  $\judgPRK{\tctx,\var:\typ\PP}{\tm}{\bot\PP}$.
  We encode falsity as $\bot \eqdef \btyp_0\land\neg\btyp_0$
  for some fixed base type $\btyp_0$.
  With this encoding of falsity,
  recall that $\judg{}{\lemN{\btyp_0}}{\bot\NN}$
  from \rlem{prk_lemP_and_lemN}.
  Take:
  \[
    \neglamc{\var:\typ\PP}{\tm} \eqdef
    \claslamp{\under:(\neg\typ)\NN}{
      \negip{(
        \ccontrapose{\var}{\vartwo}{\tm}
          \sub{\vartwo}{\lemN{\btyp_0}}
      )}
    }
  \]
  Note that
  $\judgPRK{\tctx}{\neglamc{\var:\typ\PP}{\tm}}{(\neg\typ)\PP}$.
\item \colorneg{\NDEnot}:
  Suppose by \ih that
  $\judgPRK{\tctx}{\tm}{(\neg\typ)\PP}$
  and
  $\judgPRK{\tctx}{\tmtwo}{\typ\PP}$.
  Take:
  \[
    \negapc{\tm}{\tmtwo} \eqdef
    \abs{\bot\PP}{
      \tm
    }{
      \claslamn{(\under:\typ\PP)}{
        \negin{
          \tmtwo
        }
      }
    }
  \]
  Note that $\judgPRK{\tctx}{\negapc{\tm}{\tmtwo}}{\bot\PP}$.
\item \colorall{\NDIall}:
  Suppose by \ih
  that $\judgPRK{\tctx}{\tm}{\typ\PP}$ with $\btyp \not\in \ftv{\tctx}$.
  Take:
  \[
    \lamtc{\btyp}{\tm} \eqdef
    \claslamp{(\under:(\all{\btyp}{\typ})\NN)}{
      \lamtp{\btyp}{
        \tm
      }
    }
  \]
  Note that
  $\judgPRK{\tctx}{\lamtc{\btyp}{\tm}}{(\all{\btyp}{\typ})\PP}$.
\item \colorall{\NDEall}:
  Suppose by \ih that
  $\judgPRK{\tctx}{\tm}{(\all{\btyp}{\typtwo})\PP}$.
  Let $\apptc{\tm}{\typ}$ be the term:
  \[
    \claslamp{(\var:(\typtwo\sub{\btyp}{\typ})\NN)}{
      (\clasapp{
        \apptp{
          \tm'
        }{
          \typ
        }
      }{
        \var
      })
    }
  \]
  where $\tm' \eqdef
          \clasapp{
            \tm
          }{
            \claslamp{(\under:(\all{\btyp}{\typtwo})\PP)}{
              \patn{
                \typ
              }{
                \var
              }
            }
          }$.
  Then $\judgPRK{\tctx}{\apptc{\tm}{\typ}}{\typtwo\sub{\btyp}{\typ}\PP}$.
\item \colorall{\NDIex}:
  Suppose by \ih that
  $\judgPRK{\tctx}{\tm}{(\typtwo\sub{\btyp}{\typ})\PP}$.
  Take:
  \[
    \patc{\typ}{\tm} \eqdef
    \claslamp{(\under:(\ex{\btyp}{\typtwo})\NN)}{
      \patp{\typ}{
        \tm
      }
    }
  \]
  Then $\judgPRK{\tctx}{\patc{\typ}{\tm}}{(\ex{\btyp}{\typtwo})\PP}$.
\item \colorall{\NDEex}:
  Suppose by \ih that
  $\judgPRK{\tctx}{\tm}{(\ex{\btyp}{\typ})\PP}$
  and
  $\judgPRK{\tctx,\var:\typ\PP}{\tmtwo}{\typtwo\PP}$
  with $\btyp \not\in \ftv{\tctx,\ev}$.
  Let $\optj{\tm}{\btyp}{\var}{\tmtwo}$ be the term:
  \[
    \claslamp{(\vartwo:\typtwo\NN)}{
      (\clasapp{
        \optp{
          \tm'
        }{\btyp}{\var}{
          \tmtwo
        }
      }{
        \vartwo
      })
    }
  \]
  where $\tm' \eqdef 
          \clasapp{
            \tm
          }{
            \claslamn{(\under:(\ex{\btyp}{\typ})\PP)}{
              \lamtn{\btyp}{
                \ccontrapose{\var}{\vartwo}{\tmtwo}
              }
            }
          }$.
  Then $\judgPRK{\tctx}{\optc{\tm}{\btyp}{\var}{\tmtwo}}{\typtwo\PP}$.
\item \NDExpl:
  Suppose by \ih that $\judgPRK{\tctx}{\tm}{\bot\PP}$.
  Then
  $\judgPRK{\tctx}{\abs{\typ\PP}{\tm}{\lemN{\btyp_0}}}{\typ\PP}$.
\item \NDLem:
  Note that $\judgPRK{}{\lemP{\btyp_0}}{(\typ\lor\neg\typ)\PP}$
  by \rlem{prk_lemP_and_lemN}.
\end{enumerate}
\end{proof}

\newpage
\section{Encoding Data Types in $\lambdaPRK$}
\lsec{appendix:bohm_encoding}

In this section we show, as examples,
that the Böhm--Berarducci encoding
for disjunction and existential quantification
in terms of universal quantification and implication
allows to prove the positive introduction rules
($\Iorp$ and $\Iexp$), as well as weak variants
of the elimination rules ($\Eorp$ and $\Eexp$).
Furthermore, these constructions simulate the corresponding
reduction rules ($\ruleBOrP$ and $\ruleBExP$).

\subsection{Encoding Disjunction}

Define
$\typ_1\lor\typ_2 \eqdef
 \all{\btyp}{((\typ_1\to\btyp)\to(\typ_2\to\btyp)\to\btyp)}$.
Then positive typing rules for disjunction,
analogous to $\Iorp$ and the weak variant of $\Eorp$,
are derivable, and their constructions simulate the $\ruleBOrP$ rule.
Let:
\[
  \begin{array}{rcl}
  X_{\typtwo}
      & \eqdef & (\typ_1\to\typtwo)\to(\typ_2\to\typtwo)\to\typtwo \\
  X   & \eqdef & (\typ_1\to\btyp)\to(\typ_2\to\btyp)\to\btyp \\
  X'  & \eqdef & (\typ_2\to\btyp)\to\btyp \\
  Y_i & \eqdef & \typ_i\to\btyp \\
  \end{array}
\]
\medskip

\noindent{\em $\bullet$ Typing rule $\Iorp$.}
  Let $\judg{\tctx}{\tm}{\typ_i\PP}$. Take:
  \[
    \begin{array}{r@{\,}c@{\,}l}
    \inip{\tm} & \eqdef &
      \lamtp{\btyp}{
        \claslamp{(\under:X\NN)}{
          \lamp{(\vartwo_1:Y_1\PP)}{
            \claslamp{(\under:X'\NN)}{
              \lamp{(\vartwo_2:Y_2\PP)}{
                \claslamp{(\varthree:\btyp\NN)}{
                  \tmthree
                }
              }
            }
          }
        }
      }
    \\
    \tmthree & \eqdef &
      \clasapp{
        \app{
          \clasapp{
            \vartwo_i
          }{
            (\claslamp{(\under:Y_i\PP)}{
              \copair{
                \tm
              }{
                \varthree
              }
            })
          }
        }{
          \tm
        }
      }{
        \varthree
      }
    \end{array}
  \]
  Then $\judg{\tctx}{\inip{\tm}}{(\typ_1\lor\typ_2)\pp}$.
\medskip

\noindent{\em $\bullet$ Weak variant of the typing rule $\Eorp$.}
  Let $\judg{\tctx}{\tm}{(\typ_1\lor\typ_2)\pp}$
  and $\judg{\tctx,\var_i:\typ_i\PP}{\tmtwo_i}{\typtwo\PP}$
  for each $i \in \set{1,2}$.
  Take:
  \[
    \begin{array}{rcl}
    \casep{\tm}{\var_1}{\tmtwo_1}{\var_2}{\tmtwo_2} & \eqdef &
      \claslamp{(\vartwo:\typtwo\NN)}{
        (\clasapp{
          \app{
            \clasapp{
              \app{
                \tm'
              }{
                \tmfour_1
              }
            }{
              \tmfive
            }
          }{
            \tmfour_2
          }
        }{
          \vartwo
        })
      }
    \\
    \tm' & \eqdef &
      \clasapp{
        \apptp{
          \tm
        }{
          \typtwo
        }
      }{
        (\claslamn{(\under:X_\typtwo\PP)}{
          \copairn{
            \tmfour_1
          }{
            \tmfive
          }
        })
      }
    \\
    \tmfour_i & \eqdef &
      \claslamp{(\under:(\typ_i\to\typtwo)\NN)}{
        \lam{(\var_i:\typ_i\PP)}{
          \tmtwo_i
        }
      }
    \\
    \tmfive & \eqdef &
      \claslamn{(\under:((\typ_2\to\typtwo)\to\typtwo)\PP)}{
        \copairn{
          \tmfour_2
        }{
          \vartwo
        }
      }
    \end{array}
  \]
  Then
  $\judg{\tctx}{\casep{\tm}{\var_1}{\tmtwo_1}{\var_2}{\tmtwo_2}}{\typtwo\PP}$.
\medskip

\noindent {\em $\bullet$ Computation rule $\ruleBOrP$.}
  Let $\judg{\tctx}{\tm}{\typ_i\PP}$
  and $\judg{\tctx,\var_i:\typ_i\PP}{\tmtwo_i}{\typtwo\PP}$
  for each $i \in \set{1,2}$.
  Then:
  \[
    \begin{array}{ll}
    &
      \casep{\ini{\tm}}{\var_1}{\tmtwo_1}{\var_2}{\tmtwo_2}
    \\
    \rtoa{} &
      \claslamp{(\vartwo:\typtwo\NN)}{
        (\clasapp{
          \app{
            \clasapp{
              \tmfour_i
            }{
              (\claslamp{\under}{
                \copair{
                  \tm
                }{
                  \vartwo
                }
              })
            }
          }{
            \tm
          }
        }{
          \vartwo
        })
      }
    \\
    \rtoa{} &
      \claslamp{(\vartwo:\typtwo\NN)}{
        (\clasapp{
          \tmtwo_i\sub{\var_i}{\tm}
        }{
          \vartwo
        })
      }
    \\
    \toa{\ruleEta} &
      \tmtwo_i\sub{\var_i}{\tm}
    \end{array}
  \]

\subsection{Encoding Existential Quantification}

Define
$\ex{\btyp}{\typ} \eqdef
 \all{\btyptwo}{(\all{\btyp}{(\typ\imp\btyptwo)}\imp\btyptwo)}$.
(Note that in our notation quantifiers are of higher precedence than
binary connectives).
Then positive typing rules for existential quantification,
analogous to $\Iexp$ and a weak variant of $\Eexp$,
are derivable, and their constructions simulate the $\ruleBExP$ rule.
Let:
\[
  \begin{array}{rcl}
  X_\typtwo & \eqdef & (\all{\btyp}{(\typ\to\typtwo)})\to\typtwo \\
  Y_\typtwo & \eqdef & \all{\btyp}{(\typ\to\typtwo)} \\
  X         & \eqdef & X_{\btyptwo} \\
  Y         & \eqdef & Y_{\btyptwo} \\
  \end{array}
\]

\noindent{\em $\bullet$ Typing rule $\Iexp$.}
  Let $\judg{\tctx}{\tm}{(\typ\sub{\btyp}{\typtwo})\PP}$. Take:
  \[
    \begin{array}{r@{\,}c@{\,}l}
      \patp{\typtwo}{\tm} & \eqdef &
        \lamtp{\btyptwo}{
          \claslamp{(\under:X\NN)}{
            \lamp{(\vartwo:Y\PP)}{
              \claslamp{(\varthree:\btyptwo\NN)}{
                (\clasapp{
                  \tm'
                }{
                  \varthree
                })
              }
            }
          }
        }
    \\
      \tm' & \eqdef &
        \app{
          \clasapp{
            \apptp{
              \clasapp{
                \vartwo
              }{
                (\claslamn{(\under:Y\PP)}{
                  \patn{\typtwo}{
                    \tmthree
                  }
                })
              }
            }{
              \typtwo
            }
          }{
            \tmthree
          }
        }{
          \tm
        }
    \\
      \tmthree & \eqdef &
        \claslamn{(\under:(\typ\sub{\btyp}{\typtwo}\imp\btyptwo)\PP)}{
          \copairn{
            \tm
          }{
            \varthree
          }
        }
    \end{array}
  \]
  Then $\judg{\tctx}{\patp{\typtwo}{\tm}}{(\ex{\btyp}{\typ})\pp}$.
\medskip

\noindent{\em $\bullet$ Weak variant of the typing rule $\Eexp$.}
  Let $\judg{\tctx}{\tm}{(\ex{\btyp}{\typ})\pp}$
  and $\judg{\tctx,\var:\typ\PP}{\tmtwo}{\typthree\PP}$
  with $\btyp \notin \ftv{\tctx,\typthree\PP}$.
  Take:
  \[
    \begin{array}{r@{\,}c@{\,}l}
      \optp{\tm}{\btyp}{\var}{\tmtwo} & \eqdef &
        \claslamp{(\varthree:\typthree\NN)}{
          (\clasapp{
            \app{
              \clasapp{
                \apptp{
                  \tm
                }{
                  \typthree
                }
              }{
                \tmfour
              }
            }{
              \tmfive
            }
          }{
            \varthree
          })
        }
    \\
    \tmfour & \eqdef &
      \claslamn{(\under:X_\typthree\PP)}{
        \copairn{
          \tmfive
        }{
          \vartwo
        }
      }
    \\
    \tmfive & \eqdef &
      \claslamp{(\under:Y_\typthree\NN)}{
        \lamtp{\btyp}{
          \claslamp{(\under:(\typ\to\typthree)\NN)}{
            \lamp{(\var:\typ\PP)}{
              \tmtwo
            }
          }
        }
      }
    \end{array}
  \]
  Then $\judg{\tctx}{\optp{\tm}{\btyp}{\var}{\tmtwo}}{\typthree\PP}$.
\medskip

\noindent {\em $\bullet$ Computation rule $\ruleBExP$.}
  Let $\judg{\tctx}{\tm}{(\typ\sub{\btyp}{\typtwo})\PP}$
  and $\judg{\tctx,\var:\typ\PP}{\tmtwo}{\typthree\PP}$
  with $\btyp \notin \ftv{\tctx,\typthree\PP}$.
  Then:
  \[
    \begin{array}{ll}
    &
      \optp{\patp{\typtwo}{\tm}}{\btyp}{\var}{\tmtwo}
    \\
    \rtoa{} &
      \claslamp{(\varthree:\typthree\NN)}{
        (\clasapp{
          \app{
            \clasapp{
              \apptp{
                \clasapp{
                  \tmfive
                }{
                  (\claslamn{\under}{
                    \patn{\typtwo}{
                      \tmthree
                    }
                  })
                }
              }{
                \typtwo
              }
            }{
              \tmthree
            }
          }{
            \tm
          }
        }{
          \varthree
        })
      }
    \\
    \rtoa{} &
      \claslamp{(\varthree:\typthree\NN)}{
        (\clasapp{
          \tmtwo\sub{\btyp}{\typtwo}\sub{\var}{\tm}
        }{
          \varthree
        })
      }
    \\
    \toa{\ruleEta} &
      \tmtwo\sub{\btyp}{\typtwo}\sub{\var}{\tm}
    \end{array}
  \]

\newpage
\section{Strong Normalization}
\lsec{appendix:strong_normalization}

\subsection{Properties of Reducibility Candidates}

We recall the following well-known fact from order theory (see \eg \cite[Lem.~2.30]{davey1990}):
\begin{lemma}
\llem{all_glbs_implies_complete_lattice}
If $(A,\leq)$ is such that $\meet{B}$
exists for every $B \subseteq A$,
then $(A,\leq)$ is a complete lattice.
\end{lemma}

\begin{proposition}
\lprop{appendix:reducibility_candidates_form_complete_lattice}
The set $\RCSet$ forms a complete lattice
ordered by inclusion $\subseteq$.
\end{proposition}
\begin{proof}
By~\rlem{all_glbs_implies_complete_lattice},
it suffices to show that every subset $\rcfamily \subseteq \RCSet$ has
a greatest lower bound.
Let $\rcfamily \subseteq \RCSet$ be a family of \RCs.
Take $\meet\rcfamily := \bigcap\set{\rc \ST \rc \in \rcfamily}$,
where, if $\rcfamily$ is empty,
this means that $\meet\rcfamily = \SNTerms$.
Note that $\meet\rcfamily$ is a \RC and a greatest lower bound:
\begin{enumerate}
\item {\bf Closed by reduction.}
  Let $\atm \in \meet\rcfamily$ and $\atm \tou \atmtwo$.
  By definition, $\atm \in \rc$ for every $\rc \in \rcfamily$.
  Since each $\rc \in \rcfamily$ is a \RC,
  we have that $\atmtwo \in \rc$ for every $\rc \in \rcfamily$.
  Hence $\atmtwo \in \meet\rcfamily$.
\item {\bf Complete.}
  Let $\atm \in \SNTerms$ be such that that
  $\forall \atmtwo \in \CanTerms.
       ((\atm \rtou \atmtwo) \implies \atmtwo \in \meet\rcfamily)$.
  Then note that, for each $\rc \in \rcfamily$ one has that
  $\forall \atmtwo \in \CanTerms.
       ((\atm \rtou \atmtwo) \implies \atmtwo \in \rc)$
  holds.
  Hence, since each $\rc \in \rcfamily$ is a \RC,
  $\atm \in \rc$ for every $\rc \in \rcfamily$.
  Hence $\atm \in \meet\rcfamily$, as required.
\item {\bf Greatest lower bound.} 
  By standard properties of the set-theoretic intersection,
  $\meet\rcfamily = \bigcap\set{\rc \ST \rc \in \rcfamily}$
  is the greatest lower bound of $\rcfamily$ with respect to inclusion.
\end{enumerate}
\end{proof}

\begin{remark}
The top element of $\RCSet$
is given by $\rctop := \meet\emptyset = \SNTerms$.
The bottom element of $\RCSet$
is given by the set of terms that do not reduce to a canonical
term,
that is,
$\rcbot := \set{\atm \in \SNTerms \ST \forall \atmtwo \in \CanTerms.\ \neg(\atm \rtou \atmtwo)}$.
Note that $\rcbot$ is a \RC and the least element:
\begin{enumerate}
\item {\bf Closed by reduction.}
  Let $\atm \in \rcbot$ and $\atm \tou \atmtwo$.
  Then $\atmtwo$ does not reduce to a canonical term, for otherwise
  $\atm$ would reduce to a canonical term.
  Hence $\atmtwo \in \rcbot$.
\item {\bf Complete.}
  Let $\atm \in \SNTerms$ be such that
  $(\atm \rtou \atmtwo) \implies \atmtwo \in \rcbot$ holds
  for each canonical term $\atmtwo \in \CanTerms$.
  Observe that $\atmtwo \notin \rcbot$,
  given that a canonical term always reduces to itself.
  Hence, by the contrapositive, we conclude that
  $\neg(\atm \rtou \atmtwo)$.
  Hence we have that
  $\forall\atmtwo\in\CanTerms.\,\neg(\atm \rtou \atmtwo)$,
  that is $\atm \in \rcbot$, as required.
\item {\bf Least element.} 
  We argue that $\rcbot \subseteq \rc$
  for each $\rc \in \RCSet$.
  Indeed, let $\atm \in \rcbot$,
  and let us show that $\atm \in \rc$.
  Note that $\atm$ does not reduce to any canonical term, so
  since $\rc$ is complete, we have that $\atm \in \rc$.
\end{enumerate}
\end{remark}

\begin{remark}
If $\atm \in \CanTerms$ is a canonical term,
$\atm \in \rclosure{X}$ implies $\atm \in X$.
\end{remark}

\begin{lemma}[Operations are well-defined on \RCs]
\llem{appendix:operations_on_reducibility_candidates_well_defined}
If $\rc_1,\rc_2$ are \RCs and $\set{\rc_i}_{i\inI}$ is a set of \RCs,
then $\rc_1 \rctimes \rc_2$, $\rc_1 \rcplus \rc_2$, $\rc_1 \rcimp \rc_2$,
$\rc_1 \rccoimp \rc_2$, $\rcneg\rc$, $\rcall{i\inI}{\rc_i}$,
and $\rcex{i\inI}{\rc_i}$ are well-defined \RCs. 
\end{lemma}
\begin{proof}
First, we claim that, for any set $X \subseteq \CanTerms$,
its closure $\rclosure{X}$ is a reducibility candidate:
\begin{enumerate}
\item Closed by reduction:
  Let $\atm \in \rclosure{X}$ and $\atm \tou \atm'$.
  We claim that $\atm' \in \rclosure{X}$.
  Indeed, if $\atm'$ reduces to a canonical term,
  \ie $\atm' \rtou \atmtwo \in \CanTerms$
  then also $\atm \rtou \atmtwo$,
  so $\atmtwo \in X$, as required.
\item Complete:
  Let $\atm \in \SNTerms$ be such that
  for every $\atmtwo \in \CanTerms$
  we have that $\atm \rtou \atmtwo$ implies $\atmtwo \in \rclosure{X}$.
  We claim that $\atm \in \rclosure{X}$.
  Indeed, suppose that $\atm \rtou \atmtwo \in \CanTerms$.
  Then by hypothesis $\atmtwo \in \rclosure{X}$.
  But $\atmtwo$ is a canonical term and it reduces to itself
  in zero steps, so $\atmtwo \in X$, as required.
\end{enumerate}
Second, we prove, for each each operation, that the resulting set is
a \RC, \ie closed by reduction and complete.
If $\rc_1,\rc_2 \in \RCSet$ are reducibility candidates
then it is immediate to conclude that
$\rc_1\rctimes\rc_2$ is a reducibility candidate,
given that the product is defined as a closure.
Similarly for the sum, co-implication product,
negation, and indexed sum,
\ie if $\rc_1,\rc_2\in\RCSet$ then $\rc_1\rcplus\rc_2\in\RCSet$;
if $\rc_1,\rc_2\in\RCSet$ then $\rc_1\rccoimp\rc_2\in\RCSet$;
if $\rc\in\RCSet$ then $\rcneg\rc\in\RCSet$; and 
if $\set{\rc_i}_{i\inI} \subseteq \RCSet$ then $\rcex{i\inI}{\rc_i}\in\RCSet$.
The remaining operations are:
\begin{enumerate}
\item {\bf Arrow.}
  Let $\rc_1,\rc_2 \in \RCSet$.
  Then $\rc_1\rcimp\rc_2 \in \RCSet$:
  \begin{enumerate}
  \item Closed by reduction:
    Let $\atm \in \rc_1\rcimp\rc_2$ and $\atm \tou \atm'$.
    We claim that $\atm' \in \rc_1\rcimp\rc_2$.
    Indeed, let $\atmtwo \in \rc_1$,
    and let us check that $\ap{\atm'}{\atmtwo} \in \rc_2$.
    Note that, by definition, $\ap{\atm}{\atmtwo} \in \rc_2$,
    and moreover $\ap{\atm}{\atmtwo} \tou \ap{\atm'}{\atmtwo}$.
    Since $\rc_2$ is closed by reduction, $\ap{\atm'}{\atmtwo} \in \rc_2$,
    as required.
  \item Complete:
    Let $\atm \in \SNTerms$ be such that
    for every $\atmtwo \in \CanTerms$
    we have that $\atm \rtou \atmtwo$ implies $\atmtwo \in \rc_1\rcimp\rc_2$.
    We claim that $\atm \in \rc_1\rcimp\rc_2$.
    Indeed, let $\atmthree \in \rc_1$ and
    let us show that $\ap{\atm}{\atmthree} \in \rc_2$.
    Since $\rc_2$ is complete, it suffices to show that
    if $\atmtwo \in \CanTerms$ is a canonical term
    and $\ap{\atm}{\atmthree} \rtou \atmtwo$
    then $\atmtwo \in \rc_2$.
    Observe that any reduction
    $\ap{\atm}{\atmthree} \rtou \atmtwo \in \CanTerms$
    must be of the form
    $\ap{\atm}{\atmthree}
     \rtou \ap{(\lam{\var}{\atm'})}{\atmthree'}
     \tou \atm'\sub{\var}{\atmthree'}
     \rtou \atmtwo$
    with $\atm \rtou \lam{\var}{\atm'}$
    and $\atmthree \rtou \atmthree'$.
    By hypothesis, $\lam{\var}{\atm'} \in \rc_1\rcimp\rc_2$.
    Furthermore, since $\rc_1$ is closed by reduction,
    we have that $\atmthree' \in \rc_1$.
    Therefore $\ap{(\lam{\var}{\atm'})}{\atmthree'} \in \rc_2$.
    Finally, since $\rc_2$ is closed by reduction,
    we conclude that $\atmtwo \in \rc_2$, as required.
  \end{enumerate}
\item {\bf Indexed product.}
  Let $\set{\rc_i}_{i\inI} \subseteq \RCSet$.
  Then $\rcall{i\inI}{\rc_i} \in \RCSet$:
  \begin{enumerate}
  \item Closed by reduction:
    Let $\atm \in \rcall{i\inI}{\rc_i}$ and $\atm \tou \atm'$.
    We claim that $\atm' \in \rcall{i\inI}{\rc_i}$.
    Let $i \in I$ and note that $\apptu{\atm} \in \rc_i$
    and $\apptu{\atm} \tou \apptu{\atm'}$.
    Since $\rc_i$ is closed by reduction, $\apptu{\atm'} \in \rc_i$.
    Hence $\apptu{\atm'} \in \rc_i$ for arbitrary $i\inI$,
    which means that $\atm' \in \rcall{i\inI}{\rc_i}$.
  \item Complete:
    Let $\atm \in \SNTerms$ be such that
    for every $\atmtwo \in \CanTerms$
    we have that $\atm \rtou \atmtwo$ implies $\atmtwo \in \rcall{i\inI}{\rc_i}$.
    We claim that $\atm \in \rcall{i\inI}{\rc_i}$.
    Indeed, let $i \in I$ and let us show that $\apptu{\atm} \in \rc_i$.
    Since $\rc_i$ is complete, it suffices to show that
    if $\atmtwo \in \CanTerms$ is a canonical term
    and $\apptu{\atm} \rtou \atmtwo$
    then $\atmtwo \in \rc_i$.
    Observe that any reduction $\apptu{\atm} \rtou \atmtwo$
    must be of the form
    $\apptu{\atm} \rtou \apptu{(\lamtu{\atm'})} \tou \atm' \rtou \atmtwo$
    with $\atm \rtou \lamtu{\atm'}$.
    By hypothesis, $\lamtu{\atm'} \in \rcall{i\inI}{\rc_i}$.
    Therefore $\apptu{(\lamtu{\atm'})} \in \rc_i$.
    Finally, since $\rc_i$ is closed by reduction, 
    we conclude that $\atmtwo \in \rc_i$, as required.
  \end{enumerate}
\end{enumerate}
\end{proof}

\begin{lemma}[Arrow is order-reversing on the left]
\llem{rcimp_order_reversing}
\llem{appendix:rcimp_order_reversing}
Let $\rc_1$, $\rc'_1$, $\rc_2$, and $\rc_3$ denote
reducibility candidates.
\begin{enumerate}
\item
  If $\rc_1 \subseteq \rc'_1$ then
  $(\rc'_1\rcimp\rc_2) \subseteq (\rc_1\rcimp\rc_2)$.
\item
  If $\rc_1 \subseteq \rc'_1$ then
  $((\rc_1\rcimp\rc_2)\rcimp\rc_3) \subseteq
   ((\rc'_1\rcimp\rc_2)\rcimp\rc_3)$.
\end{enumerate}
\end{lemma}
\begin{proof}
Suppose that $\rc_1 \subseteq \rc'_1$,
let $\atm \in (\rc'_1\rcimp\rc_2)$,
and let us show that $\atm \in (\rc_1\rcimp\rc_2)$.
By definition, consider an arbitrary $\atmtwo \in \rc_1$
and let us show that $\ap{\atm}{\atmtwo} \in \rc_2$.
But $\atmtwo \in \rc'_1$, so in fact  
$\ap{\atm}{\atmtwo} \in \rc_2$.
The second item is an immediate consequence of the first.
\end{proof}

\begin{remark}
The set $\RCPerp$ is non-empty.
Note that $(\rcbot,\rcbot) \in \RCPerp$.
Indeed, if $\atm,\atmtwo\in\rcbot$
note that $\atm$ and $\atmtwo$ are both strongly normalizing.
Moreover, all the reducts of
$\strongabs{}{\atm}{\atmtwo}$
are of the form
$\strongabs{}{\atm'}{\atmtwo'}$
with $\atm \rtou \atm'$ and $\atmtwo \rtou \atmtwo'$,
because $\atm$ and $\atmtwo$ do not reduce to canonical
terms, so $\strongabs{}{\atm'}{\atmtwo'}$ does not have a
redex at the root.
\end{remark}

\begin{lemma}[Reducible terms are well-defined]
For each type $\ev$ and each environment $\env$,
the set $\red{\ev}{\env}$ is a reducibility candidate.
\end{lemma}
\begin{proof}
By induction on the measure $\#(\ev)$.
Most cases are straightforard by induction hypothesis,
using the fact that
operations on reducibility candidates ($\rctimes$, $\rcplus$, etc.)
are well defined~(\rlem{appendix:operations_on_reducibility_candidates_well_defined}).
The interesting cases are $\red{\typ\PP}{\env}$ and $\red{\typ\NN}{\env}$.
We study the positive case; the negative case is similar.

To see that
$\red{\typ\PP}{\env} =
  \lfp{\rc}{((\rc\rcimp\red{\typ\nn}{\env})\rcimp\red{\typ\pp}{\env})}
  \in \RCSet$
note that, by \ih,
$\red{\typ\nn}{\env} \in \RCSet$ and
$\red{\typ\pp}{\env} \in \RCSet$.
By the Knaster--Tarski theorem~(\rthm{knaster_tarski}),
to see that the least fixed point exists,
it suffices to show that the mapping
$f(\rc) = ((\rc\rcimp\red{\typ\nn}{\env})\rcimp\red{\typ\pp}{\env})$
is order-preserving.
This results from \rlem{rcimp_order_reversing}.
\end{proof}

\begin{lemma}[Irrelevance for reducible terms]
\llem{reducible_terms_irrelevance}
\llem{appendix:reducible_terms_irrelevance}
Let $\env,\env'$ be environments
that agree on all the free type variables of $\ev$.
More precisely,
suppose that
$\env(\btyp\pp) = \env'(\btyp\pp)$
and
$\env(\btyp\nn) = \env'(\btyp\nn)$
for every $\btyp \in \ftv{\ev}$.
Then $\red{\ev}{\env} = \red{\ev}{\env'}$.
\end{lemma}
\begin{proof}
Straightforward by induction on the measure $\#(\ev)$.
\end{proof}

\begin{lemma}[Substitution for reducible terms]
\llem{reducible_terms_substitution}
\llem{appendix:reducible_terms_substitution}
$\red{\ev\sub{\btyp}{\typ}}{\env} =
 \red{\ev}{
   \env\esub{\btyp}{\red{\typ\pp}{\env},\red{\typ\nn}{\env}}
 }$.
\end{lemma}
\begin{proof}
Straightforward by induction on the measure $\#(\ev)$,
resorting to the irrelevance lemma~(\rlem{reducible_terms_irrelevance})
in the cases of second-order quantifiers.
\end{proof}

\begin{lemma}[Reducible terms of weak type]
\llem{reducible_terms_of_classical_type}
\llem{appendix:reducible_terms_of_classical_type}
The following hold:
\begin{enumerate}
\item
  $\red{\typ\PP}{\env} = \red{\typ\NN}{\env} \rcimp \red{\typ\pp}{\env}$
\item
  $\red{\typ\NN}{\env} = \red{\typ\PP}{\env} \rcimp \red{\typ\nn}{\env}$
\end{enumerate}
\end{lemma}
\begin{proof}
We only show the first item; the second one is similar.
Let $f(\rc) = ((\rc\rcimp\red{\typ\nn}{\env})\rcimp\red{\typ\pp}{\env})$
and $g(\rc) = ((\rc\rcimp\red{\typ\pp}{\env})\rcimp\red{\typ\nn}{\env})$.
Recall that, by definition,
  $\red{\typ\PP}{\env} = \lfpF{f}$
and
  $\red{\typ\NN}{\env} = \gfpF{f}$.
To prove the equation of the first item
it suffices to show that 
$\red{\typ\NN}{\env} \rcimp \red{\typ\pp}{\env}$
is the least fixed point of $f$:
\begin{enumerate}
\item
  {\bf Fixed point.}
  \[
    \begin{array}{rcll}
    &&
      \red{\typ\NN}{\env} \rcimp \red{\typ\pp}{\env}
    \\
    & = &
      g(\red{\typ\NN}{\env}) \rcimp \red{\typ\pp}{\env}
    \\
    & = &
      ((\red{\typ\NN}{\env}\rcimp\red{\typ\pp}{\env})\rcimp\red{\typ\nn}{\env})
      \rcimp \red{\typ\pp}{\env}
    \\
    & = &
      f(\red{\typ\NN}{\env}\rcimp\red{\typ\pp}{\env})
    \end{array}
  \]
  The first equality is justified because $\red{\typ\NN}{\env} = \gfpF{g}$.
\item
  {\bf Least of the fixed points.}
  Suppose that $\rc_0 = f(\rc_0)$ is another fixed point of $f$,
  and let us show that
  $(\red{\typ\NN}{\env} \rcimp \red{\typ\pp}{\env}) \subseteq \rc_0$.
  Observe that $\rc_0\rcimp\red{\typ\nn}{\env}$ is a fixed point of $g$:
  \[
    \begin{array}{rcll}
    &&
      \rc_0\rcimp\red{\typ\nn}{\env}
    \\
    & = &
      f(\rc_0)\rcimp\red{\typ\nn}{\env}
    \\
    & = &
      ((\rc_0\rcimp\red{\typ\nn}{\env})\rcimp\red{\typ\pp}{\env})
      \rcimp\red{\typ\nn}{\env}
    \\
    & = &
      g(\rc_0\rcimp\red{\typ\nn}{\env})
    \end{array}
  \]
  Hence $(\rc_0\rcimp\red{\typ\nn}{\env}) \subseteq \gfpF{g}$,
  and we have:
  \[
    \begin{array}{rcll}
    &&
    \red{\typ\NN}{\env} \rcimp \red{\typ\pp}{\env}
    \\
    & = &
    \gfpF{g} \rcimp \red{\typ\pp}{\env}
    \\
    & \subseteq &
      (\rc_0\rcimp\red{\typ\nn}{\env})\rcimp\red{\typ\pp}{\env}
      & \text{by~\rlem{rcimp_order_reversing}}
    \\
    & = &
      f(\rc_0)
    \\
    & = &
      \rc_0
    \end{array}
  \]
\end{enumerate}
\end{proof}

\begin{lemma}[Reducible terms of opposite strong types are orthogonal]
\llem{red_opposite_orthogonal}
\llem{appendix:red_opposite_orthogonal}
$\rcperp{\red{\typ\pp}{\env}}{\red{\typ\nn}{\env}}$.
\end{lemma}
\begin{proof}
We proceed by induction on $\typ$:
\begin{enumerate}
\item {\bf Type variable, $\typ = \btyp$.}
  Then, indeed, we have that $\rcperp{\env(\btyp\pp)}{\env(\btyp\nn)}$
  because, by definition,
  an environment $\env$ maps each pair of type variables
  $\btyp\pp,\btyp\nn$ to orthogonal reducibility candidates.
\item {\bf Conjunction, $\typ = \typ_1\land\typ_2$.}
  Let $\atm \in \red{(\typ_1\land\typ_2)\pp}{\env}$
  and $\atmtwo \in \red{(\typ_1\land\typ_2)\nn}{\env}$.
  Since these sets are reducibility candidates,
  we have that $\atm,\atmtwo \in \SNTerms$.
  To show that $(\strongabs{}{\atm}{\atmtwo}) \in \SNTerms$,
  we proceed by induction on $\snsize{\atm} + \snsize{\atmtwo}$.
  It suffices to show that all the one-step
  reducts of $\strongabs{}{\atm}{\atmtwo}$
  are strongly normalizing, \ie that if
  $(\strongabs{}{\atm}{\atmtwo}) \tou \atmthree$
  then $\atmthree \in \SNTerms$.
  There are three subcases for the step:
  \begin{enumerate}
  \item
    {\bf Step internal to $\atm$.}
    That is,
    $(\strongabs{}{\atm}{\atmtwo}) \tou (\strongabs{}{\atm'}{\atmtwo})$
    with $\atm \tou \atm'$.
    Since reducibility candidates are closed by reduction,
    we have that
    $\atm' \in \red{(\typ_1\land\typ_2)\pp}{\env}$
    and moreover $\snsize{\atm} > \snsize{\atm'}$.
    By the inner \ih, $(\strongabs{}{\atm'}{\atmtwo}) \in \SNTerms$.
  \item
    {\bf Step internal to $\atmtwo$.}
    That is,
    $(\strongabs{}{\atm}{\atmtwo}) \tou (\strongabs{}{\atm'}{\atmtwo})$
    with $\atmtwo \tou \atmtwo'$.
    Since reducibility candidates are closed by reduction,
    we have that
    $\atmtwo' \in \red{(\typ_1\land\typ_2)\nn}{\env}$
    and moreover $\snsize{\atmtwo} > \snsize{\atmtwo'}$.
    By the inner \ih, $(\strongabs{}{\atm}{\atmtwo'}) \in \SNTerms$.
  \item
    {\bf Step at the root.}
    If there is a step at the root, then,
    by the forms of the left-hand sides of rewriting rules
    involving $\strongabs{}{}{}$,
    we know that $\atm$ and $\atmtwo$ must be canonical terms.
    Recall that
    $\red{(\typ_1\land\typ_2)\pp}{\env}
     = \red{\typ_1\PP}{\env}\rctimes\red{\typ_2\PP}{\env}$
    is defined as the closure of
    $X = \set{\pair{\atm_1}{\atm_2} \ST
           \atm_1 \in \red{\typ_1\PP}{\env},
           \atm_2 \in \red{\typ_2\PP}{\env}}$.
    Also, recall that any canonical term in $\rclosure{X}$
    must be an element of $X$.
    In particular, $\atm \in X$,
    so it must be of the form $\atm = \pair{\atm_1}{\atm_2}$ 
    for some $\atm_1 \in \red{\typ_1\PP}{\env}$
    and some $\atm_2 \in \red{\typ_2\PP}{\env}$.

    Similarly, recall that
    $\red{(\typ_1\land\typ_2)\nn}{\env}
     = \red{\typ_1\NN}{\env}\rcplus\red{\typ_2\NN}{\env}$
    is defined as the closure of
    $Y = \set{\ini{\atmtwo'} \ST
           i \in \set{1,2},
           \atmtwo' \in \red{\typ_i\NN}{\env}}$.
    Also, recall that any canonical term in $\rclosure{Y}$
    must be an element of $Y$.
    In particular, $\atmtwo \in Y$,
    so it must be of the form $\atmtwo = \ini{\atmtwo'}$ 
    for some $i \in \set{1,2}$
    and some $\atmtwo' \in \red{\typ_i\NN}{\env}$.

    Then the step at the root is of the form
    $(\strongabs{}{\pair{\atm_1}{\atm_2}}{\ini{\atmtwo'}}) \tou
     (\strongabs{}{(\ap{\atm_i}{\atmtwo'})}{(\ap{\atmtwo'}{\atm_i})})$.
    By \rlem{reducible_terms_of_classical_type},
    note that
    $\atm_i \in \red{\typ_i\PP}{\env}
              = (\red{\typ_i\NN}{\env}\rcimp\red{\typ_i\pp}{\env})$
    so $\ap{\atm_i}{\atmtwo'} \in \red{\typ_i\pp}{\env}$.
    Similarly,
    by \rlem{reducible_terms_of_classical_type},
    note that
    $\atmtwo' \in \red{\typ_i\NN}{\env}
                = (\red{\typ_i\PP}{\env}\rcimp\red{\typ_i\nn}{\env})$
    so $\ap{\atmtwo'}{\atm_i} \in \red{\typ_i\nn}{\env}$.
    Finally, note that by the outer \ih we know that
    $\rcperp{\red{\typ_i\pp}{\env}}{\red{\typ_i\nn}{\env}}$.
    Therefore we have that
    $(\strongabs{}{(\ap{\atm_i}{\atmtwo'})}{(\ap{\atmtwo'}{\atm_i})})
     \in \SNTerms$,
    as required.
  \end{enumerate}
\item {\bf Disjunction, $\typ = \typ_1\lor\typ_2$.}
  Dual to the previous case.
\item {\bf Implication, $\typ = \typ_1\imp\typ_2$.}
  Let $\atm \in \red{(\typ_1\imp\typ_2)\pp}{\env}$
  and $\atmtwo \in \red{(\typ_1\imp\typ_2)\nn}{\env}$.
  As in the case of conjunction,
  note that $\atm,\atmtwo \in \SNTerms$
  and proceed by induction on $\snsize{\atm}+\snsize{\atmtwo}$
  to show that for each step $(\strongabs{}{\atm}{\atmtwo}) \tou \atmthree$
  we have $\atmthree \in \SNTerms$.
  The interesting case is when there is a step at the root.

  If there is a step at the root, then, by the forms of
  the left-hand sides of rewriting rules involving $\strongabs{}{}{}$,
  we know that $\atm$ and $\atmtwo$ must be canonical terms.
  Recall that $\red{(\typ_1\imp\typ_2)\nn}{\env}
              = \red{\typ_1\PP}{\env}\rccoimp\red{\typ_2\NN}{\env}$
  is defined as the closure of
  $X = \set{\copair{\atmtwo_1}{\atmtwo_2} \ST
         \atmtwo_1 \in \red{\typ_1\PP}{\env},
         \atmtwo_2 \in \red{\typ_2\NN}{\env}
       }$.
  Also, recall that any canonical term in $\rclosure{X}$
  must be an element of $X$.
  In particular, $\atmtwo \in X$,
  so it must be of the form $\atmtwo = \copair{\atmtwo_1}{\atmtwo_2}$
  for some $\atmtwo_1 \in \red{\typ_1\PP}{\env}$,
  and some $\atmtwo_2 \in \red{\typ_2\NN}{\env}$.

  Note that there is only one rewriting rule that may apply at the root
  in this case, so the step is of the form
  $(\strongabs{}{(\lam{\var}{\atm'})}{\copair{\atmtwo_1}{\atmtwo_2}})
   \tou
   \strongabs{}{(\ap{\atm\sub{\var}{\atmtwo_1}}{\atmtwo_2})}{(\ap{\atmtwo_2}{\atm\sub{\var}{\atmtwo_1}})}$
  where $\atm = \lam{\var}{\atm'}$.
  Moreover, since
  $\atm = \lam{\var}{\atm'}
        \in \red{(\typ_1\imp\typ_2)\pp}{\env}
          = \red{\typ_1\PP}{\env}\rcimp\red{\typ_2\PP}{\env}$,
  we have that
  $\ap{(\lam{\var}{\atm'})}{\atmtwo_1} \in \red{\typ_2\PP}{\env}$.
  Furthermore, there is a reduction step
  $\ap{(\lam{\var}{\atm'})}{\atmtwo_1} \tou \atm'\sub{\var}{\atmtwo_1}$
  so, since \RCs are closed by reduction,
  we have that $\atm'\sub{\var}{\atmtwo_1} \in \red{\typ_2\PP}{\env}$.

  By \rlem{reducible_terms_of_classical_type},
  note that
  $\atm'\sub{\var}{\atmtwo_1}
   \in \red{\typ_2\PP}{\env}
    = (\red{\typ_2\NN}{\env}\rcimp\red{\typ_2\pp}{\env})$,
  so
  $\ap{\atm'\sub{\var}{\atmtwo_1}}{\atmtwo_2} \in \red{\typ_2\pp}{\env}$.
  Similarly, by \rlem{reducible_terms_of_classical_type},
  note that
  $\atmtwo_2 \in \red{\typ_2\NN}{\env}
               = (\red{\typ_2\PP}{\env}\rcimp\red{\typ_2\nn}{\env})$,
  so
  $\ap{\atmtwo_2}{\atm'\sub{\var}{\atmtwo_1}} \in \red{\typ_2\nn}{\env}$.
  By the outer \ih we know that
  $\rcperp{\red{\typ_2\pp}{\env}}{\red{\typ_2\nn}{\env}}$.
  Therefore
  $(\strongabs{}{
     \ap{\atm'\sub{\var}{\atmtwo_1}}{\atmtwo_2}
   }{
     \ap{\atmtwo_2}{\atm'\sub{\var}{\atmtwo_1}}
   }) \in \SNTerms$,
  as required.
\item {\bf Co-implication, $\typ = \typ_1\coimp\typ_2$.}
  Dual to the previous case.
\item {\bf Negation, $\typ = \neg\typtwo$.}
  Suppose that $\atm \in \red{(\neg\typtwo)\pp}{\env}$
  and $\atmtwo \in \red{(\neg\typtwo)\nn}{\env}$.
  As in the case of conjunction,
  note that $\atm,\atmtwo \in \SNTerms$
  and proceed by induction on $\snsize{\atm}+\snsize{\atmtwo}$
  to show that for each step
  $(\strongabs{}{\atm}{\atmtwo}) \tou \atmthree$
  we have $\atmthree \in \SNTerms$.
  The interesting case is when there is a step at the root.

  If there is a step at the root, then, by the forms of the left-hand
  sides of rewriting rules involving $\strongabs{}{}{}$,
  we know that $\atm$ and $\atmtwo$ must be canonical terms.
  Recall that $\red{(\neg\typtwo)\pp}{\env} = \rcneg\red{\typtwo\NN}{\env}$
  is defined as the closure of
  $X = \set{\negi{\atm'} \ST \atm' \in \red{\typtwo\NN}{\env}}$.
  Also, recall that any canonical term in $\rclosure{X}$ must
  be an element of $X$. In particular, $\atm \in X$,
  so it must be of the form $\atm = \negi{\atm'}$
  for some $\atm' \in \red{\typtwo\NN}{\env}$.

  Similarly, $\atmtwo$ must be of the form $\atmtwo = \negi{\atmtwo'}$
  for some $\atmtwo' \in \red{\typtwo\PP}{\env}$.

  Then the step at the root is of the form
  $(\strongabs{}{\negi{\atm'}}{\negi{\atmtwo'}})
   \tou (\strongabs{}{(\ap{\atmtwo'}{\atm'})}{(\ap{\atm'}{\atmtwo'})})$.
  By \rlem{reducible_terms_of_classical_type},
  note that
  $\atm'
   \in \red{\typtwo\NN}{\env}
     = (\red{\typtwo\PP}{\env}\rcimp\red{\typtwo\nn}{\env})$
  so $\ap{\atm'}{\atmtwo'} \in \red{\typtwo\nn}{\env}$.
  Similarly, by \rlem{reducible_terms_of_classical_type},
  $\atmtwo'
   \in \red{\typtwo\PP}{\env}
     = (\red{\typtwo\NN}{\env}\rcimp\red{\typtwo\pp}{\env})$
  so $\ap{\atmtwo'}{\atm'} \in \red{\typtwo\pp}{\env}$.
  Finally, note that by the outer \ih we know that
  $\rcperp{\red{\typtwo\pp}{\env}}{\red{\typtwo\nn}{\env}}$.
  Therefore 
  $\strongabs{}{(\ap{\atmtwo'}{\atm'})}{(\ap{\atm'}{\atmtwo'})} \in \SNTerms$,
  as required.
\item {\bf Universal quantification, $\typ = \all{\btyp}{\typtwo}$.}
  Suppose that $\atm \in \red{(\all{\btyp}{\typtwo})\pp}{\env}$
  and $\atmtwo \in \red{(\all{\btyp}{\typtwo})\nn}{\env}$.
  As in the case of conjunction, note that $\atm,\atmtwo \in \SNTerms$
  and proceed by induction on $\snsize{\atm}+\snsize{\atmtwo}$
  to show that for each step $(\strongabs{}{\atm}{\atmtwo}) \tou \atmthree$
  we have $\atmthree \in \SNTerms$.
  The interesting case is when there is a step at the root.

  If there is a step at the root, then, by the forms of the left-hand
  sides of rewriting rules involving $\strongabs{}{}{}$,
  we know that $\atm$ and $\atmtwo$ must be canonical terms.
  Recall that
  $\red{(\all{\btyp}{\typtwo})\nn}{\env}
   = \rcex{(\rc\pp,\rc\nn)\in\RCPerp}{
       \red{\typtwo\NN}{\env\esub{\btyp}{\rc\pp,\rc\nn}}
     }$
  is defined as the closure of:\\
  $
    X = \set{\patu{\atmtwo'} \ST
         \exists{(\rc\pp,\rc\nn)\in\RCPerp}.\,%
           \atmtwo'~\in~\red{\typtwo\NN}{\env\esub{\btyp}{\rc\pp,\rc\nn}}
       }
  $
  \\
  Also, recall that any canonical term in $\rclosure{X}$
  must be an element of $X$.
  In particular, $\atmtwo \in X$,
  so there must exist $(\rc_0\pp,\rc_0\nn) \in \RCPerp$
  such that $\atmtwo = \patu{\atmtwo'}$
  for some $\atmtwo' \in \red{\typtwo\NN}{\env\esub{\btyp}{\rc_0\pp,\rc_0\nn}}$.

  Note that there is only one rewriting rule that may apply at the root
  in this case, so the step is of the form
  $(\strongabs{}{(\lamtu{\atm'})}{\patu{\atmtwo'}})
   \tou
   (\strongabs{}{(\ap{\atm'}{\atmtwo'})}{(\ap{\atmtwo'}{\atm'})})$
  where $\atm = \lamtu{\atm'}$.
  Moreover, since
  $\atm = \lamtu{\atm'}
   \in \red{(\all{\btyp}{\typtwo})\pp}{\env}
   = \rcall{(\rc\pp,\rc\nn)\in\RCPerp}{
       \red{\typtwo\PP}{\env\esub{\btyp}{\rc\pp,\rc\nn}}}$,
  we have in particular that 
  $\apptu{(\lamtu{\atm'})}
   \in \red{\typtwo\PP}{\env\esub{\btyp}{\rc_0\pp,\rc_0\nn}}$.
  Furthermore, there is a reduction step 
  $\apptu{(\lamtu{\atm'})} \tou \atm'$ so,
  since \RCs are closed by reduction, we have that
  $\atm' \in \red{\typtwo\PP}{\env\esub{\btyp}{\rc_0\pp,\rc_0\nn}}$.

  By~\rlem{reducible_terms_of_classical_type},
  note that
  $\atm'
   \in \red{\typtwo\PP}{\env\esub{\btyp}{\rc_0\pp,\rc_0\nn}}$
  which means that
  $\atm'
     (\red{\typtwo\NN}{\env\esub{\btyp}{\rc_0\pp,\rc_0\nn}}
     \rcimp
     \red{\typtwo\pp}{\env\esub{\btyp}{\rc_0\pp,\rc_0\nn}})$
  so $\ap{\atm'}{\atmtwo'}
      \in \red{\typtwo\pp}{\env\esub{\btyp}{\rc_0\pp,\rc_0\nn}}$.
  Similarly, by~\rlem{reducible_terms_of_classical_type},
  note that
  $\atmtwo'
   \in \red{\typtwo\NN}{\env\esub{\btyp}{\rc_0\pp,\rc_0\nn}}
     = (\red{\typtwo\PP}{\env\esub{\btyp}{\rc_0\pp,\rc_0\nn}}
       \rcimp
       \red{\typtwo\nn}{\env\esub{\btyp}{\rc_0\pp,\rc_0\nn}})$
  so $\ap{\atmtwo'}{\atm'}
      \in \red{\typtwo\nn}{\env\esub{\btyp}{\rc_0\pp,\rc_0\nn}}$.
  By the outer \ih,
  $\rcperp{
     \red{\typtwo\pp}{\env\esub{\btyp}{\rc_0\pp,\rc_0\nn}}
   }{
     \red{\typtwo\nn}{\env\esub{\btyp}{\rc_0\pp,\rc_0\nn}}
   }$.
  Therefore
  $(\strongabs{}{(\ap{\atm'}{\atmtwo'})}{(\ap{\atmtwo'}{\atm'})})
   \in \SNTerms$,
  as required.
\item {\bf Existential quantification, $\typ = \ex{\btyp}{\typtwo}$.}
  Dual to the previous case.
\end{enumerate}
\end{proof}

\subsection{Adequacy of the Reducibility Model}

\begin{lemma}[Type erasure preserves non-termination]
\llem{erasure_preserves_non_termination}
If $\erase{\tm}$ is strongly normalizing with respect to $\tou$,
then $\tm$ is strongly normalizing with respect to $\toa{}$.
\end{lemma}
\begin{proof}
It is straightforward
to show that if $\tm \toa{} \tmtwo$
then $\erase{\tm} \tou \erase{\tmtwo}$.
Hence an infinite reduction sequence
$\tm_1 \toa{} \tm_2 \toa{} \hdots$
induces an infinite reduction sequence
$\erase{\tm_1} \tou{} \erase{\tm_2} \tou{} \hdots$.
\end{proof}

\begin{lemma}[Adequacy of absurdity]
\llem{adequacy_absurdity}
Let $\rc_1,\rc_2,\rc' \in \RCSet$ be such that $\rcperp{\rc_1}{\rc_2}$.
If $\atm \in \rc_1$ and $\atmtwo \in \rc_2$
then $(\strongabs{}{\atm}{\atmtwo}) \in \rc'$.
\end{lemma}
\begin{proof}
First, since $\rcperp{\rc_1}{\rc_2}$
we have that $(\strongabs{}{\atm}{\atmtwo}) \in \SNTerms$.
Second, to see that $(\strongabs{}{\atm}{\atmtwo}) \in \rc'$,
since $\rc'$ is complete,
it suffices to show that all the canonical reducts
of $\strongabs{}{\atm}{\atmtwo}$ are in $\rc'$.
Indeed, this holds vacuously, because
a term of the form $\strongabs{}{\atm}{\atmtwo}$
has no canonical reducts.
Note that
its reducts are always of the form $\strongabs{}{\atm'}{\atmtwo'}$,
as can be checked by inspection of all the rewriting rules
defining $\tou$.
\end{proof}

\begin{lemma}[Adequacy of pairing]
\llem{adequacy_pairing}
Let $\rc_1,\rc_2 \in \RCSet$. Then:
\begin{enumerate}
\item
  If $\atm_1 \in \rc_1$ and $\atm_2 \in \rc_2$,
  then $\pair{\atm_1}{\atm_2} \in \rc_1\rctimes\rc_2$.
\item
  If $\atm_1 \in \rc_1$ and $\atm_2 \in \rc_2$,
  then $\copair{\atm_1}{\atm_2} \in \rc_1\rccoimp\rc_2$.
\end{enumerate}
\end{lemma}
\begin{proof}
We only prove the first item; the second one is similar.
First, note that $\atm_1$ and $\atm_2$ are both strongly normalizing
since $\atm_1 \in \rc_1$ and $\atm_2 \in \rc_2$.
From this it is immediate to conclude
that $\pair{\atm_1}{\atm_2} \in \SNTerms$.
Second, to see that $\pair{\atm_1}{\atm_2} \in \rc_1\rctimes\rc_2$,
by definition of the product $\rc_1\rctimes\rc_2$,
it suffices to show that all canonical reducts of $\pair{\atm_1}{\atm_2}$
are of the form $\pair{\atm'_1}{\atm'_2}$
with $\atm'_1 \in \rc_1$ and $\atm'_2 \in \rc_2$.
Indeed, let $\pair{\atm_1}{\atm_2} \rtou \atmthree \in \CanTerms$,
and note that the reduction must be of the form
$\pair{\atm_1}{\atm_2} \rtou \pair{\atm'_1}{\atm'_2} = \atmthree$
with $\atm_1 \rtou \atm'_1$ and $\atm_2 \rtou \atm'_2$.
Since $\rc_1$ and $\rc_2$ are closed by reduction, we have that
$\atm'_1 \in \rc_1$ and $\atm'_2 \in \rc_2$, as required.
\end{proof}

\begin{lemma}[Adequacy of projection]
\llem{adequacy_projection}
Let $\rc_1,\rc_2 \in \RCSet$.
If $\atm \in \rc_1\rctimes\rc_2$ then $\proji{\atm} \in \rc_i$.
\end{lemma}
\begin{proof}
First we claim that $\proji{\atm} \in \SNTerms$.
Note that $\atm \in \SNTerms$ since $\atm \in \rc_1\rctimes\rc_2$.
By induction on $\snsize{\atm}$,
we argue that if $\proji{\atm} \tou \atmtwo$ then $\atmtwo \in \SNTerms$.
We consider two cases, depending on whether the reduction step is
internal to $\atm$ or at the root:
\begin{enumerate}
\item
  If the reduction step is internal to $\atm$,
  that is $\proji{\atm} \tou \proji{\atm'}$ with $\atm \tou \atm'$,
  then $\snsize{\atm} > \snsize{\atm'}$.
  Note that $\atm' \in \rc_1\rctimes\rc_2$
  because $\rc_1\rctimes\rc_2$ is closed by reduction.
  Hence, by \ih, we have that $\proji{\atm'} \in \SNTerms$.
\item
  If the reduction step is at the root,
  then the step must be of the form
  $\proji{\pair{\atm_1}{\atm_2}} \tou \atm_i$
  where $\atm = \pair{\atm_1}{\atm_2}$.
  Since $\atm = \pair{\atm_1}{\atm_2}$ is canonical
  and $\pair{\atm_1}{\atm_2} \in \rc_1\rctimes\rc_2$,
  by definition of the product $\rc_1\rctimes\rc_2$,
  we have that $\atm_i \in \rc_i$.
  Hence $\atm_i \in \SNTerms$.
\end{enumerate}
Second, to see that $\proji{\atm} \in \rc_i$,
since $\rc_i$ is complete, it suffices to show that all canonical
reducts of $\proji{\atm}$ are in $\rc_i$.
That is, let $\proji{\atm} \rtou \atmthree \in \CanTerms$
and let us show that $\atmthree \in \rc_i$.
Note that the reduction sequence $\proji{\atm} \rtou \atmthree$
must be of the form
$\proji{\atm}
 \rtou \proji{\pair{\atm_1}{\atm_2}}
 \tou \atm_i
 \rtou \atmthree$
with $\atm \rtou \pair{\atm_1}{\atm_2}$.
By definition of the product $\rc_1\rctimes\rc_2$,
this means that $\atm_i \in \rc_i$,
and since $\rc_i$ is closed by reduction
we conclude that $\atmthree \in \rc_i$, as required.
\end{proof}

\begin{lemma}[Adequacy of injection]
\llem{adequacy_injection}
Let $\rc_1, \rc_2 \in \RCSet$, and let $i \in \set{1,2}$.
If $\atm \in \rc_i$ then $\ini{\atm} \in \rc_1\rcplus\rc_2$.
\end{lemma}
\begin{proof}
First note that $\atm \in \SNTerms$ since $\atm \in \rc_i$.
From this it is immediate to conclude that $\ini{\atm} \in \SNTerms$.
Second, to see that $\ini{\atm} \in \rc_1\rcplus\rc_2$,
by definition of the sum,
it suffices to show that all canonical reducts of $\ini{\atm}$
are of the form $\ini{\atm'}$ with $\atm' \in \rc_i$.
Indeed, let $\ini{\atm} \rtou \atmthree \in \CanTerms$,
and note that the reduction must be of the form
$\ini{\atm} \rtou \ini{\atm'} = \atmthree$
with $\atm \rtou \atm'$.
Since $\rc_i$ is closed by reduction, we have that $\atm' \in \rc_i$,
as required.
\end{proof}

\begin{lemma}[Adequacy of case]
\llem{adequacy_case}
Let $\rc_1,\rc_2,\rc' \in \RCSet$.
Let $\atm \in \rc_1\rcplus\rc_2$,
and let $\atmtwo_1,\atmtwo_2$
be terms such that
  for all $\atm' \in \rc_1$ we have $\atmtwo_1\sub{\var}{\atm'} \in \rc'$,
and
  for all $\atm' \in \rc_2$ we have $\atmtwo_2\sub{\var}{\atm'} \in \rc'$.
Then $\case{\atm}{\var}{\atmtwo_1}{\var}{\atmtwo_2} \in \rc'$.
\end{lemma}
\begin{proof}
First we claim that
$\case{\atm}{\var}{\atmtwo_1}{\var}{\atmtwo_2} \in \SNTerms$.
Note that $\atm \in \SNTerms$ since $\atm \in \rc_1\rcplus\rc_2$.
Moreover, $\atmtwo_1 \in \SNTerms$
because $\var \in \rc_1$ so $\atmtwo_1 = \atmtwo_1\sub{\var}{\var} \in \rc'$.
Similarly, $\atmtwo_2 \in \SNTerms$.
By induction on $\snsize{\atm} + \snsize{\atmtwo_1} + \snsize{\atmtwo_2}$,
we argue that if
$\case{\atm}{\var}{\atmtwo_1}{\var}{\atmtwo_2} \tou \atmthree$
then $\atmthree \in \SNTerms$.
We consider four cases, depending on whether the reduction step
is internal to $\atm$,
   internal to $\atmtwo_1$,
   internal to $\atmtwo_2$,
or at the root:
\begin{enumerate}
\item
  If the reduction step is internal to $\atm$,
  that is, the step is of the form
  $\case{\atm}{\var}{\atmtwo_1}{\var}{\atmtwo_2}
   \tou \case{\atm'}{\var}{\atmtwo_1}{\var}{\atmtwo_2}$
  with $\atm \tou \atm'$,
  then
  $\snsize{\atm} + \snsize{\atmtwo_1} + \snsize{\atmtwo_2} >
   \snsize{\atm'} + \snsize{\atmtwo_1} + \snsize{\atmtwo_2}$.
  Note that $\atm' \in \rc_1\rcplus\rc_2$ still holds
  since $\rc_1\rcplus\rc_2$ is closed by reduction.
  Hence, by \ih, we have that 
  $\case{\atm'}{\var}{\atmtwo_1}{\var}{\atmtwo_2} \in \SNTerms$.
\item
  If the reduction step is internal to $\atmtwo_1$,
  that is, the step is of the form
  $\case{\atm}{\var}{\atmtwo_1}{\var}{\atmtwo_2}
   \tou \case{\atm}{\var}{\atmtwo'_1}{\var}{\atmtwo_2}$
  with $\atmtwo_1 \tou \atmtwo'_1$,
  then
  $\snsize{\atm} + \snsize{\atmtwo_1} + \snsize{\atmtwo_2} >
   \snsize{\atm} + \snsize{\atmtwo'_1} + \snsize{\atmtwo_2}$.
  Note that $\atmtwo'_1$ still has the property that
  for all $\atm' \in \rc_1$ we have $\atmtwo'_1\sub{\var}{\atm'} \in \rc'$,
  because
  $\atmtwo_1\sub{\var}{\atm'} \tou \atmtwo'_1\sub{\var}{\atm'}$
  and $\rc'$ is closed by reduction,
  and furthermore $\atmtwo_1\sub{\var}{\atm'} \in \rc'$ holds by hypothesis.
  Hence, by \ih, we have that
  $\case{\atm}{\var}{\atmtwo'_1}{\var}{\atmtwo_2} \in \SNTerms$.
\item
  If the reduction step is internal to $\atmtwo_2$,
  the proof is similar to the previous case.
\item
  If the reduction step is at the root, then the step must be of the form
  $\case{\ini{\atm'}}{\var}{\atmtwo_1}{\var}{\atmtwo_2} \tou
   \atmtwo_i\sub{\var}{\atm'}$
  where $\atm = \ini{\atm'}$ for some $i \in \set{1,2}$.
  Since $\atm = \ini{\atm'}$ is canonical and
  $\ini{\atm'} \in \rc_1\rcplus\rc_2$, by definition of the sum,
  we have that $\atm' \in \rc_i$.
  Thus, by hypothesis, $\atmtwo_i\sub{\var}{\atm'} \in \rc'$,
  which implies that $\atmtwo_i\sub{\var}{\atm'} \in \SNTerms$.
\end{enumerate}
Second, to see that
$\case{\atm}{\var}{\atmtwo_1}{\var}{\atmtwo_2} \in \rc'$,
since $\rc'$ is complete,
it suffices to show that all canonical reducts of
$\case{\atm}{\var}{\atmtwo_1}{\var}{\atmtwo_2}$ are in $\rc'$.
That is, let
$\case{\atm}{\var}{\atmtwo_1}{\var}{\atmtwo_2} \rtou \atmthree \in \CanTerms$
and let us show that $\atmthree \in \rc'$.
Note that the reduction sequence
$\case{\atm}{\var}{\atmtwo_1}{\var}{\atmtwo_2} \rtou \atmthree$
must be of the form
$\case{\atm}{\var}{\atmtwo_1}{\var}{\atmtwo_2}
 \rtou \case{\ini{\atm'}}{\var}{\atmtwo'_1}{\var}{\atmtwo'_2}
 \tou \atmtwo'_i\sub{\var}{\atm'}
 \rtou \atmthree$
where $i \in \set{1,2}$
and $\atm \rtou \ini{\atm'}$
and $\atmtwo_1 \rtou \atmtwo'_1$
and $\atmtwo_2 \rtou \atmtwo'_2$.
Since $\ini{\atm'}$ is a canonical reduct of $\atm \in \rc_1\rcplus\rc_2$,
by definition of the sum, we have that $\atm' \in \rc_i$.
Moreover, by hypothesis $\atmtwo_i\sub{\var}{\atm'} \in \rc'$
and, by compatibility of $\tou$-reduction under substitution,
$\atmtwo_i\sub{\var}{\atm'}
\rtou \atmtwo'_i\sub{\var}{\atm'}
\rtou \atmthree$.
Since $\rc'$ is closed by reduction, this implies that
$\atmthree \in \rc'$, as required.
\end{proof}

\begin{lemma}[Adequacy of abstraction]
\llem{adequacy_abstraction}
Let $\rc_1, \rc_2 \in \RCSet$.
Let $\atm$ be such that for every $\atmtwo \in \rc_1$
we have that $\atm\sub{\var}{\atmtwo} \in \rc_2$.
Then $\lam{\var}{\atm} \in \rc_1\rcimp\rc_2$.
\end{lemma}
\begin{proof}
Note that $\var \in \rc_1$,
so by hypothesis $\atm = \atm\sub{\var}{\var} \in \rc_2$.
In particular, $\atm \in \SNTerms$.
From this it is immediate to conclude that $\lam{\var}{\atm} \in \SNTerms$.
To see that $\lam{\var}{\atm} \in \rc_1\rcimp\rc_2$,
by definition of the arrow operator, let $\atmtwo \in \rc_1$ and let us
show that $\ap{(\lam{\var}{\atm})}{\atmtwo} \in \rc_2$.

First we claim that $\ap{(\lam{\var}{\atm})}{\atmtwo} \in \SNTerms$.
We have already argued that $\lam{\var}{\atm} \in \SNTerms$,
and moreover $\atmtwo \in \SNTerms$ since $\atmtwo \in \rc_1$.
By induction on $\snsize{\lam{\var}{\atm}} + \snsize{\atmtwo}$,
we argue that if $\ap{(\lam{\var}{\atm})}{\atmtwo} \tou \atmthree$
then $\atmthree \in \SNTerms$.
We consider three cases, depending on whether the reduction step
is internal to $\lam{\var}{\atm}$,
   internal to $\atmtwo$,
or at the root:
\begin{enumerate}
\item
  If the reduction step is internal to $\lam{\var}{\atm}$,
  that is, the step is of the form
  $\ap{(\lam{\var}{\atm})}{\atmtwo} \tou \ap{(\lam{\var}{\atm'})}{\atmtwo}$
  with $\lam{\var}{\atm} \tou \lam{\var}{\atm'}$,
  then $\snsize{\lam{\var}{\atm}} + \snsize{\atmtwo} >
        \snsize{\lam{\var}{\atm'}} + \snsize{\atmtwo}$.
  Note that $\atm'$ still has the property that
  for every $\atmthree \in \rc_1$
  we have that $\atm'\sub{\var}{\atmthree} \in \rc'$,
  because $\rc'$ is closed by reduction,
  and $\atm\sub{\var}{\atmthree} \tou \atm'\sub{\var}{\atmthree}$,
  and furthermore $\atm\sub{\var}{\atmthree} \in \rc'$
  holds by hypothesis.
  Hence, by \ih, we have that $\ap{(\lam{\var}{\atm'})}{\atmtwo} \in \SNTerms$.
\item
  If the reduction step is internal to $\atmtwo$,
  that is, the step is of the form
  $\ap{(\lam{\var}{\atm})}{\atmtwo} \tou \ap{(\lam{\var}{\atm})}{\atmtwo'}$
  with $\atmtwo \tou \atmtwo'$,
  then $\snsize{\lam{\var}{\atm}} + \snsize{\atmtwo} >
        \snsize{\lam{\var}{\atm}} + \snsize{\atmtwo'}$.
  Note that $\atmtwo' \in \rc_1$ still holds
  because $\rc_1$ is closed by reduction.
  Hence, by \ih, we have that $\ap{(\lam{\var}{\atm})}{\atmtwo'} \in \SNTerms$.
\item
  If the reduction step is at the root,
  then the step must be of the form
  $\ap{(\lam{\var}{\atm})}{\atmtwo} \tou \atm\sub{\var}{\atmtwo}$.
  Thus, by hypothesis, $\atm\sub{\var}{\atmtwo} \in \rc_2$,
  which implies that $\atm\sub{\var}{\atmtwo} \in \SNTerms$.
\end{enumerate}
Second, to see that $\ap{(\lam{\var}{\atm})}{\atmtwo} \in \rc'$,
since $\rc'$ is complete, it suffices to show that all canonical
reducts of $\ap{(\lam{\var}{\atm})}{\atmtwo}$ are in $\rc'$.
That is, let $\ap{(\lam{\var}{\atm})}{\atmtwo} \rtou \atmthree \in \CanTerms$ 
and let us show that $\atmthree \in \rc'$.
Note that the reduction sequence
$\ap{(\lam{\var}{\atm})}{\atmtwo} \rtou \atmthree$
must be of the form
$\ap{(\lam{\var}{\atm})}{\atmtwo}
 \rtou \ap{(\lam{\var}{\atm'})}{\atmtwo'}
 \tou \atm'\sub{\var}{\atmtwo'}
 \rtou \atmthree$
with $\atm \rtou \atm'$ and $\atmtwo \rtou \atmtwo'$.
Note that $\atmtwo' \in \rc_1$ since $\rc_1$ is closed by reduction.
Hence by hypothesis $\atm\sub{\var}{\atmtwo'} \in \rc'$.
Moreover, by compatibility of $\tou$-reduction under substitution,
$\atm\sub{\var}{\atmtwo'}
 \rtou \atm\sub{\var}{\atmtwo'}
 \rtou \atmthree$.
Since $\rc'$ is closed by reduction, this implies that $\atmthree \in \rc'$,
as required.
\end{proof}

\begin{lemma}[Adequacy of negation introduction]
\llem{adequacy_negi}
Let $\rc \in \RCSet$.
If $\atm \in \rc$ then $\negi{\atm} \in \rcneg\rc$.
\end{lemma}
\begin{proof}
First note that $\atm \in \SNTerms$ because $\atm \in \rc$.
From this it is immediate to conclude that $\negi{\atm} \in \SNTerms$.
Second, to see that $\negi{\atm} \in \rcneg\rc$, by definition of
the negation operator for \RCs,
it suffices to show that all canonical reducts of $\negi{\atm}$
are of the form $\negi{\atm'}$ with $\atm' \in \rc$.
Indeed, let $\negi{\atm} \rtou \atmthree \in \CanTerms$,
and note that the reduction must be of the form
$\negi{\atm} \rtou \negi{\atm'}$ with $\atm \rtou \atm'$.
Since $\rc$ is closed by reduction, we have that $\atm' \in \rc$,
as required.
\end{proof}

\begin{lemma}[Adequacy of co-implication elimination]
\llem{adequacy_coimplication_elimination}
Let $\rc_1,\rc_2,\rc' \in \RCSet$.
Let $\atm \in \rc_1\rccoimp\rc_2$,
and let $\atmtwo$ be a term such that
  for all $\atm_1 \in \rc_1$ and for all $\atm_2 \in \rc_2$
  we have $\atmtwo\sub{\var}{\atm_1}\sub{\vartwo}{\atm_2} \in \rc'$.
Then $\colam{\atm}{\var}{\vartwo}{\atmtwo} \in \rc'$.
\end{lemma}
\begin{proof}
First we claim that $\colam{\atm}{\var}{\vartwo}{\atmtwo} \in \SNTerms$.
Note that $\atm \in \SNTerms$ since $\atm \in \rc_1\rccoimp\rc_2$.
Moreover, $\atmtwo \in \SNTerms$
because $\var\in\rc_1$ and $\vartwo\in\rc_2$,
so $\atmtwo = \atmtwo\sub{\var}{\var}\sub{\vartwo}{\vartwo} \in \rc'$.
By induction on $\snsize{\atm} + \snsize{\atmtwo}$,
we argue that if $\colam{\atm}{\var}{\vartwo}{\atmtwo} \tou \atmthree$
then $\atmthree \in \SNTerms$.
We consider three cases, depending on whether the reduction step is internal
to $\atm$, internal to $\atmtwo$, or at the root:
\begin{enumerate}
\item
  If the reduction step is internal to $\atm$,
  that is, the step is of the form
  $\colam{\atm}{\var}{\vartwo}{\atmtwo} \tou \colam{\atm'}{\var}{\vartwo}{\atmtwo}$
  with $\atm \tou \atm'$,
  then $\snsize{\atm} + \snsize{\atmtwo} > \snsize{\atm'} + \snsize{\atmtwo}$.
  Note that $\atm' \in \rc_1\rccoimp\rc_2$ still holds
  since $\rc_1\rccoimp\rc_2$ is closed by reduction.
  Hence, by \ih, we have that 
  $\colam{\atm'}{\var}{\vartwo}{\atmtwo} \in \SNTerms$.
\item
  If the reduction step is internal to $\atmtwo$,
  that is, the step is of the form
  $\colam{\atm}{\var}{\vartwo}{\atmtwo} \tou \colam{\atm}{\var}{\vartwo}{\atmtwo'}$
  with $\atmtwo \tou \atmtwo'$,
  then $\snsize{\atm} + \snsize{\atmtwo} > \snsize{\atm} + \snsize{\atmtwo'}$.
  Note that $\atmtwo'$ still has the property that
  for every $\atm_1 \in \rc_1$ and every $\atm_2 \in \rc_2$
  we have $\atmtwo'\sub{\var}{\atm_1}\sub{\vartwo}{\atm_2} \in \rc'$,
  because
  $\atmtwo\sub{\var}{\atm_1}\sub{\vartwo}{\atm_2} \tou \atmtwo'\sub{\var}{\atm_1}\sub{\vartwo}{\atm_2}$
  and $\rc'$ is closed by reduction
  and furthermore $\atmtwo\sub{\var}{\atm_1}\sub{\vartwo}{\atm_2} \in \rc'$
  holds by hypothesis.
  Hence, by \ih, we have that
  $\colam{\atm}{\var}{\vartwo}{\atmtwo'} \in \SNTerms$.
\item
  If the reduction step is at the root, then the step must be of the form
  $\colam{\copair{\atm_1}{\atm_2}}{\var}{\vartwo}{\atmtwo}
   \to \atmtwo\sub{\var}{\atm_1}\sub{\vartwo}{\atm_2}$,
  where $\atm = \copair{\atm_1}{\atm_2}$.
  Since $\atm = \copair{\atm_1}{\atm_2}$ is canonical
  and $\copair{\atm_1}{\atm_2} \in \rc_1\rccoimp\rc_2$,
  by definition of the co-implication operator on reducibility candidates,
  we have that $\atm_1 \in \rc_1$ and $\atm_2 \in \rc_2$.
  Thus, by hypothesis,
  $\atmtwo\sub{\var}{\atm_1}\sub{\vartwo}{\atm_2} \in \rc'$,
  which implies
  $\atmtwo\sub{\var}{\atm_1}\sub{\vartwo}{\atm_2} \in \SNTerms$.
\end{enumerate}
Second, to see that $\colam{\atm}{\var}{\vartwo}{\atmtwo} \in \rc'$,
since $\rc'$ is complete, it suffices to show that all canonical reducts
of $\colam{\atm}{\var}{\vartwo}{\atmtwo}$ are in $\rc'$.
That is, let $\colam{\atm}{\var}{\vartwo}{\atmtwo} \rtou \atmthree \in \CanTerms$
and let us show that $\atmthree \in \rc'$.
Note that the reduction sequence
$\colam{\atm}{\var}{\vartwo}{\atmtwo} \rtou \atmthree$
must be of the form
$\colam{\atm}{\var}{\vartwo}{\atmtwo}
 \rtou \colam{\copair{\atm_1}{\atm_2}}{\var}{\vartwo}{\atmtwo'}
 \rtou \atmtwo'\sub{\var}{\atm_1}\sub{\vartwo}{\atm_2}
 \rtou \atmthree$
where $\atm \rtou \copair{\atm_1}{\atm_2}$
and $\atmtwo \rtou \atmtwo'$.
Since $\copair{\atm_1}{\atm_2}$
is a canonical reduct of $\atm \in \rc_1\rccoimp\rc_2$,
by definition of the co-implication operator on reducibility candidates,
we have that $\atm_1 \in \rc_1$ and $\atm_2 \in \rc_2$.
Moreover, by hypothesis
$\atmtwo\sub{\var}{\atm_1}\sub{\vartwo}{\atm_2} \in \rc'$ and,
by compatibility of $\tou$-reduction under substitution,
$\atmtwo\sub{\var}{\atm_1}\sub{\vartwo}{\atm_2} \rtou
 \atmtwo'\sub{\var}{\atm_1}\sub{\vartwo}{\atm_2} \rtou \atmthree$.
Since $\rc'$ is closed by reduction, this implies that $\atmthree \in \rc'$,
as required.
\end{proof}

\begin{lemma}[Adequacy of negation elimination]
\llem{adequacy_nege}
Let $\rc \in \RCSet$.
If $\atm \in \rcneg\rc$ then $\nege{\atm} \in \rc$.
\end{lemma}
\begin{proof}
First we claim that $\nege{\atm} \in \SNTerms$.
Note that $\atm \in \SNTerms$ since $\atm \in \rcneg\rc$.
By induction on $\snsize{\atm}$, we argue that if
$\nege{\atm} \tou \atmthree$ then $\atmthree \in \SNTerms$.
We consider two cases, depending on whether the reduction step
is internal to $\atm$ or at the root:
\begin{enumerate}
\item
  If the reduction step is internal to $\atm$,
  that is, $\nege{\atm} \tou \nege{\atm'}$ with $\atm \tou \atm'$,
  then $\snsize{\atm} > \snsize{\atm'}$.
  Note that $\atm' \in \rcneg\rc$ still holds since $\rcneg\rc$
  is closed by reduction. Hence, by \ih, we have that
  $\nege{\atm'} \in \SNTerms$.
\item
  If the reduction step is at the root, then the step must be of the
  form $\nege{(\negi{\atm'})} \tou \atm'$
  where $\atm = \negi{\atm'}$.
  Since $\atm = \negi{\atm'}$ is canonical and $\negi{\atm'} \in \rcneg\rc$,
  by definition of the negation operator for \RCs,
  we have that $\atm' \in \rc$,
  which implies $\atm' \in \SNTerms$.
\end{enumerate}
Second, to see that $\nege{\atm} \in \rc$, since $\rc$ is complete,
it suffices to show that all canonical reducts of $\nege{\atm}$
are in $\rc$. That is, let $\nege{\atm} \rtou \atmthree \in \CanTerms$
and let us show that $\atmthree \in \rc$.
Note that the reduction sequence $\nege{\atm} \rtou \atmthree$
must be of the form
$\nege{\atm}
 \rtou \nege{(\negi{\atm'})}
 \tou \atm'
 \rtou \atmthree$
with $\atm \rtou \negi{\atm'}$.
Since $\negi{\atm'}$ is a canonical reduct of $\atm \in \rcneg\rc$,
by definition of the negation operator for \RCs,
we have that $\atm' \in \rc$.
Since $\rc$ is closed by reduction, this implies that $\atmthree \in \rc$,
as required.
\end{proof}

\begin{lemma}[Adequacy of universal abstraction]
\llem{adequacy_type_abstraction}
Suppose that $\set{\rc_i}_{i\inI} \subseteq \RCSet$,
where $I$ is assumed to be non-empty.
If $\atm \in \rc_i$ for all $i\inI$,
then $\lamtu{\atm} \in \rcall{i\inI}{\rc_i}$.
\end{lemma}
\begin{proof}
Since $I$ is non-empty, $\atm \in \rc_{i_0}$
for at least one index $i_0 \in I$.
This implies that $\atm \in \SNTerms$, from which it is immediate to
conclude that $\lamtu{\atm} \in \SNTerms$.
To see that $\lamtu{\atm} \in \rcall{i\inI}{\rc_i}$,
by definition of the indexed product,
let $j \in I$ be an arbitrary index
and let us show that $\apptu{(\lamtu{\atm})} \in \rc_j$.

First, we claim that $\apptu{(\lamtu{\atm})} \in \SNTerms$.
By induction on $\snsize{\lamtu{\atm}}$, we argue that if
$\apptu{(\lamtu{\atm})} \tou \atmthree$ then $\atmthree \in \SNTerms$.
We consider two cases, depending on whether the reduction step is
internal to $\lamtu{\atm}$ or at the root:
\begin{enumerate}
\item
  If the reduction step is internal to $\lamtu{\atm}$,
  that is, the step is of the form
  $\apptu{(\lamtu{\atm})} \tou \apptu{(\lamtu{\atm'})}$
  with $\lamtu{\atm} \tou \lamtu{\atm'}$,
  then $\snsize{\lamtu{\atm}} > \snsize{\lamtu{\atm'}}$.
  Note that $\atm' \in \rc_j$,
  because $\rc_j$ is closed by reduction,
  and $\atm \tou \atm'$,
  and furthermore $\atm \in \rc_j$ holds by hypothesis.
  Hence by \ih we have that $\apptu{(\lamtu{\atm})} \in \SNTerms$.
\item
  If the reduction step is at the root, then the step must be
  of the form $\apptu{(\lamtu{\atm})} \tou \atm$.
  By hypothesis, $\atm \in \rc_j$, which implies that $\atm \in \SNTerms$.
\end{enumerate}
Second, to see that
$\apptu{(\lamtu{\atm})} \in \rc_j$, since $\rc_j$ is complete,
it suffices to show that all canonical reducts of
$\apptu{(\lamtu{\atm})}$ are in $\rc_j$.
That is, let $\apptu{(\lamtu{\atm})} \rtou \atmthree \in \CanTerms$,
and let us show that $\atmthree \in \rc_j$.
Note that the reduction sequence
$\apptu{(\lamtu{\atm})} \rtou \atmthree$
must be of the form
$\apptu{(\lamtu{\atm})}
 \rtou \apptu{(\lamtu{\atm'})}
 \tou \atm'
 \rtou \atmthree$
with $\atm \rtou \atm'$.
Since $\rc_j$ is closed by reduction,
and $\atm \rtou \atm' \rtou \atmthree$,
and furthermore $\atm \in \rc_j$,
we conclude that $\atmthree \in \rc_j$,
as required.
\end{proof}

\begin{lemma}[Adequacy of existential introduction]
\llem{adequacy_existential_introduction}
Let $\set{\rc_i}_{i\inI} \subseteq \RCSet$.
If $\atm \in \rc_j$ for some $j \in I$,
then $\patu{\atm} \in \rcex{i\inI}{\rc_i}$.
\end{lemma}
\begin{proof}
First note that $\atm \in \SNTerms$ because $\atm \in \rc_j$.
From this it is immediate to conclude that $\patu{\atm} \in \SNTerms$.
Second, to see that $\patu{\atm} \in \rcex{i\inI}{\rc_i}$, by
definition of the indexed sum, it suffices to show that all canonical
reducts of $\patu{\atm}$ are of the form $\patu{\atm'}$ with
$\atm' \in \rc_i$ for some $i\inI$.
More specifically, we argue that, in this case, $\atm' \in \rc_j$.
Indeed, let $\patu{\atm} \rtou \atmthree \in \CanTerms$
and note that the reduction must be of the form
$\patu{\atm} \rtou \patu{\atm'}$ with $\atm \rtou \atm'$.
Since $\rc_j$ is closed by reduction, we have that $\atm' \in \rc_j$,
as required.
\end{proof}

\begin{lemma}[Adequacy of existential elimination]
\llem{adequacy_existential_elimination}
Suppose that $\set{\rc_i}_{i\inI} \subseteq \RCSet$,
where $I$ is assumed to be non-empty,
and let $\rc' \in \RCSet$.
Let $\atm \in \rcex{i\inI}{\rc_i}$,
and let $\atmtwo$ be such that
for all $i\inI$ and for all $\atm'\in\rc_i$
we have that $\atmtwo\sub{\var}{\atm'} \in \rc'$.
Then $\optu{\atm}{\var}{\atmtwo} \in \rc'$.
\end{lemma}
\begin{proof}
First we claim that $\optu{\atm}{\var}{\atmtwo} \in \SNTerms$.
Note that $\atm \in \SNTerms$ because $\atm \in \rcex{i\inI}{\rc_i}$.
Moreover, since $I$ is non-empty, there is at least one index $i_0 \in I$,
and $\var\in\rc_{i_0}$, so by hypothesis
$\atmtwo = \atmtwo\sub{\var}{\var} \in \rc'$, which implies that
$\atmtwo \in \SNTerms$.
By induction on $\snsize{\atm} + \snsize{\atmtwo}$,
we argue that if $\optu{\atm}{\var}{\atmtwo} \tou \atmthree$
then $\atmthree \in \SNTerms$.
We consider three cases, depending on whether the reduction
step is internal to $\atm$, internal to $\atmtwo$, or at the root:
\begin{enumerate}
\item
  If the reduction step is internal to $\atm$,
  that is, the step is of the form
  $\optu{\atm}{\var}{\atmtwo} \tou \optu{\atm'}{\var}{\atmtwo}$
  with $\atm \tou \atm'$,
  then $\snsize{\atm} + \snsize{\atmtwo} > \snsize{\atm'} + \snsize{\atmtwo}$.
  Note that $\atm' \in \rcex{i\inI}{\rc_i}$
  still holds since $\rcex{i\inI}{\rc_i}$ is closed by reduction.
  Hence, by \ih, we have that
  $\optu{\atm'}{\var}{\atmtwo} \in \SNTerms$.
\item
  If the reduction step is internal to $\atmtwo$,
  that is, the step is of the form
  $\optu{\atm}{\var}{\atmtwo} \tou \optu{\atm}{\var}{\atmtwo'}$
  with $\atmtwo \tou \atmtwo'$,
  then $\snsize{\atm} + \snsize{\atmtwo} > \snsize{\atm} + \snsize{\atmtwo'}$.
  Note that $\atmtwo'$ still has the property that
  for all $i\inI$ and all $\atm'\in\rc_i$
  we have $\atmtwo'\sub{\var}{\atm'} \in \rc'$,
  because $\atmtwo\sub{\var}{\atm'} \tou \atmtwo'\sub{\var}{\atm'}$
  and $\rc'$ is closed by reduction,
  and furthermore $\atmtwo\sub{\var}{\atm'} \in \rc'$
  holds by hypothesis.
  Hence, by \ih, we have that
  $\optu{\atm}{\var}{\atmtwo'} \in \SNTerms$.
\item
  If the reduction step is at the root, then the step must be of the form
  $\optu{\patu{\atm'}}{\var}{\atmtwo} \tou \atmtwo\sub{\var}{\atm'}$
  where $\atm = \patu{\atm'}$.
  Since $\atm = \patu{\atm'}$ is canonical and
  $\patu{\atm'} \in \rcex{i\inI}{\rc_i}$,
  by definition of the indexed sum,
  we have that $\atm' \in \rc_j$ for some $j\inI$.
  Thus, by hypothesis, $\atmtwo\sub{\var}{\atm'} \in \rc'$,
  which implies that $\atmtwo\sub{\var}{\atm'} \in \SNTerms$.
\end{enumerate}
Second, to see that $\optu{\atm}{\var}{\atmtwo} \in \rc'$,
since $\rc'$ is complete, it suffices to show that all canonical
reducts of $\optu{\atm}{\var}{\atmtwo}$ are in $\rc'$.
That is, let $\optu{\atm}{\var}{\atmtwo} \rtou \atmthree \in \CanTerms$
and let us show that $\atmthree \in \rc'$.
Note that the reduction sequence
$\optu{\atm}{\var}{\atmtwo} \rtou \atmthree$
must be of the form
$\optu{\atm}{\var}{\atmtwo}
 \rtou \optu{\patu{\atm'}}{\var}{\atmtwo'}
 \tou \atmtwo'\sub{\var}{\atm'}
 \rtou \atmthree$.
where $\atm \rtou \patu{\atm'}$ and $\atmtwo \rtou \atmtwo'$.
Since $\patu{\atm'}$ is a canonical reduct of $\atm \in \rcex{i\inI}{\rc_i}$,
by definition of the indexed sum, we have that there exists an
index $j\inI$ such that $\atm' \in \rc_j$.
Moreover, by hypothesis, $\atmtwo\sub{\var}{\atm'} \in \rc'$
and, by compatibility of $\tou$-reduction under substitution,
$\atmtwo\sub{\var}{\atm'}
 \rtou \atmtwo'\sub{\var}{\atm'}
 \rtou \atmthree$.
Since $\rc'$ is closed by reduction, this implies that
$\atmthree \in \rc'$, as required.
\end{proof}


\begin{definition}[Adequate substitutions]
A {\em substitution} is a function $\subst$ mapping
each variable to an term in $\UTerms$.
We write $\subst\esub{\var}{\atm}$ for the substitution $\subst'$
that results from extending $\subst$ in such a way that
$\subst'(\var) = \atm$ and $\subst'(\vartwo) = \subst(\vartwo)$
for any other variable $\vartwo \neq \var$.
We write $\atm^\subst$ for the term that results from
the capture-avoiding substitution
of each free occurrence of each variable $\var$ in $\atm$
by $\subst(\var)$.
We say that the substitution $\subst$ is
{\em adequate} for the typing context $\tctx$ under the environment $\env$,
and we write $\judgSubst{\subst}{\env}{\tctx}$,
if for each type assignment $(\var:\ev)\in\tctx$
we have that $\subst(\var) \in \red{\ev}{\env}$.
\end{definition}

\begin{theorem}[Adequacy]
\lthm{appendix:adequacy}
If $\judg{\tctx}{\tm}{\ev}$
and $\judgSubst{\subst}{\env}{\tctx}$
then $\erase{\tm}^{\subst} \in \red{\ev}{\env}$.
\end{theorem}
\begin{proof}
We proceed by induction
on the derivation of the typing judgment $\judg{\tctx}{\tm}{\ev}$.
For each pair of dual rules, we only study the positive one
(\eg we study $\Iandp$ but not the dual rule $\Iorn$):
\begin{enumerate}
\item \Ax:
  Let $\judg{\tctx,\var:\ev}{\var}{\ev}$
  and $\judgSubst{\subst}{\env}{\tctx,\var:\ev}$.
  Then $\erase{\var}^\subst = \subst(\var) \in \red{\ev}{\env}$
  by the fact that $\subst$ is adequate.
\item \Abs:
  Let
  $\judg{\tctx}{\strongabs{\ev}{\tm}{\tmtwo}}{\ev}$
  be derived from $\judg{\tctx}{\tm}{\typ\pp}$
  and $\judg{\tctx}{\tmtwo}{\typ\nn}$,
  and let $\judgSubst{\subst}{\env}{\tctx}$.
  By \ih we have that
  $\erase{\tm}^\subst \in \red{\typ\pp}{\env}$
  and
  $\erase{\tmtwo}^\subst \in \red{\typ\nn}{\env}$.
  Recall from~\rlem{red_opposite_orthogonal}
  that $\rcperp{\red{\typ\pp}{\env}}{\red{\typ\nn}{\env}}$.
  By~\rlem{adequacy_absurdity},
  $(\strongabs{}{\erase{\tm}^\subst}{\erase{\tmtwo}^\subst})
   \in \red{\ev}{\env}$.
\item \Icp:
  Let $\judg{\tctx}{\claslamp{(\var:\typ\NN)}{\tm}}{\typ\PP}$
  be derived from $\judg{\tctx, \var : \typ\NN}{\tm}{\typ\pp}$,
  and let $\judgSubst{\subst}{\env}{\tctx}$.
  By \ih, for every substitution $\subst'$
  such that $\judgSubst{\subst'}{\env}{\tctx,\var:\typ\NN}$
  we have that $\erase{\tm}^{\subst'} \in \red{\typ\pp}{\env}$.
  In particular, for every $\atm \in \red{\typ\NN}{\env}$
  we have that $\erase{\tm}^{\subst\esub{\var}{\atm}} \in \red{\typ\pp}{\env}$.
  Moreover, note that
  $\erase{\tm}^{\subst\esub{\var}{\atm}} =
   \erase{\tm}^{\subst\esub{\var}{\var}}\sub{\var}{\atm}$.
  Hence by~\rlem{adequacy_abstraction}
  we have that
  $\erase{\claslamp{(\var:\typ\NN)}{\tm}}^{\subst}
   = \lam{\var}{\erase{\tm}^{\subst\esub{\var}{\var}}}
   \in \red{\typ\NN}{\env}\rcimp\red{\typ\pp}{\env}$.
  To conclude, recall from~\rlem{reducible_terms_of_classical_type}
  that $\red{\typ\NN}{\env}\rcimp\red{\typ\pp}{\env} = \red{\typ\PP}{\env}$,
  so
  $\erase{\claslamp{(\var:\typ\NN)}{\tm}}^{\subst} \in \red{\typ\PP}{\env}$.
\item \Ecp:
  Let $\judg{\tctx}{\clasapp{\tm}{\tmtwo}}{\typ\pp}$
  be derived from $\judg{\tctx}{\tm}{\typ\PP}$
  and $\judg{\tctx}{\tmtwo}{\typ\NN}$,
  and let $\judgSubst{\subst}{\env}{\tctx}$.
  By \ih, $\erase{\tm}^\subst \in \red{\typ\PP}{\env}$
  and $\erase{\tmtwo}^\subst \in \red{\typ\NN}{\env}$.
  Recall from~\rlem{reducible_terms_of_classical_type}
  that $\red{\typ\PP}{\env} = \red{\typ\NN}{\env}\rcimp\red{\typ\pp}{\env}$,
  so
  $\erase{\clasapp{\tm}{\tmtwo}}^\subst
   = \ap{\erase{\tm}^\subst}{\erase{\tmtwo}^\subst}
   \in \red{\typ\pp}{\env}$.
\item \Iandp:
  Let $\judg{\tctx}{\pairp{\tm}{\tmtwo}}{(\typ \land \typtwo)\pp}$
  be derived from $\judg{\tctx}{\tm}{\typ\PP}$
  and $\judg{\tctx}{\tmtwo}{\typtwo\PP}$,
  and let $\judgSubst{\subst}{\env}{\tctx}$.
  By \ih, $\erase{\tm}^\subst \in \red{\typ\PP}{\env}$
  and $\erase{\tmtwo}^\subst \in \red{\typtwo\PP}{\env}$.
  By~\rlem{adequacy_pairing},
  $\pair{\erase{\tm}^\subst}{\erase{\tmtwo}^\subst}
   \in \red{\typ\PP}{\env}\rctimes\red{\typtwo\PP}{\env}
     = \red{(\typ\land\typtwo)\pp}{\env}$.
\item \Eandp:
  Let $\judg{\tctx}{\projip{\tm}}{\typ_i\PP}$
  be derived from
  $\judg{\tctx}{\tm}{(\typ_1 \land \typ_2)\pp}$
  for some $i \in \set{1, 2}$,
  and let $\judgSubst{\subst}{\env}{\tctx}$.
  By \ih,
  $\erase{\tm}^\subst
   \in \red{(\typ_1\land\typ_2)\pp}{\env}
     = \red{\typ_1\PP}{\env}\rctimes\red{\typ_2\PP}{\env}$.
  By \rlem{adequacy_projection},
  $\proji{\erase{\tm}^\subst} \in \red{\typ_i\PP}{\env}$.
\item \Iorp:
  Let $\judg{\tctx}{\inip{\tm}}{(\typ_1 \lor \typ_2)\pp}$
  be derived from $\judg{\tctx}{\tm}{\typ_i\PP}$
  for some $i \in \set{1, 2}$,
  and let $\judgSubst{\subst}{\env}{\tctx}$.
  By \ih, $\erase{\tm}^\subst \in \red{\typ_i\PP}{\env}$.
  By~\rlem{adequacy_injection},
  $\ini{\erase{\tm}^\subst}
   \in \red{\typ_1\PP}{\env}\rcplus\red{\typ_2\PP}{\env}
     = \red{(\typ_1\lor\typ_2)\pp}{\env}$.
\item \Eorp:
  Let
  $\judg{\tctx}{\casep{\tm}{\var:\typ\PP}{\tmtwo}{\vartwo:\typtwo\PP}{\tmthree}}{\ev}$
  be derived from
  $\judg{\tctx}{\tm}{(\typ \lor \typtwo)\pp}$
  and
  $\judg{\tctx, \var:\typ\PP}{\tmtwo}{\ev}$
  and
  $\judg{\tctx, \vartwo:\typtwo\PP}{\tmthree}{\ev}$,
  and let $\judgSubst{\subst}{\env}{\tctx}$.
  By \ih on the first premise,
  $\erase{\tm}^\subst
   \in \red{(\typ\lor\typtwo)\pp}{\env}
   = \red{\typ\PP}{\env}\rcplus\red{\typtwo\PP}{\env}$.
  By \ih on the second premise,
  we have that
  $\erase{\tmtwo}^{\subst'} \in \red{\ev}{\env}$
  for every substitution $\subst'$
  such that $\judgSubst{\subst'}{\env}{\tctx,\var:\typ\PP}$.
  In particular, for every $\atm \in \red{\typ\PP}{\env}$
  we have that
  $\erase{\tmtwo}^{\subst\esub{\var}{\atm}} \in \red{\ev}{\env}$.
  Moreover, note that
  $\erase{\tmtwo}^{\subst\esub{\var}{\atm}} =
   \erase{\tmtwo}^{\subst\esub{\var}{\var}}\sub{\var}{\atm}$.
  Similarly, by \ih on the third premise,
  for every $\atmtwo \in \red{\typtwo\PP}{\env}$,
  we have that
  $\erase{\tmthree}^{\subst\esub{\vartwo}{\vartwo}}\sub{\vartwo}{\atmtwo}
   \in \red{\ev}{\env}$.
  Hence by~\rlem{adequacy_case},
  we have that
  $\erase{\casep{\tm}{\var:\typ\PP}{\tmtwo}{\vartwo:\typtwo\PP}{\tmthree}}^\subst
  = \case{\erase{\tm}^\subst}{\var}{
      \erase{\tmtwo}^{\subst\esub{\var}{\var}}
    }{\vartwo}{
      \erase{\tmthree}^{\subst\esub{\vartwo}{\vartwo}}
    }
  \in \red{\ev}{\env}$,
  as required.
\item \Iimpp:
  Let
  $\judg{\tctx}{\lamp{(\var:\typ\PP)}{\tm}}{(\typ\imp\typtwo)\pp}$
  be derived from
  $\judg{\tctx,\var:\typ\PP}{\tm}{\typtwo\PP}$,
  and let $\judgSubst{\subst}{\env}{\tctx}$.
  By \ih, we have that $\erase{\tm}^{\subst'} \in \red{\typtwo\PP}{\env}$
  for every substitution $\subst'$ such that
  $\judgSubst{\subst'}{\env}{\tctx,\var:\typ\PP}$.
  In particular, for every $\atm \in \red{\typ\PP}{\env}$
  we have that
  $\erase{\tm}^{\subst\esub{\var}{\atm}} \in \red{\typtwo\PP}{\env}$.
  Moreover, note that
  $\erase{\tm}^{\subst\esub{\var}{\atm}} =
   \erase{\tm}^{\subst\esub{\var}{\var}}\sub{\var}{\atm}$.
  Hence by~\rlem{adequacy_abstraction}
  we have that
  $\erase{\lamp{(\var:\typ\PP)}{\tm}}^\subst =
   \lam{\var}{\erase{\tm}^{\subst\esub{\var}{\var}}}
   \in (\red{\typ\PP}{\env}\rcimp\red{\typtwo\PP}{\env})
     = \red{(\typ\imp\typtwo)\pp}{\env}$.
\item \Eimpp:
  Let
  $\judg{\tctx}{\app{\tm}{\tmtwo}}{\typtwo\PP}$
  be derived from $\judg{\tctx}{\tm}{(\typ\imp\typtwo)\pp}$
  and $\judg{\tctx}{\tmtwo}{\typ\PP}$,
  and let $\judgSubst{\subst}{\env}{\tctx}$.
  By \ih on the first premise,
  $\erase{\tm}^\subst
   \in \red{(\typ\imp\typtwo)\pp}{\env}
     = \red{\typ\PP}{\env}\rcimp\red{\typtwo\PP}{\env}$.
  By \ih on the second premise,
  $\erase{\tmtwo}^\subst \in \red{\typ\PP}{\env}$.
  By definition of the arrow operator,
  $\ap{\erase{\tm}^\subst}{\erase{\tmtwo}^\subst} \in \red{\typtwo\PP}{\env}$,
  as required.
\item \Icoimpp:
  Let $\judg{\tctx}{\copairp{\tm}{\tmtwo}}{(\typ\coimp\typtwo)\pp}$
  be derived from $\judg{\tctx}{\tm}{\typ\NN}$
  and $\judg{\tctx}{\tmtwo}{\typtwo\PP}$.
  By \ih, $\erase{\tm}^\subst \in \red{\typ\NN}{\env}$
  and $\erase{\tmtwo}^\subst \in \red{\typtwo\PP}{\env}$.
  By~\rlem{adequacy_pairing},
  $\copair{\erase{\tm}^\subst}{\erase{\tmtwo}^\subst}
   \in \red{\typ\NN}{\env}\rctimes\red{\typtwo\PP}{\env}
     = \red{(\typ\coimp\typtwo)\pp}{\env}$.
\item \Ecoimpp:
  Let $\judg{\tctx}{\colamp{\tm}{\var}{\vartwo}{\tmtwo}}{\ev}$
  be derived from
  $\judg{\tctx}{\tm}{(\typ \coimp \typtwo)\pp}$
  and
  $\judg{\tctx,\var:\typ\NN,\vartwo:\typtwo\PP}{\tmtwo}{\ev}$,
  and let $\judgSubst{\subst}{\env}{\tctx}$.
  By \ih on the first premise,
  $\erase{\tm}^\subst \in \red{(\typ\coimp\typtwo)\pp}{\env}
                        = \red{\typ\NN}{\env}\coimp\red{\typtwo\PP}{\env}$.
  By \ih on the second premise,
  we have that
  $\erase{\tmtwo}^{\subst'} \in \red{\ev}{\env}$
  for every substitution $\subst'$
  such that $\judgSubst{\subst'}{\env}{\tctx,\var:\typ\NN,\vartwo:\typtwo\PP}$.
  In particular,
  for every $\atm_1\in\red{\typ\NN}{\env}$
  and every $\atm_2\in\red{\typtwo\PP}{\env}$
  we have that
  $\erase{\tmtwo}^{\subst\esub{\var}{\atm_1}\esub{\vartwo}{\atm_2}} \in \red{\ev}{\env}$
  Moreover, note that
  $\erase{\tmtwo}^{\subst\esub{\var}{\atm_1}\esub{\vartwo}{\atm_2}}
  = \erase{\tmtwo}^{\subst\esub{\var}{\var}\esub{\vartwo}{\vartwo}}\sub{\var}{\atm_1}\sub{\vartwo}{\atm_2}$.
  Hence by \rlem{adequacy_coimplication_elimination},
  we have that
  $\erase{\colamp{\tm}{\var}{\vartwo}{\tmtwo}}^{\subst}
  = \colam{\erase{\tm}^{\subst}}{\var}{\vartwo}{\erase{\tmtwo}^{\subst\esub{\var}{\var}\esub{\vartwo}{\vartwo}}}
  \in \red{\ev}{\env}$,
  as required.
\item \Inotp:
  Let $\judg{\tctx}{\negip{\tm}}{(\neg\typ)\pp}$
  be derived from $\judg{\tctx}{\tm}{\typ\NN}$,
  and let $\judgSubst{\subst}{\env}{\tctx}$.
  By \ih, $\erase{\tm}^\subst \in \red{\typ\NN}{\env}$.
  By \rlem{adequacy_negi},
  $\negi{\erase{\tm}^\subst}
   \in \rcneg\red{\typ\NN}{\env}
     = \red{(\neg\typ)\pp}{\env}$.
\item \Enotp:
  Let $\judg{\tctx}{\negep{\tm}}{\typ\NN}$
  be derived from $\judg{\tctx}{\tm}{(\neg\typ)\pp}$,
  and let $\judgSubst{\subst}{\env}{\tctx}$.
  By \ih,
  $\erase{\tm}^\subst
   \in \red{(\neg\typ)\pp}{\env}
     = \rcneg\red{\typ\NN}{\env}$.
  By \rlem{adequacy_nege},
  $\nege{\erase{\tm}^\subst} \in \red{\typ\NN}{\env}$.
\item \Iallp:
  Let
    $\tctx \vdash \lamtp{\btyp}{\tm} : (\all{\btyp}{\typ})\pp$
  be derived from
    $\tctx \vdash \tm : \typ\PP$,
  where we assume that $\btyp \notin \ftv{\tctx}$,
  and let $\judgSubst{\subst}{\env}{\tctx}$.
  By \ih,
  for every environment $\env'$
  and every substitution $\subst'$
  such that $\judgSubst{\subst'}{\env'}{\tctx}$,
  we have that $\erase{\tm}^{\subst'} \in \red{\typ\PP}{\env'}$.
  In particular, consider two arbitrary orthogonal reducibility candidates
  $(\rc\pp,\rc\nn) \in \RCPerp$,
  consider the environment $\env' = \env\esub{\btyp}{\rc\pp,\rc\nn}$,
  and observe that $\judgSubst{\subst}{\env'}{\tctx}$,
  because if $(\var:\ev)\in\tctx$
  by irrelevance~\rlem{reducible_terms_irrelevance}
  $\subst(\var) \in \red{\typ\PP}{\env} = \red{\typ\PP}{\env'}$
  since $\btyp\notin\ftv{\ev}$.
  Then
  $\erase{\tm}^{\subst} \in \red{\typ\PP}{\env\esub{\btyp}{\rc\pp,\rc\nn}}$,
  where $(\rc\pp,\rc\nn) \in \RCPerp$ are arbitrary.
  To conclude note that, by \rlem{adequacy_type_abstraction},
  $\lamtu{\erase{\tm}^{\subst}}
   \in \rcall{(\rc\pp,\rc\nn) \in \RCPerp}{
         \red{\typ\PP}{\env\esub{\btyp}{\rc\pp,\rc\nn}}
       }
     = \red{(\all{\btyp}{\typ})\pp}{\env}$.
\item \Eallp:
  Let
    $\tctx \vdash \apptp{\tm}{\typ} : \typtwo\PP\sub{\btyp}{\typ}$
  be derived from
    $\tctx \vdash \tm : (\all{\btyp}{\typtwo})\pp$,
  and let $\judgSubst{\subst}{\env}{\tctx}$.
  By \ih,
  $\erase{\tm}^\subst
   \in \red{(\all{\btyp}{\typtwo})\pp}{\env}
     = \rcall{(\rc\pp,\rc\nn)\in\RCPerp}{
         \red{\typtwo\PP}{\env\esub{\btyp}{\rc\pp,\rc\nn}}
       }$.
  By definition of the indexed product, this means that
  $\apptu{\erase{\tm}^\subst}
   \in \red{\typtwo\PP}{\env\esub{\btyp}{\rc_0\pp,\rc_0\nn}}$
  for an arbitrary choice of
  $(\rc_0\pp,\rc_0\nn)\in\RCPerp$.
  Recall from~\rlem{red_opposite_orthogonal}
  that $\rcperp{\red{\typ\pp}{\env}}{\red{\typ\nn}{\env}}$
  so, in particular,
  taking $\rc_0\pp := \red{\typ\pp}{\env}$
  and $\rc_0\nn := \red{\typ\nn}{\env}$,
  we have that
  $\apptu{\erase{\tm}^\subst}
   \in \red{\typtwo\PP}{
         \env\esub{\btyp}{\red{\typ\pp}{\env},\red{\typ\nn}{\env}}
       }$.
  To conclude, observe that by~\rlem{reducible_terms_substitution}
  $\red{\typtwo\PP}{
         \env\esub{\btyp}{\red{\typ\pp}{\env},\red{\typ\nn}{\env}}
    } = \red{\typtwo\PP\sub{\btyp}{\typ}}{\env}$.
  Hence
  $\apptu{\erase{\tm}^\subst}
   \in \red{\typtwo\PP\sub{\btyp}{\typ}}{\env}$,
  as required.
\item \Iexp:
  Let $\tctx \vdash \patp{\typ}{\tm} : (\ex{\btyp}{\typtwo})\pp$
  be derived from $\tctx \vdash \tm : \typtwo\PP\sub{\btyp}{\typ}$,
  and let $\judgSubst{\subst}{\env}{\tctx}$. 
  By \ih, $\erase{\tm}^{\subst} \in \red{\typtwo\PP\sub{\btyp}{\typ}}{\env}$.
  By~\rlem{reducible_terms_substitution}, note that:\\
  $
   \red{\typtwo\PP\sub{\btyp}{\typ}}{\env} =
   \red{\typtwo\PP}{\env\esub{\btyp}{\red{\typ\pp}{\env},\red{\typ\nn}{\env}}}
  $\\
  so
  $\erase{\tm}^{\subst}
   \in \red{\typtwo\PP}{
         \env\esub{\btyp}{\red{\typ\pp}{\env},\red{\typ\nn}{\env}}
       }$.
  Recall from~\rlem{red_opposite_orthogonal}
  that $\rcperp{\red{\typ\pp}{\env}}{\red{\typ\nn}{\env}}$.
  Hence, by~\rlem{adequacy_existential_introduction},
  $\patu{\erase{\tm}^{\subst}}
   \in \rcex{(\rc\pp,\rc\nn)\in\RCPerp}{
         \red{\typtwo\PP}{\env\esub{\btyp}{\rc\pp,\rc\nn}}
       }
     = \red{(\ex{\btyp}{\typtwo})\pp}{\env}$.
\item \Eexp:
  Let $\tctx \vdash \optp{\tm}{\btyp}{\var}{\tmtwo} : \ev$
  be derived from $\tctx \vdash \tm : (\ex{\btyp}{\typ})\pp$
  and $\tctx,\var:\typ\PP \vdash \tmtwo : \ev$,
  where we assume that $\btyp \notin \ftv{\tctx,\ev}$.
  Moreover, let $\judgSubst{\subst}{\env}{\tctx}$.
  By \ih on the first premise, we have that:
  \begin{center}
  $
    \erase{\tm}^{\subst}
    \in \red{(\ex{\btyp}{\typ})\pp}{\env}
       = \rcex{(\rc\pp,\rc\nn)\in\RCPerp}{
           \red{\typ\PP}{\env\esub{\btyp}{\rc\pp,\rc\nn}}
         }
  $
  \end{center}
  By \ih on the second premise,
  for every environment $\env'$ and every substitution $\subst'$
  such that $\judgSubst{\subst'}{\env'}{\tctx,\var:\typ\PP}$,
  we have that
  $\erase{\tmtwo}^{\subst'} \in \red{\ev}{\env'}$.
  In particular, consider two arbitrary orthogonal reducibility
  candidates $(\rc\pp,\rc\nn) \in \RCPerp$,
  and consider the environment $\env' = \env\esub{\btyp}{\rc\pp,\rc\nn}$.
  By hypothesis $\btyp\notin\ftv{\ev}$,
  so by irrelevance~(\rlem{reducible_terms_irrelevance})
  we have that $\red{\ev}{\env} = \red{\ev}{\env'}$.
  Moreover, consider an arbitrary term $\atm \in \red{\typ\PP}{\env'}$,
  and consider the substitution $\subst' = \subst\esub{\var}{\atm}$.

  We argue that $\judgSubst{\subst'}{\env'}{\tctx,\var:\typ\PP}$ holds.
  To see this, note, on one hand,
  that $\subst'(\var) = \atm \in \red{\typ\PP}{\env'}$.
  On the other hand, given an association $(\vartwo:\evtwo)\in\tctx$
  for a variable other than $\var$,
  note that by hypothesis
  $\subst'(\vartwo) = \subst(\vartwo) \in \red{\evtwo}{\env}$.
  Moreover, by hypothesis $\btyp\notin\ftv{\evtwo}$,
  so by irrelevance~(\rlem{reducible_terms_irrelevance})
  we have that $\red{\evtwo}{\env} = \red{\evtwo}{\env'}$.

  In summary, the \ih on the second premise tells us that
  for all $(\rc\pp,\rc\nn) \in \RCPerp$,
  and for all $\atm \in \red{\typ\PP}{\env\esub{\btyp}{\rc\pp,\rc\nn}}$,
  we have
  $\erase{\tmtwo}^{\subst\esub{\var}{\atm}} \in \red{\ev}{\env}$.
  Note that:
  \begin{center}
  $\erase{\tmtwo}^{\subst\esub{\var}{\atm}} =
   \erase{\tmtwo}^{\subst\esub{\var}{\var}}\sub{\var}{\atm}$
  \end{center}
  Finally, by \rlem{adequacy_existential_elimination},
  we conclude that
  $\erase{\optp{\tm}{\btyp}{\var}{\tmtwo}}^\subst
  = \optu{\erase{\tm}^\subst}{\var}{\erase{\tmtwo}^{\subst\esub{\var}{\var}}}
  \in \red{\ev}{\env}$,
  as required.

\end{enumerate}
\end{proof}

\begin{corollary}[Strong normalization]
Let $\judg{\tctx}{\tm}{\ev}$.
Then $\tm$ is strongly normalizing.
\end{corollary}
\begin{proof}
Consider the environment $\env$ that maps all type variables
to the bottom reducibility candidate,
that is $\env(\btyp\pp) = \env(\btyp\nn) = \rcbot$
for every type variable $\btyp$.
Note that this is indeed an environment because
$(\rcbot,\rcbot) \in \RCPerp$.
Moreover, let $\subst$ be the identity substitution,
that is, $\subst(\var) = \var$ for all variables $\var$.
Remark that $\judgSubst{\subst}{\env}{\tctx}$ 
because a variable $\var$ is in any reducibility candidate,
so if $(\var:\evtwo)\in\tctx$, then indeed $\var \in \red{\evtwo}{\env}$.
Then, by adequacy~\rthm{adequacy}, we have that
$\erase{\tm} \in \red{\ev}{\env}$,
so in particular $\erase{\tm}$ is strongly normalizing with respect to $\tou$.
Finally, recall from~\llem{erasure_preserves_non_termination}
that this entails that $\tm$ is strongly normalizing with respect to $\toa{}$.
\end{proof}

\newpage
\section{Admissible Principles in $\lambdaPRJ$}
\lsec{appendix:prj_principles}

\begin{definition}{$\varset$-intuitionistic term}
If $\varset$ is an arbitrary set of variables,
a term $\tm$ is {\em $\varset$-intuitionistic}
if it is intuitionistic and, furthermore,
it has no useful free occurrences of variables in $\varset$.
Remark that a term is intuitionistic if and only if it is
$\emptyset$-intuitionistic.

If $\varset$ is an arbitrary set of variables,
we generalize the judgment $\judgPRJ{\tctx}{\tm}{\typ}$
by writing $\judgPRJ[\varset]{\tctx}{\tm}{\typ}$,
if the judgment is derivable in $\PRK$ and $\tm$ is $\varset$-intuitionistic.
\end{definition}

\begin{lemma}[Admissible Principles in $\lambdaPRJ$]
\llem{prj_admissible_rules}
\quad
\begin{enumerate}
\item {\bf Counterfactual weakening.}
  If $\judgPRJ[\varset]{\tctx}{\tm}{\ev}$
  and $\varset \supseteq \varset'$
  then $\judgPRJ[\varset']{\tctx}{\tm}{\ev}$.
\item {\bf Weakening} ($\Weakening$):
  If $\judgPRJ[\varset]{\tctx}{\tm}{\ev}$
  and $\var \notin \fv{\tm}$
  then $\judgPRJ[\varset]{\tctx, \var:\evtwo}{\tm}{\ev}$.
\item {\bf Intuitionistic cut} ($\ICut$):
  If $\judgPRJ[\varset]{\tctx,\var:\ev}{\tm}{\evtwo}$
  and $\judgPRJ[\varset]{\tctx}{\tmtwo}{\ev}$
  then $\judgPRJ[\varset]{\tctx}{\tm\sub{\var}{\tmtwo}}{\evtwo}$.
\item {\bf Counterfactual cut} ($\CCut$):
  If $\judgPRJ[\varset\cup\set{\var}]{\tctx,\var:\ev}{\tm}{\evtwo}$
  and $\judgPRK{\tctx}{\tmtwo}{\ev}$,
  then $\judgPRJ[\varset]{\tctx}{\tm\sub{\var}{\tmtwo}}{\evtwo}$.
\item {\bf Generalized absurdity} ($\Abs'$):
  If $\judgPRJ[\varset]{\tctx}{\tm}{\ev}$
  and $\judgPRJ[\varset]{\tctx}{\tmtwo}{\ev\OP}$,
  where $\ev$ is not necessarily strong,
  there is a term $\abs{\evtwo}{\tm}{\tmtwo}$
  such that $\judgPRJ[\varset]{\tctx}{\abs{\evtwo}{\tm}{\tmtwo}}{\evtwo}$.
\item {\bf Intuitionistic contraposition} ($\IContrapose$):
  Let $\judgPRJ[\varset]{\tctx, \var : \typ\PP}{\tm}{\evtwo}$
  and suppose that $\vartwo\notin\varset\cup\fv{\tm}$.
  Then there is a term $\icontrapose{\var}{\vartwo}{\tm}$
  such that
  $\judgPRJ[\varset]{\tctx, \vartwo : \evtwo\OP}{\icontrapose{\var}{\vartwo}{\tm}}{\typ\NN}$.
\item {\bf Counterfactual contraposition} ($\CContrapose$):
  $\judgPRJ[\varset\cup\set{\var}]{\tctx, \var : \typ\NN}{\tm}{\evtwo}$,
  there is a term $\ccontrapose{\var}{\vartwo}{\tm}$
  such that
  $\judgPRJ[\varset]{\tctx, \vartwo : \evtwo\OP}{\ccontrapose{\var}{\vartwo}{\tm}}{\typ\PP}$.
\item {\bf Weak negation introduction}:
  If $\judgPRJ{\tctx}{\tm}{\typ\NN}$,
  there is a term $\negiP{\tm}$
  such that $\judgPRJ{\tctx}{\negiP{\tm}}{(\neg\typ)\PP}$.
\item {\bf Weak negation elimination}:
  If $\judgPRJ{\tctx}{\tm}{(\neg\typ)\PP}$,
  there is a term $\negeP{\tm}$
  such that $\judgPRJ{\tctx}{\negeP{\tm}}{\typ\NN}$.
\end{enumerate}
\end{lemma}
\begin{proof}
{\bf Counterfactual weakening}, {\bf weakening}, {\bf cut},
and {\bf counterfactual cut} are straightforward by induction
on the derivation of the given judgment.

For {\bf generalized absurdity}, it suffices to take:
\[
  \abs{\evtwo}{\tm}{\tmtwo} \eqdef
    \begin{cases}
      \strongabs{\evtwo}{\tm}{\tmtwo}
      & \text{if $\ev = \typ\pp$} \\
      \strongabs{\evtwo}{\tmtwo}{\tm}
      & \text{if $\ev = \typ\nn$} \\
      \strongabs{\evtwo}{(\clasapp{\tm}{\tmtwo})}{(\clasapn{\tmtwo}{\tm})}
      & \text{if $\ev = \typ\PP$} \\
      \strongabs{\evtwo}{(\clasapp{\tmtwo}{\tm})}{(\clasapn{\tm}{\tmtwo})}
      & \text{if $\ev = \typ\NN$} \\
    \end{cases}
\]

For {\bf intuitionistic contraposition}, take the term below.
Observe that no condition must be imposed on the free occurrences
of $\var$ in $\tm$, because $\var$ is a {\em positive} counterfactual:
\[
  \icontrapose{\var}{\vartwo}{\tm} \eqdef
    \claslamn{(\var:\typ\PP)}{
      (\abs{\typ\nn}{
        \tm
      }{ 
        \vartwo
      })
    }
\]

For {\bf counterfactual contraposition}, take the term below.
Observe that $\var$ is a {\em negative} counterfactual, and this is
why we require in the hypothesis that $\var$ has no useful free occurrences
in~$\tm$:
\[
  \ccontrapose{\var}{\vartwo}{\tm} \eqdef
    \claslamp{(\var:\typ\NN)}{
      (\abs{\typ\pp}{
        \tm
      }{ 
        \vartwo
      })
    }
\]

For {\bf weak negation introduction},
let $\judgPRJ{\tctx}{\tm}{\typ\NN}$ and take the term below.
Observe that there are no free occurrences of the negative counterfactual
in $\negip{\tm}$:
\[
  \negiP{\tm} \eqdef
  \claslamp{(\under:(\neg\typ)\NN)}{
    \negip{
      \tm
    }
  }
\]

For {\bf weak negation elimination},
let $\judgPRJ{\tctx}{\tm}{(\neg\typ)\PP}$ and take
the term below. Observe that all the occurrences
of the negative counterfactual $\var$ are useless:
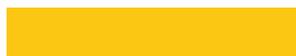
\[
  \negeP{\tm} \eqdef
  \claslamn{(\var:\typ\PP)}{
    \clasappn{
      \negep{
        (\clasapp{
          \tm
        }{
          \claslamn{\under:(\neg\typ)\PP}{
            \negin{
              \var
            }
          }
        })
      }
    }{
      \var
    }
  }
\]
\end{proof}

\newpage
\section{An Alternative Presentation of $\lambdaPRJ$: the $\lambdaPRJV$ Type System}
\lsec{appendix:prjv_definition}

The $\lambdaPRJ$ type system
is formulated by imposing a global restriction on terms,
forbidding that subterms of specific forms occur in specific positions.
Alternatively, valid judgments in $\PRJ$ may be defined inductively
with inference rules, as shown in this section.

\begin{definition}[The $\lambdaPRJV$ type system]
The type system $\lambdaPRJV$
is defined with typing judgments of the form ``$\judgV{\tctx}{\tm}{\ev}$''
for each set of variables $\varset$
by means of inductive rules below.
We write $\judgPRJV{\tctx}{\tm}{\typ}$
if $\judgV{\tctx}{\tm}{\typ}$ is derivable in $\PRJV$.

\begin{fragBox}{Basic rules}
\[
\indrule{\Ax\PRJVsym}{
  \var\notin\varset
}{
  \judgV{\tctx,\var:\ev}{\var}{\ev}
}
\indrule{\Abs\PRJVsym}{
  \judgV{\tctx}{\tm}{\typ\pp}
  \HS
  \judgV{\tctx}{\tmtwo}{\typ\nn}
}{
  \judgV{\tctx}{\strongabs{\ev}{\tm}{\tmtwo}}{\ev}
}
\]
\[
\indrule{\Icp\PRJVsym}{
  \judgV[\varset\cup\set{\var}]{\tctx, \var : \typ\NN}{\tm}{\typ\pp}
}{
  \judgV{\tctx}{\claslamp{(\var:\typ\NN)}{\tm}}{\typ\PP}
}
\indrule{\Icn\PRJVsym}{
  \judgV{\tctx, \var : \typ\PP}{\tm}{\typ\nn}
}{
  \judgV{\tctx}{\claslamp{(\var:\typ\PP)}{\tm}}{\typ\NN}
}
\]
\[
\indrule{\Ecp\PRJVsym}{
  \judgV{\tctx}{\tm}{\typ\PP}
  \HS
  \judgPRK{\tctx}{\tmtwo}{\typ\NN}
}{
  \judgV{\tctx}{\clasapp{\tm}{\tmtwo}}{\typ\pp}
}
\indrule{\Ecn\PRJVsym}{
  \judgV{\tctx}{\tm}{\typ\NN}
  \HS
  \judgV{\tctx}{\tmtwo}{\typ\PP}
}{
  \judgV{\tctx}{\clasapn{\tm}{\tmtwo}}{\typ\nn}
}
\]
\end{fragBox}

\begin{fragBox}{Conjunction and disjunction}
\[
\indrule{\colorand{\Iandp\PRJVsym}}{
  \judgV{\tctx}{\tm}{\typ\PP}
  \HS
  \judgV{\tctx}{\tmtwo}{\typtwo\PP}
}{
  \judgV{\tctx}{\pairp{\tm}{\tmtwo}}{(\typ \land \typtwo)\pp}
}
\indrule{\colorand{\Iorn\PRJVsym}}{
  \judgV{\tctx}{\tm}{\typ\NN}
  \HS
  \judgV{\tctx}{\tmtwo}{\typtwo\NN}
}{
  \judgV{\tctx}{\pairn{\tm}{\tmtwo}}{(\typ \lor \typtwo)\nn}
}
\]
\[
\indrule{\colorand{\Eandp\PRJVsym}}{
  \judgV{\tctx}{\tm}{(\typ_1 \land \typ_2)\pp}
}{
  \judgV{\tctx}{\projip{\tm}}{\typ_i\PP}
}
\indrule{\colorand{\Eorn\PRJVsym}}{
  \judgV{\tctx}{\tm}{(\typ_1 \lor \typ_2)\nn}
}{
  \judgV{\tctx}{\projin{\tm}}{\typ_i\NN}
}
\]
\[
\indrule{\colorand{\Iorp\PRJVsym}}{
  \judgV{\tctx}{\tm}{\typ_i\PP}
}{
  \judgV{\tctx}{\inip{\tm}}{(\typ_1 \lor \typ_2)\pp}
}
\indrule{\colorand{\Iandn\PRJVsym}}{
  \judgV{\tctx}{\tm}{\typ_i\NN}
}{
  \judgV{\tctx}{\inin{\tm}}{(\typ_1 \land \typ_2)\nn}
}
\]
\[
\indrule{\colorand{\Eorp\PRJVsym}}{
  \judgV{\tctx}{\tm}{(\typ \lor \typtwo)\pp}
  \HS
  \judgV{\tctx, \var:\typ\PP}{\tmtwo}{\ev}
  \HS
  \judgV{\tctx, \vartwo:\typtwo\PP}{\tmthree}{\ev}
}{
  \judgV{\tctx}{\casep{\tm}{\var:\typ\PP}{\tmtwo}{\vartwo:\typtwo\PP}{\tmthree}}{\ev}
}
\]
\end{fragBox}

\begin{fragBox}{Implication and co-implication}
\[
\indrule{\colorimp{\Iimpp\PRJVsym}}{
  \judgV{\tctx,\var:\typ\PP}{\tm}{\typtwo\PP}
}{
  \judgV{\tctx}{\lamp{(\var:\typ\PP)}{\tm}}{(\typ\imp\typtwo)\pp}
}
\indrule{\colorimp{\Icoimpn\PRJVsym}}{
  \judgV{\tctx,\var:\typ\NN}{\tm}{\typtwo\NN}
}{
  \judgV{\tctx}{\lamn{(\var:\typ\NN)}{\tm}}{(\typ\coimp\typtwo)\nn}
}
\]
\[
\indrule{\colorimp{\Eimpp\PRJVsym}}{
  \judgV{\tctx}{\tm}{(\typ\imp\typtwo)\pp}
  \HS
  \judgV{\tctx}{\tmtwo}{\typ\PP}
}{
  \judgV{\tctx}{\app{\tm}{\tmtwo}}{\typtwo\PP}
}
\indrule{\colorimp{\Ecoimpn\PRJVsym}}{
  \judgV{\tctx}{\tm}{(\typ\coimp\typtwo)\nn}
  \HS
  \judgV{\tctx}{\tmtwo}{\typ\NN}
}{
  \judgV{\tctx}{\apn{\tm}{\tmtwo}}{\typtwo\NN}
}
\]
\[
\indrule{\colorimp{\Icoimpp\PRJVsym}}{
  \judgV{\tctx}{\tm}{\typ\NN}
  \HS
  \judgV{\tctx}{\tmtwo}{\typtwo\PP}
}{
  \judgV{\tctx}{\copairp{\tm}{\tmtwo}}{(\typ\coimp\typtwo)\pp}
}
\indrule{\colorimp{\Iimpn\PRJVsym}}{
  \judgV{\tctx}{\tm}{\typ\PP}
  \HS
  \judgV{\tctx}{\tmtwo}{\typtwo\NN}
}{
  \judgV{\tctx}{\copairn{\tm}{\tmtwo}}{(\typ\imp\typtwo)\nn}
}
\]
\medskip
\[
\indrule{\colorimp{\Ecoimpp\PRJVsym}}{
  \judgV{\tctx}{\tm}{(\typ\coimp\typtwo)\pp}
  \HS
  \judgV{\tctx,\var:\typ\NN,\vartwo:\typtwo\PP}{\tmtwo}{\ev}
}{
  \judgV{\tctx}{\colamp{\tm}{\var}{\vartwo}{\tmtwo}}{\ev}
}
\]
\end{fragBox}

\newpage
\begin{fragBox}{Negation}
\[
\indrule{\colorneg{\Inotp\PRJVsym}}{
  \judgV{\tctx}{\tm}{\typ\NN}
}{
  \judgV{\tctx}{\negip{\tm}}{(\neg\typ)\pp}
}
\indrule{\colorneg{\Inotn\PRJVsym}}{
  \judgV{\tctx}{\tm}{\typ\PP}
}{
  \judgV{\tctx}{\negin{\tm}}{(\neg\typ)\nn}
}
\indrule{\colorneg{\Enotp\PRJVsym}}{
  \judgV{\tctx}{\tm}{(\neg\typ)\pp}
}{
  \judgV{\tctx}{\negep{\tm}}{\typ\NN}
}
\]
\end{fragBox}

\begin{fragBox}{Second-order quantification}
\[
  \indrule{\colorall{\Iallp\PRJVsym}}{
    \judgV{\tctx}{\tm}{\typ\PP}
    \HS
    \btyp \not\in \ftv{\tctx}
  }{
    \judgV{\tctx}{\lamtp{\btyp}{\tm}}{(\all{\btyp}{\typ})\pp}
  }
  \indrule{\colorall{\Iexn\PRJVsym}}{
    \judgV{\tctx}{\tm}{\typ\NN}
    \HS
    \btyp \not\in \ftv{\tctx}
  }{
    \judgV{\tctx}{\lamtn{\btyp}{\tm}}{(\ex{\btyp}{\typ})\nn}
  }
\]
\[
  \indrule{\colorall{\Eallp\PRJVsym}}{
    \judgV{\tctx}{\tm}{(\all{\btyp}{\typtwo})\pp}
  }{
    \judgV{\tctx}{\apptp{\tm}{\typ}}{\typtwo\PP\sub{\btyp}{\typ}}
  }
  \indrule{\colorall{\Eexn\PRJVsym}}{
    \judgV{\tctx}{\tm}{(\ex{\btyp}{\typtwo})\nn}
  }{
    \judgV{\tctx}{\apptn{\tm}{\typ}}{\typtwo\NN\sub{\btyp}{\typ}}
  }
\]
\[
  \indrule{\colorall{\Iexp\PRJVsym}}{
    \judgV{\tctx}{\tm}{\typtwo\PP\sub{\btyp}{\typ}}
  }{
    \judgV{\tctx}{\patp{\typ}{\tm}}{(\ex{\btyp}{\typtwo})\pp}
  }
  \indrule{\colorall{\Ialln\PRJVsym}}{
    \judgV{\tctx}{\tm}{\typtwo\NN\sub{\btyp}{\typ}}
  }{
    \judgV{\tctx}{\patn{\typ}{\tm}}{(\all{\btyp}{\typtwo})\nn}
  }
\]
\[
  \indrule{\colorall{\Eexp\PRJVsym}}{
    \judgV{\tctx}{\tm}{(\ex{\btyp}{\typ})\pp}
    \HS
    \judgV{\tctx,\var:\typ\PP}{\tmtwo}{\ev}
    \HS
    \btyp \not\in \ftv{\tctx,\ev}
  }{
    \judgV{\tctx}{\optp{\tm}{\btyp}{\var:\typ\PP}{\tmtwo}}{\ev}
  }
\]
\end{fragBox}
\end{definition}

\begin{proposition}[Equivalence of $\lambdaPRJ$ and $\lambdaPRJV$]
\lprop{prj_equiv_prjv}
\lprop{appendix:prj_equiv_prjv}
The following are equivalent:
\begin{enumerate}
\item
  $\judgPRJ[\varset]{\tctx}{\tm}{\typ}$
\item
  $\judgPRJV[\varset]{\tctx}{\tm}{\typ}$
\end{enumerate}
\end{proposition}
\begin{proof}
By induction on $\tm$. We focus on the interesting cases only:
\begin{itemize}
\item[($\Rightarrow$)]
  Suppose that $\judgPRJ[\varset]{\tctx}{\tm}{\typ}$,
  \ie that $\judgPRK{\tctx}{\tm}{\typ}$ and $\tm$ is $\varset$-intuitionistic.
  We proceed by induction on the derivation of the judgment.
  The interesting cases are:
  \begin{enumerate}
  \item \Ax:
    Let $\judgPRK{\tctx,\var:\ev}{\var}{\ev}$.
    Note that $\var$ is $\varset$-intuitionistic
    and the term $\var$ contains a useful free occurrence of $\var$,
    so $\var\notin\varset$.
    Hence $\judgPRJV{\tctx,\var:\ev}{\var}{\ev}$ by $\Ax\PRJVsym$.
  \item \Icp:
    Let $\judg{\tctx}{\claslamp{(\var:\typ\NN)}{\tmtwo}}{\typ\PP}$
    be derived from $\judg{\tctx, \var : \typ\NN}{\tmtwo}{\typ\pp}$.
    Since $\tm = \claslamp{(\var:\typ\NN)}{\tmtwo}$
    is $\varset$-intuitionistic
    we know that $\tmtwo$ is $\varset$-intuitionistic
    and moreover $\tmtwo$ does not have useful free occurrences of $\var$.
    Hence $\tmtwo$ is $(\varset\cup\set{\var})$-intuitionistic
    and by \ih we have that
    $\judgPRJV[\varset\cup\set{\var}]{\tctx, \var : \typ\NN}{\tmtwo}{\typ\pp}$.
    Applying the $\Icp\PRJVsym$ rule we conclude that
    $\judgPRJV[\varset]{\tctx}{\claslamp{(\var:\typ\NN)}{\tmtwo}}{\typ\PP}$,
    as required.
  \item \Ecp:
    Let $\judg{\tctx}{\clasapp{\tmtwo}{\tmthree}}{\typ\pp}$
    be derived from $\judg{\tctx}{\tmtwo}{\typ\PP}$
    and $\judg{\tctx}{\tmthree}{\typ\NN}$.
    Since $\tm = \clasapp{\tmtwo}{\tmthree}$
    is $\varset$-intuitionistic,
    we have that $\tmtwo$ is $\varset$-intuitionistic.
    So, by \ih on the first premise, we have
    $\judgPRJV{\tctx}{\tmtwo}{\typ\PP}$.
    Applying the $\Ecp\PRJVsym$ rule directly on the second premise,
    \ie without the need of resorting to the \ih for the second premise,
    we conclude that $\judgPRJV{\tctx}{\clasapp{\tmtwo}{\tmthree}}{\typ\pp}$,
    as required.
  \item \colorand{\Eandn}:
    Then this case is impossible, given that $\tm$
    must be of the form
    $\casen{\tmtwo}{\var:\typ\NN}{\tmthree}{\vartwo:\typtwo\NN}{\tmfour}$,
    which is not $\varset$-intuitionistic,
    contradicting the hypothesis.
  \item \colorimp{\Eimpn}:
    Then this case is impossible, given that $\tm$
    must be of the form $\judg{\tctx}{\colam{\tmtwo}{\var}{\vartwo}{\tmthree}}{\ev}$,
    which is not $\varset$-intuitionistic,
    contradicting the hypothesis.
  \item \colorneg{\Enotn}:
    Then this case is impossible, given that $\tm$
    must be of the form $\negen{\tmtwo}$,
    which is not $\varset$-intuitionistic,
    contradicting the hypothesis.
  \item \colorall{\Ealln}:
    Then this case is impossible, given that $\tm$
    must be of the form $\optn{\tmtwo}{\btyp}{\var}{\tmthree}$,
    which is not $\varset$-intuitionistic,
    contradicting the hypothesis.
  \end{enumerate}
  The remaining cases are all straightforward by \ih.
  For example, for the \Iandp case,
  let $\judg{\tctx}{\pairp{\tmtwo}{\tmthree}}{(\typ \land \typtwo)\pp}$
  be derived from $\judg{\tctx}{\tmtwo}{\typ\PP}$
  and $\judg{\tctx}{\tmthree}{\typtwo\PP}$.
  Since $\tm = \pairp{\tmtwo}{\tmthree}$ is $\varset$-intuitionistic,
  $\tmtwo$ and $\tmthree$ are also $\varset$-intuitionistic
  so, by \ih, we have $\judgPRJV{\tctx}{\tmtwo}{\typ\PP}$
  and $\judgPRJV{\tctx}{\tmthree}{\typtwo\PP}$.
  Applying the $\Iandp\PRJVsym$ rule,
  we conclude that
  $\judgPRJV{\tctx}{\pairp{\tmtwo}{\tmthree}}{(\typ \land \typtwo)\pp}$,
  as required.
\item[($\Leftarrow$)]
  By induction on the derivation of $\judgPRJV{\tctx}{\tm}{\typ}$.
  The reasoning is similar as for the ``only if'' direction.
\end{itemize}
\end{proof}

\newpage
\section{$\lambdaPRJ$ Refines Intuitionistic Second-Order Logic}
\lsec{appendix:prj_refinement}

The proof of~\rthm{lambdaPRJ_refinement},
that $\lambdaPRJ$ refines intuitionistic second-order logic,
is split into two lemmas:
{\bf Intuitionistic~Conservativity}~(\rlem{lambdaPRJ_conservativity})
proves the implication $1 \implies 2$, and
{\bf Intuitionistic~Embedding}~(\rlem{lambdaPRJ_embedding})
proves the implication $2 \implies 1$.

Recall that $\intem{\ev}$ is defined as follows:
\[
  \begin{array}{rcl@{\HS}rcl}
    \intem{\typ\pp} & \eqdef & \typ
  &
    \intem{\typ\PP} & \eqdef & \typ \\
    \intem{\typ\nn} & \eqdef & \neg\typ
  &
    \intem{\typ\NN} & \eqdef & \neg\typ \\
  \end{array}
\]

\begin{lemma}[Intuitionistic Conservativity]
\llem{lambdaPRJ_conservativity}
If $\logPRJ{\tctx}{\ev}$
then $\logNJ{\intem{\tctx}}{\intem{\ev}}$.
\end{lemma}
\begin{proof}
We shall prove a slightly more general property.
We claim that
if $\judgPRJ[\dom{\tctxtwo}]{\tctx,\tctxtwo}{\tm}{\ev}$,
then $\logNJ{\intem{\tctx}}{\intem{\ev}}$.
Observe that the general property implies the statement of the lemma
taking $\tctxtwo = \emptyset$.

To prove the general property,
recall by \rprop{prj_equiv_prjv}
that $\judgPRJ[\dom{\tctxtwo}]{\tctx,\tctxtwo}{\tm}{\ev}$
if and only if
$\judgPRJV[\dom{\tctxtwo}]{\tctx,\tctxtwo}{\tm}{\ev}$
and we proceed by induction on the derivation
of $\judgV[\dom{\tctxtwo}]{\tctx,\tctxtwo}{\tm}{\ev}$ in $\PRJV$.
The more interesting cases are
the axiom rule (\Ax)
and the positive weak introduction and elimination (\Icp, \Ecp).
\begin{enumerate}
\item $\Ax\PRJVsym$:
  Let $\judgPRJV[\dom{\tctxtwo}]{\tctx,\tctxtwo}{\var}{\ev}$
  where $(\var:\ev) \in (\tctx,\tctxtwo)$
  and $\var \notin \dom{\tctxtwo}$.
  Hence $(\var:\ev) \in \tctx$
  and $\intem{\ev} \in \intem{\tctx}$,
  so $\logNJ{\intem{\tctx}}{\intem{\ev}}$ by \NDAx.
\item $\Abs\PRJVsym$:
  Let $\judgPRJV[\dom{\tctxtwo}]{\tctx,\tctxtwo}{\strongabs{\ev}{\tm}{\tmtwo}}{\ev}$
  be derived from $\judgPRJV[\dom{\tctxtwo}]{\tctx,\tctxtwo}{\tm}{\typ\pp}$
  and $\judgPRJV[\dom{\tctxtwo}]{\tctx,\tctxtwo}{\tmtwo}{\typ\nn}$.
  Then:
  \[
    \indruleN{\NDExpl}{
      \indruleN{\NDEnot}{
        \indih{\log{\intem{\tctx}}{\neg\typ}}
        \HS
        \indih{\log{\intem{\tctx}}{\typ}}
      }{
        \log{\intem{\tctx}}{\bot}
      }
    }{
      \log{\intem{\tctx}}{\intem{\ev}}
    }
  \]
\item $\Icp\PRJVsym$:
  Let $\judgPRJV[\dom{\tctxtwo}]{\tctx,\tctxtwo}{\claslamp{(\var:\typ\NN)}{\tm}}{\typ\PP}$
  be derived from $\judgPRJV[\dom{\tctxtwo}\cup\set{\var}]{\tctx,\tctxtwo,\var:\typ\NN}{\tm}{\typ\pp}$.
  By \ih we have that $\logNJ{\intem{\tctx}}{\typ}$, exactly as required.
\item $\Icn\PRJVsym$:
  Let $\judgPRJV[\dom{\tctxtwo}]{\tctx,\tctxtwo}{\claslamp{(\var:\typ\PP)}{\tm}}{\typ\NN}$
  be derived from $\judgPRJV[\dom{\tctxtwo}]{\tctx,\tctxtwo,\var:\typ\PP}{\tm}{\typ\nn}$.
  Then:
  \[
    \indruleN{\NDInot}{
      \indruleN{\NDEnot}{
        \indNih{\logNJ{\intem{\tctx},\typ}{\neg\typ}}
        \HS
        \indNax[\NDAx]{\logNJ{\intem{\tctx},\typ}{\typ}}
      }{
        \logNJ{\intem{\tctx},\typ}{\bot}
      }
    }{
      \logNJ{\intem{\tctx}}{\neg\typ}
    }
  \]
\item $\Ecp\PRJVsym$:
  Let $\judgPRJV[\dom{\tctxtwo}]{\tctx,\tctxtwo}{\clasapp{\tm}{\tmtwo}}{\typ\pp}$
  be derived from $\judgPRJV[\dom{\tctxtwo}]{\tctx,\tctxtwo}{\tm}{\typ\PP}$
  and $\judgPRK{\tctx,\tctxtwo}{\tmtwo}{\typ\NN}$.
  By \ih on the first premise, we have that $\logNJ{\intem{\tctx}}{\typ}$,
  exactly as required.
\item $\Ecn\PRJVsym$:
  Let $\judgPRJV[\dom{\tctxtwo}]{\tctx,\tctxtwo}{\clasapn{\tm}{\tmtwo}}{\typ\nn}$
  be derived from $\judgPRJV[\dom{\tctxtwo}]{\tctx,\tctxtwo}{\tm}{\typ\NN}$
  and $\judgPRJV[\dom{\tctxtwo}]{\tctx,\tctxtwo}{\tmtwo}{\typ\PP}$.
  By \ih on the first premise, we have that $\logNJ{\intem{\tctx}}{\neg\typ}$,
  exactly as required.
\item $\colorand{\Iandp\PRJVsym}$:
  Let $\judgPRJV[\dom{\tctxtwo}]{\tctx,\tctxtwo}{\pairp{\tm}{\tmtwo}}{(\typ \land \typtwo)\pp}$
  be derived from $\judgPRJV[\dom{\tctxtwo}]{\tctx,\tctxtwo}{\tm}{\typ\PP}$
  and $\judgPRJV[\dom{\tctxtwo}]{\tctx,\tctxtwo}{\tmtwo}{\typtwo\PP}$.
  By \ih, $\logNJ{\intem{\tctx}}{\typ}$
  and $\logNJ{\intem{\tctx}}{\typtwo}$,
  which imply $\logNJ{\intem{\tctx}}{\typ\land\typtwo}$
  by \NDIand.
\item $\colorand{\Iorn\PRJVsym}$:
  Let $\judgPRJV[\dom{\tctxtwo}]{\tctx,\tctxtwo}{\pairn{\tm}{\tmtwo}}{(\typ \lor \typtwo)\nn}$
  be derived from $\judgPRJV[\dom{\tctxtwo}]{\tctx,\tctxtwo}{\tm}{\typ\NN}$
  and $\judgPRJV[\dom{\tctxtwo}]{\tctx,\tctxtwo}{\tmtwo}{\typtwo\NN}$.
  By \ih, $\logNJ{\intem{\tctx}}{\neg\typ}$
  and $\logNJ{\intem{\tctx}}{\neg\typtwo}$.
  Let $\deriv_1$ be the derivation:
  \[
    \indruleN{\NDEnot}{
      \indruleN{\NDW}{
        \indNih{\log{\intem{\tctx}}{\neg\typ}}
      }{
        \log{\intem{\tctx},\typ\lor\typtwo,\typ}{\neg\typ}
      }
      \HS
      \indNax{\log{\intem{\tctx},\typ\lor\typtwo,\typ}{\typ}}
    }{
      \log{\intem{\tctx},\typ\lor\typtwo,\typ}{\bot}
    }
  \]
  and, symmetrically, let $\deriv_2$ be a derivation of
  $\log{\intem{\tctx},\typ\lor\typtwo,\typtwo}{\bot}$.
  Then:
  \[
    \indruleN{\NDInot}{
      \indruleN{\NDEor}{
        \indNax{\log{\intem{\tctx},\typ\lor\typtwo}{\typ\lor\typtwo}}
        \HS
        \derivdots{\deriv_1}
        \HS
        \derivdots{\deriv_2}
        \HS
      }{
        \log{\intem{\tctx},\typ\lor\typtwo}{\bot}
      }
    }{
      \log{\intem{\tctx}}{\neg(\typ\lor\typtwo)}
    }
  \]
\item $\colorand{\Eandp\PRJVsym}$:
  Let $\judgPRJV[\dom{\tctxtwo}]{\tctx,\tctxtwo}{\projip{\tm}}{\typ_i\PP}$
  be derived from $\judgPRJV[\dom{\tctxtwo}]{\tctx,\tctxtwo}{\tm}{(\typ_1 \land \typ_2)\pp}$
  for some $i \in \set{1, 2}$.
  By \ih, $\logNJ{\intem{\tctx}}{\typ_1 \land \typ_2}$,
  which implies $\logNJ{\intem{\tctx}}{\typ_i}$
  by \NDEand.
\item $\colorand{\Eorn\PRJVsym}$:
  Let $\judgPRJV[\dom{\tctxtwo}]{\tctx,\tctxtwo}{\projin{\tm}}{\typ_i\NN}$
  be derived from $\judgPRJV[\dom{\tctxtwo}]{\tctx,\tctxtwo}{\tm}{(\typ_1 \lor \typ_2)\nn}$
  for some $i \in \set{1, 2}$.
  By \ih, $\logNJ{\intem{\tctx}}{\neg(\typ_1 \lor \typ_2)}$.
  Then:
  \[
    \indruleN{\NDInot}{
      \indruleN{\NDEnot}{
        \indruleN{\NDW}{
          \indNih{\log{\intem{\tctx}}{\neg(\typ_1\lor\typ_2)}}
        }{
          \log{\intem{\tctx},\typ_i}{\neg(\typ_1\lor\typ_2)}
        }
        \HS
        \indruleN{\NDIor}{
          \indNax{\log{\intem{\tctx},\typ_i}{\typ_i}}
        }{
          \log{\intem{\tctx},\typ_i}{\typ_1\lor\typ_2}
        }
        \HS
      }{
        \log{\intem{\tctx},\typ_i}{\bot}
      }
    }{
      \log{\intem{\tctx}}{\neg\typ_i}
    }
  \]
\item $\colorand{\Iorp\PRJVsym}$:
  Let $\judgPRJV[\dom{\tctxtwo}]{\tctx,\tctxtwo}{\inip{\tm}}{(\typ_1 \lor \typ_2)\pp}$
  be derived from $\judgPRJV[\dom{\tctxtwo}]{\tctx,\tctxtwo}{\tm}{\typ_i\PP}$
  for some $i \in \set{1, 2}$.
  By \ih, $\logNJ{\intem{\tctx}}{\typ_i}$,
  which implies $\logNJ{\intem{\tctx}}{\typ_1\lor\typ_2}$
  by \NDIor.
\item $\colorand{\Iandn\PRJVsym}$:
  Let $\judgPRJV[\dom{\tctxtwo}]{\tctx,\tctxtwo}{\inin{\tm}}{(\typ_1 \land \typ_2)\nn}$
  be derived from $\judgPRJV[\dom{\tctxtwo}]{\tctx,\tctxtwo}{\tm}{\typ_i\NN}$
  for some $i \in \set{1, 2}$.
  By \ih, $\logNJ{\intem{\tctx}}{\neg\typ_i}$. Then:
  \[
    \indruleN{\NDInot}{
      \indruleN{\NDEnot}{
        \indruleN{\NDW}{
          \indNih{\log{\intem{\tctx}}{\neg\typ_i}}
        }{
          \log{\intem{\tctx},\typ_1\land\typ_2}{\neg\typ_i}
        }
        \HS
        \indruleN{\NDEand}{
          \indNax{\log{\intem{\tctx},\typ_1\land\typ_2}{\typ_1\land\typ_2}}
        }{
          \log{\intem{\tctx},\typ_1\land\typ_2}{\typ_i}
        }
      }{
        \log{\intem{\tctx},\typ_1\land\typ_2}{\bot}
      }
    }{
      \log{\intem{\tctx}}{\neg(\typ_1\land\typ_2)}
    }
  \]
\item $\colorand{\Eorp\PRJVsym}$:
  Let
  $\judgPRJV[\dom{\tctxtwo}]{\tctx,\tctxtwo}{\casep{\tm}{\var:\typ\PP}{\tmtwo}{\vartwo:\typtwo\PP}{\tmthree}}{\ev}$
  be derived from $\judgPRJV[\dom{\tctxtwo}]{\tctx,\tctxtwo}{\tm}{(\typ \lor \typtwo)\pp}$
  and $\judgPRJV[\dom{\tctxtwo}]{\tctx,\tctxtwo,\var:\typ\PP}{\tmtwo}{\ev}$
  and $\judgPRJV[\dom{\tctxtwo}]{\tctx,\tctxtwo,\vartwo:\typtwo\PP}{\tmthree}{\ev}$.
  By \ih, $\logNJ{\intem{\tctx}}{\typ \lor \typtwo}$
  and $\logNJ{\intem{\tctx},\typ}{\intem{\ev}}$
  and $\logNJ{\intem{\tctx},\typtwo}{\intem{\ev}}$,
  which imply $\logNJ{\intem{\tctx}}{\intem{\ev}}$
  by \NDEor.
\item $\colorimp{\Iimpp\PRJVsym}$:
  Let $\judgPRJV[\dom{\tctxtwo}]{\tctx,\tctxtwo}{\lamp{(\var:\typ\PP)}{\tm}}{(\typ\imp\typtwo)\pp}$
  be derived from $\judgPRJV[\dom{\tctxtwo}]{\tctx,\tctxtwo,\var:\typ\PP}{\tm}{\typtwo\PP}$.
  By \ih, $\logNJ{\intem{\tctx},\typ}{\typtwo}$,
  which implies $\logNJ{\intem{\tctx}}{\typ\imp\typtwo}$
  by \NDIimp.
\item $\colorimp{\Icoimpn\PRJVsym}$:
  Let $\judgPRJV[\dom{\tctxtwo}]{\tctx,\tctxtwo}{\lamn{(\var:\typ\NN)}{\tm}}{(\typ\coimp\typtwo)\nn}$
  be derived from $\judgPRJV[\dom{\tctxtwo}]{\tctx,\tctxtwo,\var:\typ\NN}{\tm}{\typtwo\NN}$.
  Let $\deriv$ be the derivation:
  \[
    \indruleN{\NDCut}{
      \indruleN{\NDW}{
        \indNih{\log{\intem{\tctx},\neg\typ}{\neg\typtwo}}
      }{
        \log{\intem{\tctx},\typ\coimp\typtwo,\neg\typ}{\neg\typtwo}
      }
      \HS
      \indruleN{\NDEcoimp}{
        \indNax[\NDAx]{\log{\intem{\tctx},\typ\coimp\typtwo}{\typ\coimp\typtwo}}
        \HS
        \indNax{\log{\intem{\tctx},\typ\coimp\typtwo,\neg\typ,\typtwo}{\neg\typ}}
      }{
        \log{\intem{\tctx},\typ\coimp\typtwo}{\neg\typ}
      }
    }{
      \log{\intem{\tctx},\typ\coimp\typtwo}{\neg\typtwo}
    }
  \]
  Then:
  \[
    \indruleN{\NDInot}{
      \indruleN{\NDEnot}{
        \derivdots{\deriv}
        \HS
        \indruleN{\NDEcoimp}{
          \indNax[\NDAx]{\log{\intem{\tctx},\typ\coimp\typtwo}{\typ\coimp\typtwo}}
          \HS
          \indNax[\NDAx]{\log{\intem{\tctx},\typ\coimp\typtwo,\neg\typ,\typtwo}{\typtwo}}
        }{
          \log{\intem{\tctx},\typ\coimp\typtwo}{\typtwo}
        }
      }{
        \log{\intem{\tctx},\typ\coimp\typtwo}{\bot}
      }
    }{
      \log{\intem{\tctx}}{\neg(\typ\coimp\typtwo)}
    }
  \]
\item $\colorimp{\Eimpp\PRJVsym}$:
  Let $\judgPRJV[\dom{\tctxtwo}]{\tctx,\tctxtwo}{\app{\tm}{\tmtwo}}{\typtwo\PP}$
  be derived from $\judgPRJV[\dom{\tctxtwo}]{\tctx,\tctxtwo}{\tm}{(\typ\imp\typtwo)\pp}$
  and $\judgPRJV[\dom{\tctxtwo}]{\tctx,\tctxtwo}{\tmtwo}{\typ\PP}$.
  By \ih, $\logNJ{\intem{\tctx}}{\typ\imp\typtwo}$
  and $\logNJ{\intem{\tctx}}{\typ}$,
  which imply $\logNJ{\intem{\tctx}}{\typtwo}$
  by \NDEimp.
\item $\colorimp{\Ecoimpn\PRJVsym}$:
  Let $\judgPRJV[\dom{\tctxtwo}]{\tctx,\tctxtwo}{\apn{\tm}{\tmtwo}}{\typtwo\NN}$
  be derived from $\judgPRJV[\dom{\tctxtwo}]{\tctx,\tctxtwo}{\tm}{(\typ\coimp\typtwo)\nn}$
  and $\judgPRJV[\dom{\tctxtwo}]{\tctx,\tctxtwo}{\tmtwo}{\typ\NN}$.
  Let $\deriv$ be the derivation:
  \[
    \indruleN{\NDIcoimp}{
      \indruleN{\NDW}{
        \indih{\log{\intem{\tctx}}{\neg\typ}}
      }{
        \log{\intem{\tctx},\typtwo}{\neg\typ}
      }
      \HS
      \indNax[\NDAx]{\log{\intem{\tctx},\typtwo}{\typtwo}}
    }{
      \log{\intem{\tctx},\typtwo}{\typ\coimp\typtwo}
    }
  \]
  Then:
  \[
    \indruleN{\NDInot}{
      \indruleN{\NDEnot}{
        \indruleN{\NDW}{
          \indih{\log{\intem{\tctx}}{\neg(\typ\coimp\typtwo)}}
        }{
          \log{\intem{\tctx},\typtwo}{\neg(\typ\coimp\typtwo)}
        }
        \HS
        \derivdots{\deriv}
        \HS
      }{
        \log{\intem{\tctx},\typtwo}{\bot}
      }
    }{
      \log{\intem{\tctx}}{\neg\typtwo}
    }
  \]
\item $\colorimp{\Icoimpp\PRJVsym}$:
  Let $\judgPRJV[\dom{\tctxtwo}]{\tctx,\tctxtwo}{\copairp{\tm}{\tmtwo}}{(\typ\coimp\typtwo)\pp}$
  be derived from $\judgPRJV[\dom{\tctxtwo}]{\tctx,\tctxtwo}{\tm}{\typ\NN}$
  and $\judgPRJV[\dom{\tctxtwo}]{\tctx,\tctxtwo}{\tmtwo}{\typtwo\PP}$.
  By \ih, $\logNJ{\intem{\tctx}}{\neg\typ}$
  and $\logNJ{\intem{\tctx}}{\typtwo}$,
  which imply $\logNJ{\intem{\tctx}}{\typ\coimp\typtwo}$
  by \NDIcoimp.
\item $\colorimp{\Iimpn\PRJVsym}$:
  Let $\judgPRJV[\dom{\tctxtwo}]{\tctx,\tctxtwo}{\copairn{\tm}{\tmtwo}}{(\typ\imp\typtwo)\nn}$
  be derived from $\judgPRJV[\dom{\tctxtwo}]{\tctx,\tctxtwo}{\tm}{\typ\PP}$
  and $\judgPRJV[\dom{\tctxtwo}]{\tctx,\tctxtwo}{\tmtwo}{\typtwo\NN}$.
  Let $\deriv$ be the derivation:
  \[
    \indruleN{\NDEimp}{
      \indNax[\NDAx]{\log{\intem{\tctx},\typ\imp\typtwo}{\typ\imp\typtwo}}
      \HS
      \indruleN{\NDW}{
        \indNih{\log{\intem{\tctx}}{\typ}}
      }{
        \log{\intem{\tctx},\typ\imp\typtwo}{\typ}
      }
    }{
      \log{\intem{\tctx},\typ\imp\typtwo}{\typtwo}
    }
  \]
  Then:
  \[
    \indruleN{\NDInot}{
      \indruleN{\NDEnot}{
        \indruleN{\NDW}{
          \indNih{\log{\intem{\tctx}}{\neg\typtwo}}
        }{
          \log{\intem{\tctx},\typ\imp\typtwo}{\neg\typtwo}
        }
        \HS
        \derivdots{\deriv}
        \HS
      }{
        \log{\intem{\tctx},\typ\imp\typtwo}{\bot}
      }
    }{
      \log{\intem{\tctx}}{\neg(\typ\imp\typtwo)}
    }
  \]
\item $\colorimp{\Ecoimpp\PRJVsym}$:
  Let $\judgPRJV[\dom{\tctxtwo}]{\tctx,\tctxtwo}{\colamp{\tm}{\var}{\vartwo}{\tmtwo}}{\ev}$
  be derived from
  $\judgPRJV[\dom{\tctxtwo}]{\tctx,\tctxtwo}{\tm}{(\typ\coimp\typtwo)\pp}$
  and
  $\judgPRJV[\dom{\tctxtwo}]{\tctx,\tctxtwo,\var:\typ\NN,\vartwo:\typtwo\PP}{\tmtwo}{\ev}$.
  By \ih, $\logNJ{\intem{\tctx}}{\typ\coimp\typtwo}$
  and $\logNJ{\intem{\tctx},\neg\typ,\typtwo}{\intem{\ev}}$
  which implies $\logNJ{\intem{\tctx}}{\intem{\ev}}$
  by \NDEcoimp.
\item $\colorneg{\Inotp\PRJVsym}$:
  Let $\judgPRJV[\dom{\tctxtwo}]{\tctx,\tctxtwo}{\negip{\tm}}{(\neg\typ)\pp}$
  be derived from $\judgPRJV[\dom{\tctxtwo}]{\tctx,\tctxtwo}{\tm}{\typ\NN}$.
  By \ih, $\logNJ{\intem{\tctx}}{\neg\typ}$, exactly as required.
\item $\colorneg{\Inotn\PRJVsym}$:
  Let $\judgPRJV[\dom{\tctxtwo}]{\tctx,\tctxtwo}{\negin{\tm}}{(\neg\typ)\nn}$
  be derived from $\judgPRJV[\dom{\tctxtwo}]{\tctx,\tctxtwo}{\tm}{\typ\PP}$.
  By \ih, $\logNJ{\intem{\tctx}}{\typ}$. Then:
  \[
    \indruleN{\NDInot}{
      \indruleN{\NDEnot}{
        \indNax{\log{\intem{\tctx},\neg\typ}{\neg\typ}}
        \HS
        \indruleN{\NDW}{
          \indNih{\log{\intem{\tctx}}{\typ}}
        }{
          \log{\intem{\tctx},\neg\typ}{\typ}
        }
      }{
        \log{\intem{\tctx},\neg\typ}{\bot}
      }
    }{
      \log{\intem{\tctx}}{\neg\neg\typ}
    }
  \]
\item $\colorneg{\Enotp\PRJVsym}$:
  Let $\judgPRJV[\dom{\tctxtwo}]{\tctx,\tctxtwo}{\negep{\tm}}{\typ\NN}$
  be derived from $\judgPRJV[\dom{\tctxtwo}]{\tctx,\tctxtwo}{\tm}{(\neg\typ)\pp}$.
  By \ih, $\logNJ{\intem{\tctx}}{\neg\typ}$, exactly as required.
\item $\colorall{\Iallp\PRJVsym}$:
  Let
    $\judgPRJV[\dom{\tctxtwo}]{\tctx,\tctxtwo}{\lamtp{\btyp}{\tm}}{(\all{\btyp}{\typ})\pp}$
  be derived from
    $\judgPRJV[\dom{\tctxtwo}]{\tctx,\tctxtwo}{\tm}{\typ\PP}$,
  where $\btyp \notin \ftv{\tctx,\tctxtwo}$.
  By \ih, $\logNJ{\intem{\tctx}}{\typ}$.
  Moreover, note that $\btyp \notin \ftv{\intem{\tctx}}$
  since $\btyp \notin \ftv{\tctx}$.
  Hence by $\NDIall$ we have that $\logNJ{\intem{\tctx}}{\all{\btyp}{\typ}}$.
\item $\colorall{\Iexn\PRJVsym}$:
  Let
    $\judgPRJV[\dom{\tctxtwo}]{\tctx,\tctxtwo}{\lamtn{\btyp}{\tm}}{(\ex{\btyp}{\typ})\nn}$
  be derived from
    $\judgPRJV[\dom{\tctxtwo}]{\tctx,\tctxtwo}{\tm}{\typ\NN}$,
  where $\btyp \not\in \ftv{\tctx}$.
  By \ih, $\logNJ{\intem{\tctx}}{\neg\typ}$.
  Moreover, note that $\btyp \notin \ftv{\intem{\tctx}}$.
  Let $\deriv$ be the derivation:
  \[
    \indruleN{\NDEex}{
      \indNax{\log{\intem{\tctx},\ex{\btyp}{\typ}}{\ex{\btyp}{\typ}}}
      \HS
      \btyp\notin\ftv{\intem{\tctx},\ex{\btyp}{\typ}}
      \HS
    }{
      \log{\intem{\tctx},\ex{\btyp}{\typ}}{\typ}
    }
  \]
  Then:
  \[
    \indruleN{\NDInot}{
      \indruleN{\NDEnot}{
        \indruleN{\NDW}{
          \indNih{\log{\intem{\tctx}}{\neg\typ}}
        }{
          \log{\intem{\tctx},\ex{\btyp}{\typ}}{\neg\typ}
        }
        \HS
        \derivdots{\deriv}
        \HS
      }{
        \log{\intem{\tctx},\ex{\btyp}{\typ}}{\bot}
      }
    }{
      \log{\intem{\tctx}}{\neg\ex{\btyp}{\typ}}
    }
  \]
\item $\colorall{\Eallp\PRJVsym}$:
  Let
    $\judgPRJV[\dom{\tctxtwo}]{\tctx,\tctxtwo}{\apptp{\tm}{\typ}}{\typtwo\PP\sub{\btyp}{\typ}}$
  be derived from
    $\judgPRJV[\dom{\tctxtwo}]{\tctx,\tctxtwo}{\tm}{(\all{\btyp}{\typtwo})\pp}$.
  By \ih, $\logNJ{\intem{\tctx}}{\all{\btyp}{\typtwo}}$,
  which implies
    $\logNJ{\intem{\tctx}}{\typtwo\sub{\btyp}{\typ}}$
  by $\NDEall$.
\item $\colorall{\Eexn\PRJVsym}$:
  Let
    $\judgPRJV[\dom{\tctxtwo}]{\tctx,\tctxtwo}{\apptn{\tm}{\typ}}{\typtwo\NN\sub{\btyp}{\typ}}$
  be derived from
    $\judgPRJV[\dom{\tctxtwo}]{\tctx,\tctxtwo}{\tm}{(\ex{\btyp}{\typtwo})\nn}$.
  By \ih, $\logNJ{\intem{\tctx}}{\neg\ex{\btyp}{\typtwo}}$.
  Let $\deriv$ be the derivation:
  \[
    \indruleN{\NDIex}{
      \indNax{
        \log{\intem{\tctx},\typtwo\sub{\btyp}{\typ}}{\typtwo\sub{\btyp}{\typ}}
      }
    }{
      \log{\intem{\tctx},\typtwo\sub{\btyp}{\typ}}{\ex{\btyp}{\typtwo}}
    }
  \]
  Then:
  \[
    \indruleN{\NDInot}{
      \indruleN{\NDEnot}{
        \indruleN{\NDW}{
          \indNih{\log{\intem{\tctx}}{\neg\ex{\btyp}{\typtwo}}}
        }{
          \log{\intem{\tctx},\typtwo\sub{\btyp}{\typ}}{\neg\ex{\btyp}{\typtwo}}
        }
        \HS
        \derivdots{\deriv}
        \HS
      }{
        \log{\intem{\tctx},\typtwo\sub{\btyp}{\typ}}{\bot}
      }
    }{
      \log{\intem{\tctx}}{\neg\typtwo\sub{\btyp}{\typ}}
    }
  \]
\item $\colorall{\Iexp\PRJVsym}$:
  Let
    $\judgPRJV[\dom{\tctxtwo}]{\tctx,\tctxtwo}{\patp{\typ}{\tm}}{(\ex{\btyp}{\typtwo})\pp}$
  be derived from
    $\judgPRJV[\dom{\tctxtwo}]{\tctx,\tctxtwo}{\tm}{\typtwo\PP\sub{\btyp}{\typ}}$.
  By \ih,
    $\logNJ{\intem{\tctx}}{\typtwo\sub{\btyp}{\typ}}$,
  which implies
    $\logNJ{\intem{\tctx}}{\ex{\btyp}{\typtwo}}$
  by \NDIex.
\item $\colorall{\Ialln\PRJVsym}$:
  Let
    $\judgPRJV[\dom{\tctxtwo}]{\tctx,\tctxtwo}{\patn{\typ}{\tm}}{(\all{\btyp}{\typtwo})\nn}$
  be derived from
    $\judgPRJV[\dom{\tctxtwo}]{\tctx,\tctxtwo}{\tm}{\typtwo\NN\sub{\btyp}{\typ}}$.
  By \ih,
    $\logNJ{\intem{\tctx}}{\neg\typtwo\sub{\btyp}{\typ}}$.
  Let $\deriv$ be the derivation:
  \[
    \indruleN{\NDEall}{
      \indNax{\log{\intem{\tctx},\all{\btyp}{\typtwo}}{\all{\btyp}{\typtwo}}}
    }{
      \log{\intem{\tctx},\all{\btyp}{\typtwo}}{\typtwo\sub{\btyp}{\typ}}
    }
  \]
  Then:
  \[
    \indruleN{\NDIall}{
      \indruleN{\NDEall}{
        \indruleN{\NDW}{
          \indNih{\log{\intem{\tctx}}{\neg\typtwo\sub{\btyp}{\typ}}}
        }{
          \log{\intem{\tctx},\all{\btyp}{\typtwo}}{\neg\typtwo\sub{\btyp}{\typ}}
        }
        \HS
        \derivdots{\deriv}
        \HS
      }{
        \log{\intem{\tctx},\all{\btyp}{\typtwo}}{\bot}
      }
    }{
      \log{\intem{\tctx}}{\neg\all{\btyp}{\typtwo}}
    }
  \]
\item $\colorall{\Eexp\PRJVsym}$:
  Let
    $\judgPRJV[\dom{\tctxtwo}]{\tctx,\tctxtwo}{\optp{\tm}{\btyp}{\var}{\tmtwo}}{\ev}$
  be derived from
    $\judgPRJV[\dom{\tctxtwo}]{\tctx,\tctxtwo}{\tm}{(\ex{\btyp}{\typ})\pp}$
  and
    $\judgPRJV[\dom{\tctxtwo}]{\tctx,\tctxtwo,\var:\typ\PP}{\tmtwo}{\ev}$,
  where
    $\btyp \not\in \ftv{\tctx,\tctxtwo,\ev}$.
  By \ih we have that
    $\logNJ{\intem{\tctx}}{\ex{\btyp}{\typ}}$
  and
    $\logNJ{\intem{\tctx},\typ}{\intem{\ev}}$.
  Moreover, note that $\btyp\notin\ftv{\intem{\tctx},\intem{\ev}}$.
  Hence by $\NDEex$ we have $\logNJ{\intem{\tctx}}{\intem{\ev}}$,
  as required.
\end{enumerate}
\end{proof}

\begin{lemma}[Intuitionistic Embedding]
\llem{lambdaPRJ_embedding}
If $\logNJ{\typ_1,\hdots,\typ_n}{\typtwo}$
there exists a term $\tm$
such that
$\judgPRJ{\var_1:\typ_1\PP,\hdots,\var_n:\typ_n\PP}{\tm}{\typtwo\PP}$.
\end{lemma}
\begin{proof}
We proceed by induction of the derivation
of the judgment $\logNJ{\typ_1,\hdots,\typ_n}{\typtwo}$
in intuitionistic natural deduction.
The construction of the witnesses is, in almost all cases,
the same as in the classical case~(\rlem{lambdaPRK_embedding}).
Here we only check that these witnesses are actually intuitionistic terms.
\begin{enumerate}
\item \NDAx:
  Let $\log{\typ_1,\hdots,\typ_n}{\typ_i}$ be derived from the $\NDAx$ rule.
  Then $\judgPRK{\var_1:\typ_1\PP,\hdots,\var_n:\typ_n\PP}{\var_i}{\typ_i\PP}$
  by the $\Ax$ rule, and $\var_i$ is intuitionistic
  so $\judgPRJ{\var_1:\typ_1\PP,\hdots,\var_n:\typ_n\PP}{\var_i}{\typ_i\PP}$.
\item \colorand{\NDIand}:
  Recall that:
  \[
    \pairj{\tm}{\tmtwo}
    \eqdef
    \claslamp{\under:(\typ\land\typtwo)\NN}{
      \pairp{\tm}{\tmtwo}
    }
  \]
  There are no free occurrences of the negative counterfactual in
  $\tm$ nor in $\tmtwo$,
  so in particular there are
  no free useful occurrences of the negative counterfactual
  in $\pairp{\tm}{\tmtwo}$.
  Hence $\pairj{\tm}{\tmtwo}$ is intuitionistic
  and $\judgPRJ{\tctx}{\pairj{\tm}{\tmtwo}}{(\typ\land\typtwo)\PP}$.
\item \colorand{\NDEand}:
  Recall that $\projij{\tm}$ is defined by:
  \[
    \claslamp{(\var:\typ_i\NN)}{
      \clasapp{
        \projip{
          \clasapp{
            \tm
          }{
            \claslamn{(\under:(\typ_1\land\typ_2)\PP)}{
              \inin{\var}
            }
          }
        }
      }{
        \var
      }
    }
  \]
  There are only two free occurrences
  of the negative counterfactual $\var$
  in $\clasapp{
        \projip{
          \clasapp{
            \tm
          }{
            \claslamn{(\under:(\typ_1\land\typ_2)\PP)}{
              \inin{\var}
            }
          }
        }
      }{
        \var
      }$,
  both of them useless.
  Hence $\projij{\tm}$ is intuitionistic
  and $\judgPRJ{\tctx}{\projij{\tm}}{\typ_i\PP}$.
\item \colorand{\NDIor}:
  Recall that:
  \[
    \inij{\tm} \eqdef
    \claslamp{(\under:(\typ_1\lor\typ_2)\NN)}{
      \inip{
        \tm
      }
    }
  \]
  There are no free occurrences of the negative counterfactual
  in $\tm$, so in particular there are no useful free occurrences of the
  negative counterfactual in $\inip{\tm}$.
  Hence $\inij{\tm}$ is intuitionistic and
  $\judgPRJ{\tctx}{\inij{\tm}}{(\typ_1\lor\typ_2)\PP}$.
\item \colorand{\NDEor}:
  Recall that
    $\casej{\tm}{(\var:\typ\PP)}{\tmtwo}{(\var:\typtwo\PP)}{\tmthree}$
  is defined as follows,
  where the two contrapositions are indeed intuitionistic:
  \[
    \claslamp{(\vartwo:\typthree\NN)}{
        \caseptable{
          (\clasapp{
            \tm
          }{
            \claslamn{(\under:(\typ\lor\typtwo)\PP)}{
              \pairn{
                \icontrapose{\var}{\vartwo}{\tm}
              }{
                \icontrapose{\var}{\vartwo}{\tmtwo}
              }
            }
          })
        }{(\var:\typ\PP)}{
          \clasapp{
            \tmtwo
          }{
            \vartwo
          }
        }{(\var:\typtwo\PP)}{
          \clasapp{
            \tmthree
          }{
            \vartwo
          }
        }
    }
  \]
  All the occurrences of the negative counterfactual
  $\vartwo$ in the body are useless.
  Hence $\casej{\tm}{(\var:\typ\PP)}{\tmtwo}{(\var:\typtwo\PP)}{\tmthree}$
  is intuitionistic and
  $\judgPRJ{\tctx}{\casej{\tm}{(\var:\typ\PP)}{\tmtwo}{(\var:\typtwo\PP)}{\tmthree}}{\typthree\PP}$.
\item
  \colorimp{\NDIimp}:
  Recall that:
  \[
    \lamj{(\var:\typ\PP)}{\tm} \eqdef
    \claslamp{(\under:(\typ\to\typtwo)\NN)}{
      \lamp{(\var:\typ\PP)}{
        \tm
      }
    }
  \]
  There are no free occurrences of the negative counterfactual
  in $\lamp{(\var:\typ\PP)}{\tm}$.
  Hence $\lamj{(\var:\typ\PP)}{\tm}$ is intuitionistic and
  $\judgPRJ{\tctx}{\lamj{(\var:\typ\PP)}{\tm}}{(\typ\imp\typtwo)\PP}$.
\item \colorimp{\NDEimp}:
  Recall that $\appj{\tm}{\tmtwo}$ is defined by:
  \[
    \claslamp{(\var:\typtwo\NN)}{
      (\clasapp{
        \app{
          \clasapp{
            \tm
          }{
            (\claslamn{(\under:(\typ\to\typtwo)\PP)}{
              \copairn{
                \tmtwo
              }{
                \var
              }
            })
          }
        }{
          \tmtwo
        }
      }{
        \var
      })
    }
  \]
  The occurrences of the negative counterfactual $\var$
  are useless.
  Hence $\appj{\tm}{\tmtwo}$ is intuitionistic
  and $\judgPRJ{\tctx}{\appj{\tm}{\tmtwo}}{\typtwo\PP}$.
\item \colorimp{\NDIcoimp}:
  Recall that:
  \[
    \copairj{\tm}{\tmtwo}
    \eqdef
    \claslamp{(\under:(\typ\coimp\typtwo)\NN)}{
      \copairp{
        \negeP{\tm}
      }{
        \tmtwo
      }
    }
  \]
  Note that $\copairj{\tm}{\tmtwo}$ is intuitionistic and
  $\judgPRJ{\tctx}{\copairj{\tm}{\tmtwo}}{(\typ\coimp\typtwo)\PP}$.
\item \colorimp{\NDEcoimp}:
  Recall that $\colamc{\tm}{\var}{\vartwo}{\tmtwo}$ is defined by:
  \[
    \claslamp{\varthree:\typthree\NN}{
      (\clasapp{
        \colamp{
          (\clasapp{
            \tm
          }{
            (\claslamn{\under}{
              \lamn{\var_0:\typ\NN}{
                \claslamn{\vartwo:\typtwo\PP}{
                  (\abs{\typtwo\nn}{
                    \tmtwo'
                  }{
                    \varthree
                  })
                }
              }
            })
          })
        }{\var_0:\typ\NN}{\vartwo:\typtwo\PP}{
          \tmtwo'
        }
      }{
        \varthree
      })
    }
  \]
  where $\tmtwo' \eqdef \tmtwo\sub{\var}{\negiP{\var_0}}$.
  The occurrences of the negative counterfactual $\varthree$ are useless,
  and there are no other negative counterfactuals.
  Hence $\colamc{\tm}{\var}{\vartwo}{\tmtwo}$ is intuitionistic
  and $\judgPRJ{\tctx}{\colamc{\tm}{\var}{\vartwo}{\tmtwo}}{(\neg\typ)\PP}$.
\item \colorneg{\NDInot}:
  Recall that:
  \[
    \neglamj{\var:\typ\PP}{\tm} \eqdef
    \claslamp{\under:(\neg\typ)\NN}{
      \negip{(
        \icontrapose{\var}{\vartwo}{\tm}
          \sub{\vartwo}{\lemN{\btyp_0}}
      )}
    }
  \]
  where the contraposition is indeed intuitionistic.
  There are no occurrences of the
  negative counterfactual.
  Hence $\neglamj{\var:\typ\PP}{\tm}$ is intuitionistic
  and $\judgPRJ{\tctx}{\neglamj{\var:\typ\PP}{\tm}}{(\neg\typ)\PP}$.
\item \colorneg{\NDEnot}:
  Recall that:
  \[
    \negapj{\tm}{\tmtwo} \eqdef
    \abs{\bot\PP}{
      \tm
    }{
      \claslamn{(\under:\typ\PP)}{
        \negin{
          \tmtwo
        }
      }
    }
  \]
  Note that $\negapj{\tm}{\tmtwo}$
  is intuitionistic, so $\judgPRJ{\tctx}{\negapj{\tm}{\tmtwo}}{\bot\PP}$.
\item \colorall{\NDIall}:
  Recall that:
  \[
    \lamtj{\btyp}{\tm} \eqdef
    \claslamp{(\under:(\all{\btyp}{\typ})\NN)}{
      \lamtp{\btyp}{
        \tm
      }
    }
  \]
  Note that $\lamtj{\btyp}{\tm}$ is intuitionistic,
  so
  $\judgPRJ{\tctx}{\lamtj{\btyp}{\tm}}{(\all{\btyp}{\typ})\PP}$.
\item \colorall{\NDEall}:
  Recall that $\apptj{\tm}{\typ}$ is defined by:
  \[
    \claslamp{(\var:(\typtwo\sub{\btyp}{\typ})\NN)}{
      (\clasapp{
        \apptp{
          \tm'
        }{
          \typ
        }
      }{
        \var
      })
    }
  \]
  where
    $\tm' =
          \clasapp{
            \tm
          }{
            \claslamp{(\under:(\all{\btyp}{\typtwo})\PP)}{
              \patn{
                \typ
              }{
                \var
              }
            }
          }$.
  All the occurrences of the negative counterfactual $\var$
  are useless. Hence $\apptj{\tm}{\typ}$ is intuitionistic
  and $\judgPRJ{\tctx}{\apptj{\tm}{\typ}}{\typtwo\sub{\btyp}{\typ}\PP}$.
\item \colorall{\NDIex}:
  Recall that:
  \[
    \patj{\typ}{\tm} \eqdef
    \claslamp{(\under:(\ex{\btyp}{\typtwo})\NN)}{
      \patp{\typ}{
        \tm
      }
    }
  \]
  Note that $\patj{\typ}{\tm}$ is intuitionistic,
  so $\judgPRJ{\tctx}{\patj{\typ}{\tm}}{(\ex{\btyp}{\typtwo})\PP}$.
\item \colorall{\NDEex}:
  Recall that $\optj{\tm}{\btyp}{\var}{\tmtwo}$ is defined by:
  \[
    \claslamp{(\vartwo:\typtwo\NN)}{
      (\clasapp{
        \optp{
          \tm'
        }{\btyp}{\var}{
          \tmtwo
        }
      }{
        \vartwo
      })
    }
  \]
  where $\tm' = \clasapp{
            \tm
          }{
            \claslamn{(\under:(\ex{\btyp}{\typ})\PP)}{
              \lamtn{\btyp}{
                \icontrapose{\var}{\vartwo}{\tmtwo}
              }
            }
          }$,
  and where the contraposition is indeed intuitionistic.
  All the occurrences of the negative counterfactual
  $\vartwo$ are useless.
  Hence $\optj{\tm}{\btyp}{\var}{\tmtwo}$ is intuitionistic
  and $\judgPRJ{\tctx}{\optj{\tm}{\btyp}{\var}{\tmtwo}}{\typtwo\PP}$.
\item \NDExpl:
  Note that $\abs{\typ\PP}{\tm}{\lemN{\btyp_0}}$ is intuitionistic,
  so $\judgPRJ{\tctx}{\abs{\typ\PP}{\tm}{\lemN{\btyp_0}}}{\typ\PP}$.
\end{enumerate}
\end{proof}

\newpage
\section{Canonicity}
\lsec{appendix:canonicity}

Recall that {\em neutral terms} ($\neu,\hdots$) and
{\em normal terms} ($\nf,\hdots$) are:
\[
\begin{array}{r@{\,}r@{\,}l}
  \neu & ::=  &
       \var
  \mid \strongabs{\ev}{\neu}{\nf}
  \mid \strongabs{\ev}{\nf}{\neu}
  \mid \colorand{\projipn{\neu}}
  \mid \colorand{\casepn{\neu}{\var}{\nf}{\var}{\nf}}
  \\
  & \mid &
       \colorimp{\appn{\neu}{\nf}}
  \mid \colorimp{\colam{\neu}{\var}{\vartwo}{\nf}}
  \mid \colorneg{\negepn{\neu}}
  \mid \colorall{\apptpn{\neu}{\typ}}
  \mid \colorall{\optpn{\neu}{\var}{\btyp}{\nf}}
  \mid \clasappn{\neu}{\nf}
\\
  \nf & ::=  &
       \neu
  \mid \colorand{\pairpn{\nf}{\nf}}
  \mid \colorand{\inipn{\nf}}
  \mid \colorimp{\lampn{\var}{\nf}}
  \mid \colorimp{\copairpn{\nf}{\nf}}
  \mid \colorneg{\negipn{\nf}}
  \mid \colorall{\lamtpn{\btyp}{\nf}}
  \mid \colorall{\patpn{\typ}{\nf}}
  \mid \claslampn{\var}{\nf}
\end{array}
\]
A typing context $\tctx = \var_1:\ev_1,\hdots,\var_n:\ev_n$
is {\em weak}
if $\ev_1,\hdots,\ev_n$ are all weak types.
{\em Canonical} terms are those terms built with an introduction rule:
\[
  \colorand{\pairpn{\tm}{\tmtwo}}
  \Hs
  \colorand{\inipn{\tm}}
  \Hs
  \colorimp{\lampn{\var}{\tm}}
  \Hs
  \colorimp{\copairpn{\tm}{\tmtwo}}
  \Hs
  \colorneg{\negipn{\tm}}
  \Hs
  \colorall{\lamtpn{\btyp}{\tm}}
  \Hs
  \colorall{\patpn{\typ}{\tm}}
  \Hs
  \claslampn{\var}{\tm}
\]
A {\em capsule} is
either a variable or a term of the form $\claslampn{\var}{\nf}$.
A {\em critical context} is a context $\crictx$ given by
the following grammar:
\[
  \begin{array}{r@{\,}r@{\,}l}
  \crictx & ::= &
          \ctxhole
    \mid \strongabs{}{\crictx}{\nf}
    \mid \strongabs{}{\nf}{\crictx}
    \mid \colorand{\projipn{\crictx}}
    \mid \colorand{\casepn{\crictx}{\var}{\nf_1}{\vartwo}{\nf_2}}
    \mid \colorimp{\appn{\crictx}{\nf}}
  \\
    & \mid &
         \colorimp{\colampn{\crictx}{\var}{\vartwo}{\nf}}
    \mid \colorneg{\negepn{\crictx}}
    \mid \colorall{\apptpn{\crictx}{\typ}}
    \mid \colorall{\optpn{\crictx}{\var}{\btyp}{\nf}}
    \mid \clasappn{\crictx}{\nf}
    \mid \clasappn{\var}{\crictx}
  \end{array}
\]
A typing context $\tctx = \var_1:\ev_1,\hdots,\var_n:\ev_n$
is {\em weak}
if $\ev_1,\hdots,\ev_n$ are all weak types.

\begin{lemma}[Shape of neutral terms]
\llem{appendix:shape_of_neutral_terms}
Let $\judgPRK{\tctx}{\tm}{\ev}$
where $\tctx$ is weak
and $\tm$ is a neutral term.
Then:
\begin{enumerate}
\item
  If $\ev$ is strong,
  then $\tm$ is of the form $\crictxof{\clasappn{\var}{\tmfive}}$
  where $\tmfive$ is a capsule.
\item
  If $\ev$ is weak,
  then $\tm$ is either a variable
  or of the form $\crictxof{\clasappn{\var}{\tmfive}}$,
  where $\tmfive$ is a capsule.
\end{enumerate}
\end{lemma}
\begin{proof}
We proceed by induction on the derivation that $\tm$ is a neutral term:
\begin{enumerate}
\item $\tm = \var$:
  If $\ev$ is weak, \ie $\ev = \typ\PP$ or $\ev = \typ\NN$,
  then we are done, given that $\tm$ is a variable.
  If $\ev$ is strong, \ie $\ev = \typ^{\pm}$,
  then note that this case is impossible,
  since $\judg{\tctx}{\var}{\typ^{\pm}}$
  must be derived from the $\Ax$ rule,
  so $(\var:\typ^{\pm}) \in \tctx$, contradicting the hypothesis that
  $\tctx$ is weak.
\item $\tm = \strongabs{}{\neu}{\nf}$:
  Note that $\judg{\tctx}{(\strongabs{}{\neu}{\nf})}{\ev}$
  must be derived from the $\Abs$ rule,
  so in particular we must have that $\judg{\tctx}{\neu}{\typ\pp}$.
  By \ih, $\neu$ must be
  of the form $\neu = \crictx'\ctxof{\clasapn{\var}{\tmfive}}$
  where $\tmfive$ is a capsule.
  Then $\tm = \strongabs{}{\crictx'\ctxof{\clasapn{\var}{\tmfive}}}{\nf}$,
  so taking $\crictx := (\strongabs{}{\crictx'}{\nf})$ we conclude.
\item $\tm = \strongabs{}{\nf}{\neu}$:
  Similar to the previous case, applying the \ih on the judgment
  $\judg{\tctx}{\neu}{\typ\nn}$.
\item \colorand{$\tm = \projipn{\neu}$}:
  Note that $\judg{\tctx}{\projipn{\neu}}{\ev}$
  must be derived from either of the rules $\Eandp$ or $\Eorn$,
  from a judgment of the form
  $\judg{\tctx}{\neu}{(\typ\land\typtwo)\pp}$
  or of the form
  $\judg{\tctx}{\neu}{(\typ\lor\typtwo)\nn}$.
  In any case, the type of $\neu$ is strong so we may apply the \ih
  to conclude that $\neu = \crictx'\ctxof{\clasappn{\var}{\tmfive}}$,
  where $\tmfive$ is a capsule.
  Then $\tm = \projipn{\crictx'\ctxof{\clasappn{\var}{\tmfive}}}$
  (where the signs of the projection and the weak elimination do
  not necessarily match),
  so taking $\crictx := \projipn{\crictx'}$ we conclude.
\item \colorand{$\tm = \casepn{\neu}{\var}{\nf_1}{\vartwo}{\nf_2}$}:
  Similar to the previous case, noting that
  $\judg{\tctx}{\casepn{\neu}{\var}{\nf_1}{\vartwo}{\nf_2}}{\ev}$
  must be derived from either of the rules $\Eorp$ or $\Eandn$,
  from a judgment of the form
  $\judg{\tctx}{\neu}{(\typ\lor\typtwo)\pp}$
  or of the form
  $\judg{\tctx}{\neu}{(\typ\land\typtwo)\nn}$.
\item \colorimp{$\tm = \appn{\neu}{\nf}$}:
  Similar to the previous case, noting that
  $\judg{\tctx}{\appn{\neu}{\nf}}{\ev}$
  must be derived from either of the rules $\Eimpp$ or $\Ecoimpn$,
  from a judgment of the form
  $\judg{\tctx}{\neu}{(\typ\imp\typtwo)\pp}$
  or of the form
  $\judg{\tctx}{\neu}{(\typ\coimp\typtwo)\nn}$.
\item \colorimp{$\tm = \colampn{\neu}{\var}{\vartwo}{\nf}$}:
  Similar to the previous case, noting that
  must be derived from either of the rules $\Ecoimpp$ or $\Eimpn$,
  from a judgment of the form
  $\judg{\tctx}{\neu}{(\typ\coimp\typtwo)\pp}$
  or of the form
  $\judg{\tctx}{\neu}{(\typ\imp\typtwo)\nn}$.
\item \colorneg{$\tm = \negepn{\neu}$}:
  Similar to the previous case, noting that
  $\judg{\tctx}{\negepn{\neu}}{\ev}$
  must be derived from either of the rules $\Enotp$ or $\Enotn$,
  from a judgment of the form
  $\judg{\tctx}{\neu}{(\neg\typ)\pp}$
  or of the form
  $\judg{\tctx}{\neu}{(\neg\typ)\nn}$.
\item \colorall{$\tm = \apptpn{\neu}{\typ}$}:
  Similar to the previous case, noting that
  $\judg{\tctx}{\apptpn{\neu}{\typ}}{\ev}$
  must be derived from either of the rules $\Eallp$ or $\Eexn$,
  from a judgment of the form
  $\judg{\tctx}{\neu}{(\all{\btyp}{\typtwo})\pp}$
  with $\ev = \typtwo\PP\sub{\var}{\typ}$,
  or from a judgment of the form
  $\judg{\tctx}{\neu}{(\ex{\btyp}{\typtwo})\pp}$
  with $\ev = \typtwo\NN\sub{\var}{\typ}$.
\item \colorall{$\tm = \optpn{\neu}{\var}{\btyp}{\nf}$}:
  Similar to the previous case, noting that
  $\judg{\tctx}{\optpn{\neu}{\var}{\btyp}{\nf}}{\ev}$
  must be derived from either of the rules $\Eexp$ or $\Ealln$,
  from a judgment of the form
  $\judg{\tctx}{\neu}{(\ex{\btyp}{\typ})\pp}$
  or of the form
  $\judg{\tctx}{\neu}{(\all{\btyp}{\typ})\nn}$.
\item $\tm = \clasappn{\neu}{\nf}$:
  Note that $\judg{\tctx}{\clasappn{\neu}{\nf}}{\ev}$
  must be derived from either of the rules $\Ecp$ or $\Ecn$,
  from a judgment of the form $\judg{\tctx}{\neu}{\typ\PP}$
  or of the form $\judg{\tctx}{\neu}{\typ\NN}$.
  Since the type of $\neu$ is weak,
  by \ih we have that $\neu$ is either a variable or of the form
  $\neu = \crictx'\ctxof{\clasappn{\var}{\tmfive}}$,
  where $\tmfive$ is a capsule.
  We consider these two subcases:
  \begin{enumerate}
  \item
    If $\neu$ is a variable, $\neu = \var$, 
    note first that, since $\judg{\tctx}{\clasappn{\neu}{\nf}}{\ev}$
    is derived from either of the rules $\Ecp$ or $\Ecn$,
    the type of $\nf$ must be weak.
    We consider two subcases, depending on whether
    $\nf$ is of the form $\nf = \claslamnp{\var}{\nf'}$ or not:
    \begin{enumerate}
    \item
      If $\nf = \claslamnp{\var}{\nf'}$,
      then $\tm = \clasappn{\var}{\claslamnp{\var}{\nf'}}$,
      so taking $\crictx := \ctxhole$ we are done.
    \item
      If $\nf$ is not of the form $\claslamnp{\var}{\nf'}$
      then since $\nf$ is a normal term and its type is weak
      we know that $\nf$ must be neutral.
      Hence by \ih we know that $\nf$ must be either a variable
      or of the form $\nf = \crictx''\ctxof{\clasappn{\vartwo}{\tmfive'}}$,
      where $\tmfive''$ is a capsule. We consider these two subcases:
      \begin{enumerate}
      \item
        If $\nf$ is a variable, $\nf = \vartwo$,
        then $\tm = \clasappn{\var}{\vartwo}$
        so taking $\crictx := \ctxhole$ we are done.
      \item
        If $\nf = \crictx''\ctxof{\clasappn{\vartwo}{\tmfive'}}$,
        where $\tmfive'$ is a capsule,
        then $\tm = \clasappn{\var}{\crictx''\ctxof{\clasappn{\vartwo}{\tmfive'}}}$
        (where the signs of the two weak eliminations do not
        necessarily match),
        so taking $\crictx := \clasappn{\var}{\crictx''}$ we are done.
      \end{enumerate}
    \end{enumerate}
  \item
    If $\neu = \crictx'\ctxof{\clasappn{\var}{\tmfive}}$
    where $\tmfive$ is a capsule,
    then $\tm = \clasappn{\crictx'\ctxof{\clasappn{\var}{\tmfive}}}{\nf}$
    (where the signs of the two weak eliminations do not necessarily
    match),
    so taking $\crictx := \clasappn{\crictx'}{\nf}$
    we conclude.
  \end{enumerate}
\end{enumerate}
\end{proof}

\begin{theorem}[Canonicity]
\lthm{appendix:canonicity}
\quad
\begin{enumerate}
\item
  If $\judgPRK{}{\tm}{\ev}$,
  then $\tm$ reduces to a canonical normal form $\nf$
  such that $\judgPRK{}{\nf}{\ev}$.
\item
  If $\judgPRK{}{\tm}{\ev}$, where $\ev$ is weak,
  then a canonical term $\tm'$ can be effectively found
  such that $\judgPRK{}{\claslampn{(\var:\ev\OP)}{\tm'}}{\ev}$
\end{enumerate}
\end{theorem}
\begin{proof}
We prove each case:
\begin{enumerate}
\item
  Suppose that $\judgPRK{}{\tm}{\ev}$ holds.
  By strong normalization
  consider the normal form $\nf$ of $\tm$
  and by subject reduction~(\rthm{subject_reduction})
  note that $\judgPRK{}{\nf}{\ev}$.
  By the characterization
  of normal forms~\lprop{characterization_of_normal_terms}
  we know that $\nf$ is either canonical or a neutral term.
  If $\nf$ is canonical, we are done.

  It suffices to argue that $\nf$ cannot be a neutral term.
  Indeed, note that $\nf$ cannot be a variable,
  since the typing context is empty.
  Hence, by~\rlem{appendix:shape_of_neutral_terms},
  $\nf$ must be of the form $\crictxof{\clasappn{\var}{\tmfive}}$,
  where $\crictx$ is a critical context and $\tmfive$ is a capsule.
  But note that critical contexts do not bind variables,
  so $\crictxof{\clasappn{\var}{\tmfive}}$ has a free variable $\var$.
  Hence $\nf = \crictxof{\clasappn{\var}{\tmfive}}$ cannot be typable
  under the empty typing context.
  This contradicts the fact that $\judgPRK{}{\nf}{\ev}$.
\item
  Suppose that $\judg{}{\tm}{\ev}$, where $\ev$ is weak.
  We consider the case in which $\ev = \typ\PP$;
  if $\ev = \typ\NN$ the proof is similar changing the signs.

  By the first item of this lemma, note that $\tm$ reduces to a
  canonical normal form $\nf$ such that $\judgPRK{}{\nf}{\typ\PP}$.
  Since $\nf$ is canonical,
  $\nf = \claslamp{(\vartwo:\typ\NN)}{\nf'}$
  where $\nf'$ is a normal form and
  $\judgPRK{\vartwo:\typ\NN}{\nf'}{\typ\pp}$.
  To prove the statement of the lemma, we must show
  we can find a canonical term $\tm'$ such that
  $\judgPRK{\vartwo:\typ\NN}{\tm'}{\typ\pp}$.

  We claim, more in general, that if
  $\nf'$ is a normal form such that
  $\judgPRK{\vartwo_1:\typ\NN,\hdots,\vartwo_n:\typ\NN}{\nf'}{\typ\pp}$,
  then
  we can find a canonical term $\tm'$ such that
  $\judgPRK{\vartwo:\typ\NN}{\tm'}{\typ\pp}$.
  We proceed by induction on the size of $\nf'$.
  Suppose that
  $\judgPRK{\vartwo_1:\typ\NN,\hdots,\vartwo_n:\typ\NN}{\nf'}{\typ\PP}$.
  By the characterization
  of normal forms~\lprop{characterization_of_normal_terms},
  $\nf'$ is either canonical or a neutral term.
  We consider these two cases:
  \begin{enumerate}
  \item
    If $\nf'$ is canonical:
    take
    $\tm' = \nf\sub{\vartwo_1}{\vartwo}\hdots\sub{\vartwo_n}{\vartwo}$,
    which is again canonical,
    and note that
    $\judgPRK{\vartwo:\typ\NN}{\tm'}{\typ\pp}$,
    as required.
  \item
    If $\nf'$ is a neutral term:
    since $\nf'$ is of strong type,
    by~\rlem{appendix:shape_of_neutral_terms}
    we have that $\nf' = \crictxof{\clasappn{\vartwo'}{\tmfive}}$
    where $\tmfive$ is a capsule.
    Since critical contexts do not bind variables, 
    we know that $\vartwo' = \vartwo_i$ for some $i\in1..n$
    and since $\vartwo_i$ is of type $\typ\NN$
    we have
    $\judgPRK{\vartwo_1:\typ\NN,\hdots,\vartwo_n:\typ\NN}{\tmfive}{\typ\PP}$.
    Moreover, $\tmfive$ is a capsule,
    \ie either a variable or a weak introduction.
    Note that $\tmfive$ cannot be a variable, for all the variables
    in $\vartwo_1:\typ\NN,\hdots,\vartwo_n:\typ\NN$
    are of type $\typ\NN$, whereas $\tmfive$ is of type $\typ\PP$.
    Hence $\tmfive$ must be a weak introduction,
    so we know that it must be of the form
    $\tmfive = \claslamp{(\varthree:\typ\NN)}{\nf''}$.
    Note that
    $\judgPRK{\vartwo_1:\typ\NN,\hdots,\vartwo_n:\typ\NN,\varthree:\typ\NN}{\nf''}{\typ\pp}$,
    where $\nf''$ is a strict subterm of $\nf'$.
    Hence, by \ih, there exists a canonical term $\tm'$
    such that
    $\judgPRK{\vartwo:\typ\NN}{\tm'}{\typ\pp}$.
    This concludes the proof.
  \end{enumerate}
\end{enumerate}
\end{proof}

\end{document}